\documentclass[11pt,final,twoside,nofrench]{thlifl}

\usepackage{color}
\usepackage{pdflscape}
\usepackage{minitoc}

\usepackage{float}
\usepackage{framed}

\usepackage[table]{xcolor}
\definecolor{darkpurple}{rgb}{0.7,0,0.5}
\definecolor{darkgreen}{rgb}{0,0.4,0}
\definecolor{darkblue}{rgb}{0,0,0.5}
\definecolor{navyblue}{rgb}{0,0.4,0.8}

\newcommand{\starline}{
  \bigskip \bigskip \bigskip
  \centerline{$\star$ \quad \quad \quad $\star$ \quad \quad \quad $\star$}
  \bigskip \bigskip \bigskip}

\usepackage[hyperindex,
  pdfsubject={These de doctorat},
  pdfauthor={Antoine Thomas},
  pdftitle={Rearrangement Problems with duplicated genomic content.},
  pdfdisplaydoctitle,
  bookmarksnumbered,
  colorlinks,
  linkcolor=darkgreen,
  citecolor=darkpurple,
  urlcolor=darkblue,
]{hyperref}

\usepackage{graphicx}
\usepackage{color}
\usepackage{xspace}
\usepackage{makeidx}
\makeindex
\usepackage{amsfonts}
\usepackage{amsmath}
\usepackage{amssymb}
\usepackage{rotating}
\usepackage{multirow}
\usepackage{array}
\usepackage[utf8]{inputenc}
\usepackage{graphicx} 
\usepackage{tikz}
\usepackage{color}
\usepackage{colortbl}
\usepackage{fixltx2e}
\usepackage{algorithm}
\usepackage{algorithmic}

\usetikzlibrary{positioning}
\usetikzlibrary{decorations.pathreplacing}

\definecolor{purple}{rgb}{0.7, 0, 1}
\definecolor{lightgray}{gray}{0.7}
\definecolor{Y}{RGB}{0,0,0}
\definecolor{B}{RGB}{128,128,128}
\definecolor{R}{RGB}{64,64,64}
\definecolor{P}{RGB}{192,192,192}
\definecolor{G}{RGB}{255,255,255}
\newcolumntype{G}{>{\columncolor{lightgray}}}

\newcommand{\qed}{\ensuremath{\blacksquare}}
\newcommand{\fst}[1]{ \ensuremath{#1} }
\newcommand{\snd}[1]{ \ensuremath{\overline{#1}} }
\newcommand{\msnd}[1]{ \ensuremath{{-\overline{#1}}} }
\newcommand{\mfst}[1]{ \ensuremath{{- #1}} }

\newcommand{\breakpoint}{ \textsubscript{$\textcolor{black}{\blacktriangle}$} }

\newcommand\aff[2]{\ensuremath{(\fst{#1}~\fst{#2})}}
\newcommand\asf[2]{\ensuremath{(\snd{#1}~\fst{#2})}}
\newcommand\afs[2]{\ensuremath{(\fst{#1}~\snd{#2})}}
\newcommand\ass[2]{\ensuremath{(\snd{#1}~\snd{#2})}}

\newcommand\mss[2]{\ensuremath{(\msnd{#1}~\msnd{#2})}}

\newcommand\oiff[2]{\ensuremath{]\fst{#1}~;~\fst{#2}[}}

\newcommand\oifs[2]{\ensuremath{]\fst{#1}~;~\snd{#2}[}}
\newcommand\oiss[2]{\ensuremath{]\snd{#1}~;~\snd{#2}[}}

\newcommand\cifs[2]{\ensuremath{[\fst{#1}~;~\snd{#2}]}}

\renewcommand{\NG}{\ensuremath{\mbox{\texttt{NG}}}}

\def\BI{\ensuremath{\mbox{BI}}}

\def\DCJ{\ensuremath{\mbox{\sl DCJ}}}

\newcommand{\DA}{\ensuremath{\mbox{\texttt{DA}}}}

\newtheorem{proposition}{Proposition}
\newtheorem{property}{Property}

\newtheorem{proof}{Proof}
\newtheorem{theorem}{Theorem}
\newtheorem{lemma}{Lemma}
\newtheorem{definition}{Definition}
\newtheorem{corollary}{Corollary}
\newtheorem{remark}{Remark}
\newtheorem{conjecture}{Conjecture}

\usepackage[footnotesize,bf]{caption}
\graphicspath{{./fig/}}

\author{xxx}
\title{xxx}

\begin{document}
\ThesisTitle{\huge Problèmes de réarrangement avec marqueurs génomiques dupliqués
\\[3ex]
\huge \it
Rearrangement Problems with \\ duplicated genomic content}
\ThesisAuthor{\href{mailto:antoine.thomas@gmx.com}{Antoine \textsc{Thomas}}}
\ThesisDate{18 juillet 2014}
\NewJuryCategory{Supervisor}{\textit{Supervisor :}}{\textit{Supervisor :}}
\Rapporteurs = {
  \href{http://albuquerque.bioinformatics.uottawa.ca/}{David \textsc{Sankoff}}, Canada Research Chair & {\small University of Ottawa, ON, Canada}\\
    \href{http://www.sopheetsa.org}{Sophia \textsc{Yancopoulos}}, Science Writer, \it{co-rapporteur} & {\small Feinstein Inst. for Med. Research, USA}\\
  \href{http://lbbe.univ-lyon1.fr/-Tannier-Eric-.html}{Eric \textsc{Tannier}}, CR Inria
  & {\small INRIA Grenoble, LBBE, Univ. Lyon 1}\\  
}\Examinateurs = {
  \href{http://researchers.lille.inria.fr/niehren/}{Joachim \textsc{Niehren}}, DR Inria, {\it président du jury} & {\small INRIA Lille, LIFL, Univ. Lille 1}\\
  \href{http://www.lifl.fr/~boulier/}{François \textsc{Boulier}}, Professeur, {\it directeur de thèse}
  & {\small LIFL, Univ. Lille 1}\\
}\mainmatter

\MakeThesisTitlePage

~ 

\newpage

\pdfbookmark[0]{Table of contents}{Table of contents}

\dominitoc
\setcounter{tocdepth}{2}
\tableofcontents

\chapter*{Avant-propos et remerciements}
\label{chap:acknowledgment}

Cette thèse de doctorat présente le travail que j'ai fourni de 2010 à fin 2012 au sein de l'équipe bonsai du LIFL.

Les lecteurs les plus attentifs remarqueront que la soutenance n'a pourtant pas eu lieu avant juillet 2014.

J'ai en effet été contraint d'interrompre mes activités de recherche au profit d'une activité bien moins amusante, une lutte contre un harcèlement perpétré impunément par les membres les plus haut placés de cette équipe de recherche. Dès lors que j'ai eu le malheur de dire non à leur appropriation honteuse de mon travail et à leur manque d'éthique en général, j'ai subi de nombreuses attaques dénigrant non seulement mon travail mais aussi ma personne au sein de toute l'équipe.

Le moins que je puisse dire c'est qu'il est malheureusement très difficile, en tant que simple doctorant, de défendre la valeur de son travail et de s'opposer à ces abus, dans un système où les médiateurs et supérieurs hiérarchiques ne sont autre que les collègues voire amis de nos bourreaux, et où nos pairs sont bien trop lâches pour faire preuve de solidarité.

Il n'est donc pas étonnant dans ce contexte, le sentiment d'impunité aidant, de voir la banalisation des abus de pouvoir au détriment de toute déontologie. A croire que pour certains il est très difficile de comprendre que \textit{docteur} ne signifie pas \textit{médecin}, et qu'ils n'ont donc aucune légitimité quand il s'agit d'établir des diagnostics d'ordre psychiatrique à l'égard de qui que ce soit...

Par ailleurs, discuter des problèmes médicaux de mes collègues, d'anciens membres de l'équipe ou de leurs entourages ne fait pas non plus partie de leurs fonctions et ils n'ont donc pas le droit d'en faire étalage public sans le consentement des intéressés, mais j'imagine que le concept de secret professionnel, ce doit être encore plus difficile à comprendre...

\starline

Tout d'abord merci à Mr Olivier Colot, directeur de l'école doctorale ED SPI 072, d'avoir été l'acteur principal de la résolution du problème en me permettant d'entamer une démarche de changement de directeur de thèse et d'équipe de recherche.

Merci à Mr Michel Petitot pour son temps et ses discussions variées toujours intéressantes, qui prennent le plus souvent une dimension très philosophique. Merci à lui d'incarner ce que devrait être selon moi un chercheur, une personne en éternelle quête de connaissances.

Merci aux examinateurs de cette thèse, Eric Tannier et Sophia Yancopoulos pour leur présence lors de la soutenance et leurs commentaires pertinents concernant ce travail. Cela peut sembler être peu de choses, mais croyez-moi, quand on sort de plus d'un an et demi de critiques infondées et d'injures, de combat contre des gens qui dénigrent ce travail pour mieux se l'approprier ensuite (par exemple en prétendant qu'ils ont été là pour le corriger et que tout le mérite leur revient donc), et qu'on tombe ensuite sur de vrais chercheurs, le contraste est pour le moins saisissant.

En somme, merci à tous ceux qui ont contribué au bon déroulement de cette fin de thèse, portant ainsi sur leurs épaules, bien malgré eux, le poids de l'incompétence de mes anciens encadrants.

Quant à ces derniers, parce que j'ai l'honnêteté qu'ils n'ont pas, je les remercie tout de même pour leur pathétique contribution, qui se résume très précisément à trois points : me donner quelques problèmes sur lesquels travailler par moi-même pendant qu'ils demandent cupidement de nouveaux résultats, tracer quelques figures pour mes articles, et bien-sûr y reformuler des paragraphes (il semblerait que cela soit une de leurs spécialités).

Mes remerciements les plus importants vont naturellement à ma famille, pour m'avoir inculqué des valeurs apparemment rares telles que l'intégrité, le courage et la persévérance, ainsi qu'un certain sens de la justice.

Merci également à mes amis d'avoir été là, c'est important de pouvoir rire en toutes circonstances, et aussi d'avoir du soutien.

~~

Un non-merci à tous ceux qui m'ont proposé, avec le sourire, d'accepter ces injustices, et de cautionner ma propre exploitation, en prétendant que c'était pour mon bien. Ils se reconnaîtront.

Un non-merci à tous ceux qui contribuent au vol de propriété intellectuelle organisé, à cette imposture qu'est la recherche dans certaines équipes où mélange des genres, harcèlement et impunité font un ménage à trois des plus malsains.

\chapter*{Introduction}

It is commonly known that life on earth is made of DNA, which carries genetic information, and that it evolves time after time as generations succeed each other.

DNA is a nucleic acid constituted by 4 nucleobases (its building blocks), which are guanine, adenine, thymine, and cytosine, respectively written G, A, T and C.

These nucleotides are folded in a double-helix structure where each base type is combined with its complementary base. G goes with C, A goes with T.

\begin{figure}[h!]
\centering
\includegraphics[scale=0.9]{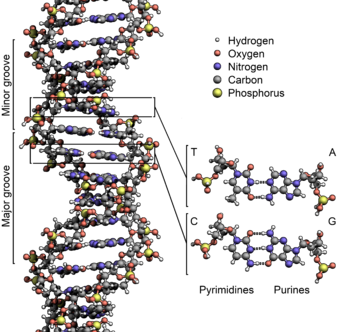}
\caption{The DNA molecule structure. Image courtesy of \url{http://en.wikipedia.org/wiki/DNA}}
\end{figure}

Biologists observed that evolution also occurred at large scale with genes or group of genes recombinations such as reversal of segments \cite{S21} \cite{PH88}. 
 Understanding rearrangement dynamics is a major issue in \emph{phylogenetics}, the study of evolutionary relationships between species or populations.

One aspect of phylogenetics is trying to reconstruct evolutionary trees.
In order to reach this goal, naturally it would be very helpful to be able to determine a relative \emph{evolutionary distance} between species (ie. knowing which species are closer with respect to which others), or to be able, given a group of genomes, to reconstruct the genome of a closest common ancestor to the group.

This is where genome rearrangements come into play. By defining a genome representation and a set of allowed operations (manipulations), one should be able to define a \emph{distance} (in the mathematical sense, see the following box if needed) between genomes.

\begin{framed}

A distance function is a function defined on a set X, which for any couple of elements from the set associates a real number, and satisfies 4 simple properties.

The distance d is a function defined as:

$d : X \times X \rightarrow R$

and $\forall x, y, z \in X$, we have:
\begin{enumerate}
\item $d(x,y) \geq 0$ (a distance cannot be negative)

\item $d(x,y) = 0   \iff  x = y$ (distance 0 means the compared elements are equal)

\item $d(x,y) = d(y,x)$     (symmetry)

\item $d(x,z) \leq d(x,y) + d(y,z)$  (triangle inequality).
\end{enumerate}

\end{framed}

~~

A rearrangement problem is always defined by a \emph{starting genome}, a \emph{goal genome}, and a set of allowed operations. Of course, since DNA sequences are a huge amount of data, algorithmic complexity of rearrangement problems is a major point of interest.

While DNA can ultimately be represented by the sequence of its nucleotides, comparative genomics usually deals with genomes at the scale of \emph{genomic markers} (ie. genes or group of genes). Each marker is labeled by an integer, with the sign representing its orientation on the genome.

Here is an example of rearrangement scenario between genomes $A = (\circ~~1~~2~-5~-7~-6~~3~~4~~8~~9~~10~~\circ)$ and $B = (\circ~~1~~2~~3~~4~~5~~6~~7~~8~~9~~10~~\circ)$, with \emph{reversals} as the only allowed operations.

\begin{center}
$(\circ~1~2~-5~\underline{-7~-6~3~4}~8~9~10~\circ)$

$\Downarrow$

$(\circ~1~2~\underline{\mfst{5}~\mfst{4}~\mfst{3}}~6~7~8~9~10~\circ)$

$\Downarrow$

$(\circ~1~2~3~4~5~6~7~8~9~10~\circ)$

\end{center}

In this example, two reversals suffice to go from the starting genome to the goal, and it cannot be done in one operation. Thus we say that the \emph{reversal distance} between genomes $A$ and $B$ is equal to 2, while the example itself is an \emph{optimal scenario}.

Computing the distance and computing an optimal scenario are two linked but different problems, with varying complexity.

\clearpage

Rearrangement problems were first introduced, for the reversal model, in \cite{S21} \cite{S41} and later rediscovered in \cite{PH88}. For more explanation, István Miklós wrote a detailed history of genome rearrangements, whose reading is much recommended\footnote{\url{http://www.renyi.hu/~miklosi/AlgorithmsOfBioinformatics.pdf}}. Another much recommended reading would be the book ``\textit{Combinatorics of Genome Rearrangement}" published by MIT Press, as it presents a mathematically oriented review of the field.

Rearrangement problems were first defined on \emph{non-duplicated} genomes, meaning each gene appears only once in both genomes, which made rearrangement problems in fact permutation sorting problems.

Among the pioneer results of the field there is a polynomial solution for sorting permutations by reversals \cite{HP95}, then generalized into the \emph{genomic distance}, a mix of reversals and translocations \cite{HP95bis}. The DCJ (double-cut and join), a further generalization of operation models was introduced in \cite{Yancopoulos05} which also allowed a better framework to study previous operation models.

Generally, it appears that rearrangement problems on non-duplicated genomes usually have polynomial complexity while duplicated markers induce NP-hardness. However, there are exceptions as the definition of a particular operation model or that of a goal genome might imply additional restrictions or release constraints, which can alter complexity.

For example, the \emph{genome halving} problem \cite{Mabrouk03}\cite{Mixtacki08} is a polynomial problem on duplicated genomes, which released the constraint of a particular goal genome, asking for a more general kind of configuration instead.

It is to note that while the distance and scenario are the usual questions when it comes to rearrangement problems, a substantious amount of work has been made to solve the question of solution spaces (the set of all optimal solutions) of rearrangement problems \cite{BS07} \cite{B09} \cite{BS10}.

~~

During the preparation of this Ph.D. thesis, I have been working on genome rearrangement problems with duplicated markers. Following the genome halving example, I studied \emph{other} hypotheses that could account for the presence of duplicated markers in genomes. I designed several rearrangement problems to account for these hypotheses and settle their algorithmic complexity.

I proved that under the hypothesis a single tandem duplication event is responsible for all observed duplicated markers, a parcimonious non-duplicated ancestor genome could be inferred in polynomial time. I provided $O(n^2)$ algorithms for the scenario, and a $O(n)$ computation for the distance, for two operation models, namely DCJ and Block Interchange.

I developed several problems based on multiple tandem reconstruction, and proved their NP-hardness.

I also studied breakpoint duplication, an intermediate model where any operation could lead to the duplication of markers at its endpoints. I proved NP-hardness of this problem and designed a fixed-parameter tractable (FPT) algorithm in the number of cycles present in the graph.

This work led to several publications \cite{Thomas11} \cite{Thomas13pre}

\noindent \cite{Thomas12} \cite{Thomas13}.

\clearpage

\starline

In the first chapter I will informally introduce what a rearrangement problem is, by revisiting a simple problem, \emph{Sorting by Block Interchange}, first solved in \cite{Christie96}.

I will then present a literature review, to provide a backdrop for my work, but also to give the reader a first intuition on what is hard and what is not when it comes to rearrangement problems.

A presentation of my own work will follow, sorted by duplication model and operation model, then a general conclusion will close the document.

\chapter[Preliminary game: sorting by block interchange]{Preliminary game: sorting by block interchange}
\label{chap:rearrangement}

Loosely based on my experience explaining what is my research to my laymen friends and family members, this rather informal chapter is meant as a playful introduction, whose purpose is to give the reader a rough understanding of the basic workings of rearrangement problems and my way of tackling them, even if they are not familiar at all with the field.

I will revisit ``Sorting permutations by Block Interchange" \cite{Christie96}, then quickly explain how it relates to genome rearrangement and my own work.

\section{Rules}

Let's assume we have a sequence of numbers we want to sort into numerical order.
\begin{center}
$(\circ~4~5~3~2~10~7~8~1~9~6~\circ)$
\end{center}

The first thing we have to do is define how we are allowed to alter the sequence in order to sort it.

For example, selecting two segments and swapping them is one way to alter a sequence. This operation is self-explanatorily called a \emph{block interchange} (BI). Is it possible to sort the sequence using only this operation?

\begin{center}
$(\circ~\fbox{4~5~3~2~10}~7~8~\fbox{1~9~6}~\circ)$

$\Downarrow$

$(\circ~1~\fbox{9}~6~7~8~4~5~3~\fbox{2}~10~\circ)$

$\Downarrow$

$(\circ~1~2~\fbox{6~7~8}~4~5~\fbox{3}~9~10~\circ)$

$\Downarrow$

$(\circ~1~2~3~4~5~6~7~8~9~10~\circ)$
\end{center}

Yes, it is possible to sort the sequence that way. Was it optimal? Yes, it took only three steps and it's not possible to sort it in two or less.

How do we know that? This is where it gets interesting.

\section{Breakpoints}

For better understanding how sorting works, let's review the very last BI operation, the one that restored the complete numerical order.

\begin{center}
$(\circ~1~2~\fbox{6~7~8}~4~5~\fbox{3}~9~10~\circ)$

$\Downarrow$

$(\circ~1~2~3~4~5~6~7~8~9~10~\circ)$
\end{center}

If given the task to sort such sequence using only one BI, anybody would correctly find the solution.

Anybody will instinctively notice that when selecting blocks, all starting/ending points are not equal (for example it would not make sense to cut between 1 and 2 given they are already in their correct relative order).

In fact it only makes sense to cut between numbers that \emph{break} the desired order. Those positions will be called \emph{breakpoints}. Note that when the sequence is fully sorted, no breakpoint remains.

Here is the sequence again, with breakpoints indicated by black triangles.

\begin{center}
$(\circ~1~2~\breakpoint~6~7~8~\breakpoint~4~5~\breakpoint~3~\breakpoint~9~10~\circ)$

$\Downarrow$

$(\circ~1~2~3~4~5~6~7~8~9~10~\circ)$
\end{center}

And of course those black triangles coincide with the extremities of the blocks for the sorting operation.

~

A breakpoint will never vanish by itself, so if a BI operation can act on at most 4 breakpoints at once, then in the best case, it will manage to solve the four of them and leave the rest of the sequence untouched.

\section{Example review}

With that new information in mind, if we have another look at the starting sequence, we can see that it contains 9 breakpoints.

\begin{center}
$(\circ~\breakpoint~4~5~\breakpoint~3~\breakpoint~2~\breakpoint~10~\breakpoint~7~8~\breakpoint~1~\breakpoint~9~\breakpoint~6~\breakpoint~\circ)$
\end{center}

Since a BI can only solve at most 4 of them at once, it is now obvious that it is not possible to sort the sequence in less than 3 steps.

We just solved the problem for \emph{this particular sequence}, but we want to be able to give an optimal solution for any sequence, therefore we cannot consider the problem is solved yet.

\section{General case}

When reviewing the full transformation sequence, you might notice that the first BI solved 3 breakpoints ($\breakpoint1$,  $6\breakpoint$ and $10\breakpoint$), the second solved 2 ($\breakpoint2$ and $9\breakpoint$), and only the last one solved the remaining 4 ($\breakpoint6$, $8\breakpoint$, $\breakpoint3$ and $3\breakpoint$).

It is easy to solve any 2 given breakpoints with a BI, but solving a 3rd and a 4th will happen only under certain conditions, and by a quick glance you can already tell that the special conditions are not obvious...

The real difficulty we are faced is to design a way to be able to understand those conditions and predict the best course of action, ie. being able of computing the minimum number of operations for \textbf{any} given sequence.

While the sequence embeds information about the number of operations needed to sort it, it embeds it in an implicit way. By changing the way we represent it, we might make the operations more explicit.

\subsection{Drawing a graph I}

We established earlier that breakpoints are crucial elements when it comes to sorting, and noted that what characterizes a sorted sequence is \emph{the absence of breakpoints}.

Our goal is to make things more explicit, therefore we will design a data structure around the sequence (which really is another way of representing the same information contained in a sequence), focused on that absence of breakpoints.

To achieve this I will draw a dot between numbers that are in consecutive order (see figure \ref{fig:sortedgraph}.

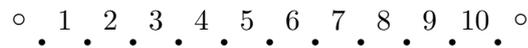
\begin{figure}[htbp]
    \centering
    
\begin{tikzpicture}[xscale=0.3,yscale=0.3]
    %\draw[style=help lines] (0,0) grid (38,-10);
    
    \foreach \i/\v in {0/\circ,2/1,4/2,6/3,8/4,10/5,12/6,14/7,16/8,18/9,20/10,22/\circ} {
      \draw (\i,0) node {$\v$};
    };

    \fontsize{4pt}{4pt}\selectfont
    %\tikzstyle{every node}=[circle,draw,fill]

    \draw (3,-1) node {$\bullet$};
        \draw (5,-1) node {$\bullet$};
    \draw (7,-1) node {$\bullet$};
        \draw (9,-1) node {$\bullet$};
    \draw (11,-1) node {$\bullet$};
        \draw (13,-1) node {$\bullet$};
    \draw (15,-1) node {$\bullet$};
        \draw (17,-1) node {$\bullet$};
    \draw (19,-1) node {$\bullet$};
        
    \draw 
    (1,-1) node {$\bullet$};
    \draw 
    (21,-1) node {$\bullet$};  
\end{tikzpicture}

\caption{Graph for a sorted sequence}
\label{fig:sortedgraph}
\end{figure}

Note that I also put dots before the number 1, and after the number 10. Since they are the ends of our sequence, they must be consecutive to the corresponding $\circ$ symbols.

\subsection{Drawing a graph II}

We designed a data structure around a sorted sequence, and will now generalize it to our shuffled sequence.

First the good cases:
``after 4 there is 5"
``after 7 there is 8"

Therefore we can add a dot between 4 and 5, as well as another one between 7 and 8 (see figure \ref{fig:3cycle}). These dots show the parts of the sequence that are in correct order.

So far it is not conceptually different from counting the breakpoints (we are counting \emph{adjacencies} instead), but in order to get more information we need something to do with the rest of the slots, where it is not possible to put dots.

Since it's about understanding the degree of shuffling, we will think in terms of what we \emph{should} have instead of what we actually have in the sequence.

~

\emph{``After 5 there should be 6"}.

So let's put a dot after 5 and join it with another dot before 6, to indicate they should be consecutive.

That reasoning alone would allow us to fill in the blanks (follow along on figure \ref{fig:3cycle}, as horizontal segments A, B and C are drawn from top to bottom):

\begin{figure}[h!]
    \centering
    
\begin{tikzpicture}[xscale=0.3,yscale=0.3]
    %\draw[style=help lines] (0,0) grid (38,-10);
    
    \foreach \i/\v in {0/\circ,2/4,4/5,6/3,8/2,10/10,12/7,14/8,16/1,18/9,20/6,22/\circ} {
      \draw (\i,0) node {$\v$};
    };

    \fontsize{4pt}{4pt}\selectfont
    %\tikzstyle{every node}=[circle,draw,fill]

    \draw (3,-1) node {$\bullet$};
    \draw (13,-1) node {$\bullet$};
\draw (11, -1.7) node {A};
    \draw 
    (5,-2) node {$\bullet$} --
    (19,-2) node {$\bullet$};

    \draw[densely dotted]
    (19,-2) node {$\bullet$} --
    (19,-3) node {$\bullet$};

\draw (14, -2.7) node {B};
    \draw
    (19,-3) node {$\bullet$} --
    (9,-3) node {$\bullet$};
    \draw[densely dotted] 
    (9,-3) node {$\bullet$} -- 
    (9,-4) node {$\bullet$};
  \draw (7, -3.7) node {C};
    \draw
    (9,-4) node {$\bullet$} --
    (5,-4) node {$\bullet$};
    \draw[densely dotted] 
    (5,-2) node {$\bullet$} -- 
    (5,-4) node {$\bullet$};
  
\end{tikzpicture}

\caption{Two 1-cycles and one 3-cycle so far.}
\label{fig:3cycle}
\end{figure}

After 5 $\rightarrow$ Before 6

We draw a segment joining a dot after 5 with another one before 6 (segment A).

In the sequence, we have $9 ~ 6$ so by arriving before 6 we are also in the position after 9. We have to put another dot there, accounting for the number 9.

After 9 $\rightarrow$ Before 10

The dot after 9 is joined with another one put before 10 (segment B), and thus we end up after 2.

After 2 $\rightarrow$ Before 3

We draw segment C and we are back to our starting position in 3 steps. We completed a \emph{cycle} of length 3 as seen in figure \ref{fig:3cycle}. Although unnecessary, in order to make the cycle stand out, dotted vertical lines are drawn between dots sharing a breakpoint in the sequence.
(Note: from now on a \emph{cycle of length n} will be called a \emph{n-cycle}).

In this context we might realize that single dots can be seen as cycles of length 1 (after 4 $\rightarrow$ before 5, we just join a dot with itself).

This is very good news, we just made the absence of breakpoint the elementary variant of a bigger concept, it means we extracted information of the same kind that is also more subtle.

The completed graph for the sequence is shown in figure \ref{fig:compgraph} (I recall that the dot before 1 is joined with one after the first $\circ$ symbol, and the dot after 10 is joined with one before the second $\circ$ symbol).

\begin{figure}[h!]
    \centering
    
\begin{tikzpicture}[xscale=0.3,yscale=0.3]
    %\draw[style=help lines] (0,0) grid (38,-10);
    
    \foreach \i/\v in {0/\circ,2/4,4/5,6/3,8/2,10/10,12/7,14/8,16/1,18/9,20/6,22/\circ} {
      \draw (\i,0) node {$\v$};
    };

    \fontsize{4pt}{4pt}\selectfont
    %\tikzstyle{every node}=[circle,draw,fill]

    \draw (3,-1) node {$\bullet$};
    \draw (13,-1) node {$\bullet$};

    \draw 
    (5,-2) node {$\bullet$} --
    (19,-2) node {$\bullet$};

    \draw[densely dotted]
    (19,-2) node {$\bullet$} --
    (19,-3) node {$\bullet$};
    \draw
    (19,-3) node {$\bullet$} --
    (9,-3) node {$\bullet$};
    \draw[densely dotted] 
    (9,-3) node {$\bullet$} -- 
    (9,-4) node {$\bullet$};
    \draw
    (9,-4) node {$\bullet$} --
    (5,-4) node {$\bullet$};
    \draw[densely dotted] 
    (5,-2) node {$\bullet$} -- 
    (5,-4) node {$\bullet$};
  
  \draw
    (1,-5) node {$\bullet$} --
        (7,-5) node {$\bullet$};
        \draw[densely dotted] 
        (7,-5) node {$\bullet$} -- 
        (7,-6) node {$\bullet$};
    \draw
      (7,-6) node {$\bullet$} --
          (17,-6) node {$\bullet$};
          \draw[densely dotted] 
          (17,-6) node {$\bullet$} -- 
          (17,-7) node {$\bullet$};
      \draw
        (17,-7) node {$\bullet$} --
            (15,-7) node {$\bullet$};
            \draw[densely dotted] 
            (15,-7) node {$\bullet$} -- 
            (15,-8) node {$\bullet$};
        \draw
          (15,-8) node {$\bullet$} --
              (1,-8) node {$\bullet$};
              \draw[densely dotted] 
              (1,-8) node {$\bullet$} -- 
              (1,-5) node {$\bullet$};
  
          \draw
            (11,-9) node {$\bullet$} --
                (21,-9) node {$\bullet$};
              \draw[densely dotted] 
                            (21,-9) node {$\bullet$} -- 
                            (21,-10) node {$\bullet$};
          \draw
            (11,-10) node {$\bullet$} --
                (21,-10) node {$\bullet$};
              \draw[densely dotted] 
                            (11,-9) node {$\bullet$} -- 
                            (11,-10) node {$\bullet$};
                   
\end{tikzpicture}

\caption{The completed graph contains two 1-cycles, one 3-cycle, one 4-cycle and one 2-cycle}
\label{fig:compgraph}
\end{figure}

\subsection{Using the graph}

Let's summarize once again:

\begin{itemize}
\item In the sorted sequence we should have eleven 1-cycles ($11 \times 1 = 11$).

\item In the shuffled sequence we have two 1-cycles, one 3-cycle, one 5-cycle and one 2-cycle ($2 \times 1 + 3 + 4 + 2 = 11$).
\end{itemize}

The total amount of edges seems invariant and equal to the number of elements in the sequence + 1.

We know the starting graph, we know the goal graph. The problem just shifted from \emph{sorting a sequence} to \emph{transforming a graph}. In order to solve this, naturally we need to study how a BI will modify the graph at each step.

It suffices to draw the graphs for each successive state of the sequence. By doing so, one will notice that each BI extracted two 1-cycles from other cycles:

\begin{itemize}

\item the first BI extracted two 1-cycles, one from the 4-cycle and another one from the 2-cycle (thus leaving as remainder a 3-cycle and \textbf{a 1-cycle}).

\item the second BI extracted two 1-cycles from the two 3-cycles (leaving two 2-cycles as remainder).

\item the last BI extracted two 1-cycles from the two 2-cycles (\textbf{leaving two 1-cycles as remainder})

\end{itemize}

In conclusion, sorting the sequence is the act of breaking down all cycles, extracting two 1-cycles with each BI.

Each bigger cycle will eventually give us an extra 1-cycle as remainder, which is like a half-bonus (a BI extracts two 1-cycles, so a 1-cycle remainder is like half a BI for free).

1-cycles already present in the graph also count as half-bonuses since they are operations we don't need to perform.

Therefore, there are as many half-bonuses as there are cycles, and we can derive an exact formula for the minimum number of BI required to sort any sequence $G$:

\begin{center}
$d_{BI}(G) = \frac{n + 1 - C}{2}$
\end{center}

n is the number of elements in the sequence $G$ (I recall we are working with $n+1$ edges in total), $C$ is the number of cycles in the graph (and the formula is divided by 2 because a BI extracts two 1-cycles).

\section{Genome rearrangements}

If we label each gene or group of genes (we will use the term \emph{marker}) with integers, then genomes can be seen as sequences of integers, and in this context, changes they undergo during evolution can be seen as operations on integer sequences.

Sorting a sequence into numerical order also allows us to study the distance between any two given genomes, as it's a matter of relabeling the elements in both genomes so that the second one is in numerical order, as illustrated in figure \ref{fig:rewriting}.

\begin{figure}[h!]
\begin{center}
\begin{tabular}{c | c}
$(\circ~2~8~5~3~9~7~6~4~10~1~\circ)$ & $(\circ~\textcolor{blue}{4~5~3~2~10~7~8~1~9~6}~\circ)$ \\

$\Downarrow$ ? & $\Downarrow$ ? \\

$(\circ~4~3~5~2~8~1~7~6~10~9~\circ)$ & $(\circ~\textcolor{blue}{1~2~3~4~5~6~7~8~9~10}~\circ)$
\end{tabular}

~

~

4$\rightarrow$\textcolor{blue}{1};
3$\rightarrow$\textcolor{blue}{2};
5$\rightarrow$\textcolor{blue}{3};
2$\rightarrow$\textcolor{blue}{4};
8$\rightarrow$\textcolor{blue}{5};
1$\rightarrow$\textcolor{blue}{6};
7$\rightarrow$\textcolor{blue}{7};
6$\rightarrow$\textcolor{blue}{8};
10$\rightarrow$\textcolor{blue}{9};
9$\rightarrow$\textcolor{blue}{10}
\end{center}

\caption{Rewriting the labels}
\label{fig:rewriting}
\end{figure}

This is, of course, considering that each marker appears \textbf{exactly once} in both genomes. When markers appear multiple times (when there are \emph{replicated markers}), rearrangement problems usually become harder. My work was to specifically design and study such harder rearrangement problems.

\section{Closing words and a bit of philosophy}

This intro is a good example of the general line of thought I used in my contributions.

Theoretical computer science (and more generally math) is an art of shapeshifting: it teaches us that everything is a model, nothing but representation of information, that could be rewritten in countless other forms. Solving an equation, or proving a theorem, can be seen as a matter of rewriting information into another form that will make the answer more explicit.

Because optimally sorting using a set of operations is always a process full of subtleties that need to be explicited, in genome rearrangement, we usually use data structures as a mean of rewriting information, as a tool of elegancy.

While it would have been possible to find a formula and proving it without the need for a graph structure, it wouldn't have been so simple, it wouldn't have been as useful in terms of comprehension (I think \emph{sorting is the act of breaking down cycles} is a result that suddenly allows a complete understanding, getting rid of any complex implicit behavior), and of course the formula wouldn't have been so simple either.

The ``hard" part I did not include here in order to let this section remain a \emph{playful} introduction is mathematically proving that there always exist a BI that will successfully extract two 1-cycles from the graph (it could be proved with another data structure meant to simplify the proof, in similar fashion to the proof I give for ``single tandem-halving by BI" later in the present document).

I'd like to give credit to Anne Bergeron as I am using her signature style of graph, drawing dots under the sequence rather than the traditional vertices and edges (I've always found it was the best way to keep the cycles apparent enough while displaying the sequence at the same time, so thanks for her idea).

\chapter[State of the art]{State of the art}

The goal isn't to be exhaustive, but rather to provide a little bit of context to better understand my contribution, what was already done, what was being done and what wasn't done yet when I worked on genome rearrangements.

\section{Notations (I) - Genome, markers, adjacencies, extremities, breakpoints}
\label{sec:notations}

Rearrangement problems deal with successive transformation of genomes, at the scale of \emph{genomic markers} (ie. genes or groups of genes).

It means in our models a genome consists of linear or circular chromosomes that are composed of
genomic markers. Markers are represented by signed integers such that the sign
indicates the orientations of markers in chromosomes. As there are only 2 possible orientations, naturally we have ${-{-x}}=
x$. A linear chromosome is represented by an ordered sequence of signed integers
surrounded by the unsigned marker $\circ$  at each end indicating the telomeres (chromosome extremities).
A circular chromosome is represented by a circularly ordered sequence of signed
integers. 
For example, $(1~~2~~{-3}) ~ (\circ~~4~~{-5}~~\circ)$ is a genome composed of
one circular and one linear chromosome.

An \emph{adjacency} in a genome is a pair of consecutive markers. Since a genome
can be read in two directions, the adjacencies $(x~~y)$ and $({-y}~~{-x})$ are
equivalent.

For example, the genome $(1~~2~-5) ~
(\circ~~{-3}~~4~~6~~\circ)$ has seven adjacencies, $(1~~2)$,
$(2~-5)$, $(-5~~1)$, $(\circ~~{-3})$, $({-3}~~4)$,
$(4~~6)$, and $(6~~\circ)$.

When an adjacency contains a $\circ$ marker, \textit{i.e.} a telomere, it is
called a \emph{telomeric adjacency}.

When needed, we will refer to marker extremities directly, indicating them using a dot. Thus, adjacency $(x~~y)$ concerns extremities $x\cdot$ and $\cdot y$.

\section{Distance and scenario}

Usually there are two problems to solve in rearrangements. Finding the \emph{minimal distance}, and finding a \emph{minimal scenario}.

The minimal distance is the minimum number of operations required to go from the input genome to a desired solution. In all models studied in this thesis this is a distance in the mathematical sense.

A minimal scenario is a sequence of operations transforming the input genome into a solution, whose length (number of operations) is the minimal distance.

While computing a scenario generally takes more time than computing the distance, in most operation models both problems belong to the same complexity class\footnote{see section \ref{sec:distscenar} for a quick proof}.

\section{Simple markers}

Even though I worked exclusively on genomes with duplicated content, I will start describing rearrangement problems where genomes have only one copy for each marker, as not only it provides context, but also helps understanding duplicated genomes problems better.

The problem is usually to find a way to transform a genome into the identity permutation (this allows to transform any genome into any other, as explained in chapter \ref{chap:rearrangement})

\subsection{Breakpoint distance}

\subsubsection*{Intro}

The breakpoint distance is a measure based on the number of breakpoints between two genomes, or between a genome and the identity permutation, as it was done in the beginning of chapter \ref{chap:rearrangement}. There are no operations associated with it and thus no scenario.

\subsubsection*{Example}

\begin{center}

\begin{tabular}{c @{\hskip 1cm}|@{\hskip 0.5cm} c}
\textbf{Breakpoint distance} & \textbf{Generalized breakpoint distance} \\
& \\
$(\circ~1~~2~\breakpoint-4~\breakpoint-5~\breakpoint~3~\breakpoint~\circ)$ & $(\circ~1~~2~\breakpoint-4~\breakpoint-5~\breakpoint-7~\textcolor{blue}{{}_{\triangle}}~\circ)$ \\
 & $(~\breakpoint6~\breakpoint~3~\breakpoint-8~\breakpoint~9~\breakpoint-10~)$\\
 & \\
$d = 4$ & $d = 8 + \textcolor{blue}{0.5} = 8.5$
\end{tabular}
\end{center}

Note: In generalized breakpoint distance, a telomeric breakpoint (indicated in blue) is worth 0.5. Also, in our example the second chromosome is circular (there are no $\circ$ markers), therefore the breakpoint before 6 is not telomeric, it is a breakpoint between markers -10 and 6.

\subsubsection*{History and references}

The breakpoint distance was first introduced as a lowerbound, a 2-approximation for the reversal distance \cite{KS93} (meaning the breakpoint distance is never more than twice the reversal distance). It was computed in $O(n)$, leading to a 2-approximate reversal scenario in $O(n^2)$.

\subsubsection*{Generalized breakpoint distance}

In its first form, the breakpoint distance was defined on permutations, ie. on unilinear genomes.

Much later, multichromosomal variants were proposed, notably one in \cite{TZS09} for which a lot of previously NP-hard problems became polynomial just by allowing circular chromosomes in genomes (and considering telomeric adjacencies/breakpoints are worth \textbf{half} their regular counterpart).

While several open questions remained in this paper, another researcher did an impressive extensive work answering them in \cite{K11}.

\subsubsection*{Outro}

While not attached to a particular operation, the breakpoint distance is useful as a lowerbound for more complex models.

The generalized variant of the distance shows some interesting results: while some usually NP-hard problems becoming polynomial is not a surprise as the model is much less accurate, a well-known polynomial problem became NP-hard in some cases\footnote{see genome halving in \cite{K11}}.

The breakpoint distance was also used in the context of comparing genomes with differing content albeit proven NP-hard \cite{BFC04}.

Along with the generalized breakpoint model, another breakpoint distance variant, the Single-cut-or-join (SCJ) was introduced  \cite{FM09} \cite{FM11} \cite{BFM13}, and further studied in other sources \cite{SCJ1} \cite{SCJ2}. This variant is similar to the other generalization in that it simplifies some NP-hard problems, but also has the benefit of having an associated operation model, and thus rearrangement scenarios.

\subsection{Sorting by reversals}

\subsubsection*{Intro}

The reversal (or \emph{inversion}) mechanism has been observed by biologists studying drosophila genomes \cite{SD36} \cite{SD38}.

A reversal acts on a segment of the genome and inverts both the order and the sign of markers within the segment.

\subsubsection*{Example}

\begin{center}
$(\circ~1~2~-5~\underline{-7~-6~3~4}~8~9~10~\circ)$

$\Downarrow$

$(\circ~1~2~\underline{\mfst{5}~\mfst{4}~\mfst{3}}~6~7~8~9~10~\circ)$

$\Downarrow$

$(\circ~1~2~3~4~5~6~7~8~9~10~\circ)$

\end{center}

\subsubsection*{History and references}

As aforementioned, study of the reversal mechanism predates bioinformatics, and it took half a century before it were reformulated in terms of computer science \cite{W82} \cite{S89}.

At first though, only approximation algorithms were given and it wasn't even known whether the problem was polynomial.
The answer came later with a $O(n^4)$ algorithm presented in \cite{HP95}, revisited and simplified \cite{KST97} \cite{B01} \cite{TS04} down to a subquadratic $O(n^{\frac{3}{2}}\sqrt{\log n})$ algorithm \cite{TBS07} making good use of an earlier data structure first introduced in \cite{KV05}.

About the distance itself, without the need for a rearrangement scenario, an efficient linear-time algorithm was published \cite{BMY01} as well as another very elegant one, still in $O(n)$ \cite{BMS-04}.

\subsubsection*{The unsigned case}

It was previously said that a reversal alters the order as well as the sign of markers. In the early days of sequencing, though, the sign couldn't be determined and thus the first algorithms considered unsigned permutations, meaning the sign was disregarded (another way to see it is that everything has a positive sign, and a reversal changes the order of markers but leaves the sign unchanged).

The problem of sorting unsigned permutations by reversals is a classical NP-hardness result in bioinformatics \cite{C97}.

It might sound like a counter-intuitive result if we consider the sign is an added constraint on the final configuration, but it is in fact more accurate to see it that way: it is a strong enough constraint to leave us no choice.

The sign gives us a valuable piece of information about the final position of a marker. Since our problems are about finding a \emph{minimal} distance, disregarding the sign is equivalent to having to find the sign affectation that will minimize the distance.

\subsubsection*{Outro}

Several generalizations and contraints have been proposed on this model such as perfect scenarios in polynomial time \cite{ST05} \cite{BBCP07} \cite{BBCP08} (scenarios that respect conservation criteria making them more likely from a biological point of view), or other studies dedicated to the solution space \cite{BS07} \cite{B09} just to name a few.

Another interesting extension was the \emph{HP distance}, using reversals to simulate different operations on a multichromosomal genome, namely reversals and translocations. It was introduced in \cite{HP-95} with a polynomial algorithm and a rather complex distance formula, further simplified \cite{OS03} \cite{JN07}, and finally elegantly tackled in \cite{BMS08} by making good use of a bigger abstraction framework, the DCJ operation, introduced in  \cite{Yancopoulos05}.

\subsection{Other operation models}

\subsubsection*{Reciprocal translocation}

We've briefly talked about the HP distance, using reversals on multichromosomal genomes to simulate other biological operations \cite{HP-95}. One of the HP distance operations is the \emph{reciprocal translocation} and it was also later studied as a standalone operation.

It consists in swapping telomeric extremities between chromosomes while reversing them.

\begin{center}
$(\circ~1~~2~~3~\underline{\textcolor{blue}{-8-7-6~\circ})~~(\textcolor{purple}{\circ~-5-4}}~~9~~10~\circ)$

$\Downarrow$

$(\circ~1~~2~~3~~\textcolor{purple}{4~~5~\circ})~(\textcolor{blue}{\circ~6~~7~~8}~~9~~10~\circ)$
\end{center}

The model was introduced in \cite{KR95} and a polynomial-time algorithm followed \cite{H95}.
An error in that algorithm was corrected and a $O(n^3)$ algorithm was given \cite{BMS05}.

Later, in \cite{OS06}, the efficient data structures developed for sorting by reversals were reused, leading to the same complexity for sorting by reciprocal translocations. The authors raised a mathematically interesting question: is it possible to linearly reduce one problem into the other?

To my knowledge this question remains open.

\subsubsection*{Block interchange}

The block interchange operation has been used as an illustration in chapter \ref{chap:rearrangement}. It is a swap of two blocks in the genome.

\begin{center}
$(\circ~\fbox{\textcolor{blue}{6~7~8~9~10}}~4~5~\fbox{\textcolor{purple}{1~2~3}}~\circ)$

$\Downarrow$

$(\circ~\textcolor{purple}{1~2~3}~4~5~\textcolor{blue}{6~7~8~9~10}~\circ)$
\end{center}

It was introduced and first solved in \cite{Christie96}. Enhancements were later made in \cite{Lin05} and \cite{FZ07} bringing the original $O(n^2)$ complexity down to a $O(n \log n)$ algorithm by using a permutation tree. The distance is linear.

\subsubsection*{Transposition}

A transposition is an operation where a single block from the genome is moved somewhere else.

\begin{center}
$(\circ~1~2~\breakpoint~5~6~7~8~9~\fbox{\textcolor{blue}{3~4}}~10~\circ)$

$\Downarrow$

$(\circ~1~2~\textcolor{blue}{3~4}~5~6~7~8~9~10~\circ)$
\end{center}

It can also be seen as a restriction on block interchange: the two blocks have to be adjacent.

\begin{center}
$(\circ~1~2~\fbox{5~6~7~8~9}~\fbox{\textcolor{blue}{3~4}}~10~\circ)$

$\Downarrow$

$(\circ~1~2~\textcolor{blue}{3~4}~5~6~7~8~9~10~\circ)$
\end{center}

First introduced in \cite{BP95}, sorting by transpositions remained an open problem for 15 years, during which publications emerged to give us approximation algorithms
\cite{BP98} \cite{HS04}.

Very recently an answer has been given in \cite{B10}, in the form of a NP-hardness proof.

It is interesting to see that such a small constraint on block interchange makes all the difference (and that there are indeed NP-hard rearrangement problems with simple markers).

\subsubsection*{Prefix reversal}

This problem is also known as \emph{pancake flipping} in its original formulation \cite{D75}.

Like with reversals, I will briefly talk both of the signed and unsigned cases.

A prefix reversal, as implied by its name, is a restriction on the reversal operation where the segment has to start at the beginning of the sequence.

In the unsigned case, once again, the sign is ignored.

\begin{center}

$\underline{(\circ~3~4~5}~2~1~6~7~8~9~10~\circ)$

$\Downarrow$

$\underline{(\circ~5~4~3~2~1}~6~7~8~9~10~\circ)$

$\Downarrow$

$(\circ~1~2~3~4~5~6~7~8~9~10~\circ)$

\end{center}

In the \emph{unsigned} case, the problem is NP-hard \cite{B12} (let us remind that it was already the case with regular reversals).

Trivia surrounding this problem aludes to the fact the famous founder of Microsoft, Bill Gates, co-signed an academic paper devoted to this problem \cite{GP79}. This is his only paper. In this article the authors also introduce the \emph{burnt pancakes} variant which is equivalent to the signed case.

To this date the \emph{signed} case remains an open problem. Polynomial subclasses of the problem have been shown as well as some bounds for general instances in \cite{L11}.

I have also briefly worked with Laurent Bulteau on a \emph{prefix DCJ} model that surprisingly allowed us to reestablish the best known bounds on the prefix reversal model in a much simpler way. Unfortunately that quick overview did not allow us to go further than this point.

\subsection{DCJ}

\subsubsection*{Intro}

A \emph{DCJ} (double-cut and join) operation on a genome $G$ cuts two different adjacencies in $G$ and glues pairs of the four exposed extremities in any possible way, forming two 
new adjacencies.
For example, the following DCJ cuts adjacencies $(1~~2)$ and $({-5}~~6)$ to produce $(1~~6)$ 
and $({-5}~~2)$.

\subsubsection*{Example}

\begin{center}

Extremities are bi-colored to indicate how they are glued afterwards (joining the same color).

$ (\circ~1~~2~~8~~9~\includegraphics[scale=0.08]{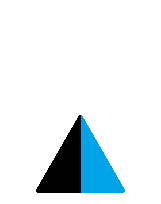}-6~-5~-4~-3~-7~\includegraphics[scale=0.08]{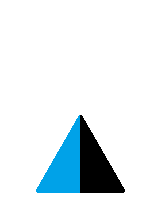}~10~\circ)$

$\Downarrow$

$ (-6~-5~-4~-3~\includegraphics[scale=0.08]{blua.png}~-7) ~ (\circ~~1~~2~\includegraphics[scale=0.08]{blua.png}~8~~9~~10~~\circ)$

$\Downarrow$

$(\circ~~1~2~3~4~5~6~7~8~9~10~~\circ)$

\end{center}

\subsubsection*{History and references}

The DCJ is more abstract than the other operations considered so far, and it allows more diversity in genome manipulation. Indeed, it is a genius generalization of all previous operations, first introduced in \cite{Yancopoulos05} and elegantly tackled in \cite{BMS06} giving us $O(n)$ algorithms for both distance and scenario in the signed case.

The unsigned case has been proven NP-hard \cite{C10}.

\subsubsection*{Outro}

As seen in the illustration, some DCJ operations can extract content into a new chromosome (\emph{excisions}), or merge chromosomes together (\emph{reintegrations}).

Because manipulating the chromosomal topology too much isn't very reasonable from a biological point of view, a restricted model was introduced at the same time \cite{Yancopoulos05}, where a chromosome excision would have to be followed by its immediate reintegration. The model has later been further studied and better algorithms were proposed \cite{Kovac10} \cite{Kovac}.

Another example of extension of the DCJ model is sorting between genomes with differing contents \cite{YF08} \cite{BWS10}, even in the restricted model \cite{Br13}.

Just like it was the case for the reversal model, the solution space of DCJ sorting has also been studied by the same author \cite{BS10}\footnote{...and ``independently" by one of my former advisors, although in a very rushed and inferior way...}.

The operation model itself was also further generalized in \cite{A08}, where the authors expand on the DCJ model, seeing it as the particular case for what they call a \emph{k-break} operation (which cuts the genome at $k$ positions then glues all exposed extremities in any possible way) where $k=2$.

\subsection{Phylogeny}

Originally, rearrangement problems were introduced as the elementary step of a more ambitious project: reconstructing phylogeny.

The principle is to use evolutionary distances as a way to measure relative distance between genomes, and rebuilding the phylogenetic tree (all common encestors and their relative positioning) from nothing but the set of present genomes.

The \emph{large phylogeny} problem aims at reconstructing the whole phylogenetic tree topology as well as all common ancestors from a given set of genomes.

\subsubsection{Small phylogeny and median}

\subsubsection*{Intro}

Small phylogeny is an ``easier" variant of the problem, where the tree topology is also given as an input \cite{SB98}.

Given a tree topology as well as the genomes present at its leaves, ancestor genomes should be inferred using a parsimony criterion: the sum of distances on all branches of the tree has to be minimal.

The \emph{median} problem is a further simplified variant of small phylogeny, the special case where there are only 3 genomes and one common ancestor to infer to them all.

NP-hardness of the median problem would naturally imply small phylogeny and large phylogeny NP-hardness as well.

\subsubsection*{History and references}

The median problem has been proven NP-hard under most rearrangement models (breakpoint \cite{PS98,B98,TZS09}, reversals \cite{C03}, and DCJ \cite{TZS09}), although in practice we have good branch and bound algorithms for computing optimal medians.

Under the generalized breakpoint model, when circular chromosomes are allowed, the problem becomes polynomial as demonstrated by a $O(n^3)$ algorithm \cite{TZS09}, then a better $O(n \log n)$ algorithm \cite{K11}.

Even though this was a very promising result, small phylogeny remains NP-hard under this model, even for as little as 4 species.

\subsubsection*{Outro}

The median problem was also studied under SCJ, the other generalized breakpoint distance.

 Under this model the median is computable in $O(n)$, even in the multilinear case, and small phylogeny is polynomial too \cite{FM09,FM11}. This result has been experimentally used with promising results \cite{BFM13}. Large phylogeny remains NP-hard, though.

Another very constrained reversal model (constrained in terms of allowed chromosomal topology as well as constrained specific reversals) allowed a polynomial answer, namely a $O(n)$ reversal median problem \cite{OAHS05}.

\clearpage

\section{Notations (II) - Duplicated genomes, double-adjacencies}
\label{sec:notations2}

A duplicated genome contains at most two occurences of each marker.  Two copies of a same marker in a genome are called paralogs. If a marker $x$ is present twice, one of the paralogs is represented by $\snd{x}$. By convention, $\snd{\snd{x}}= x$. 

\begin{definition}
A \emph{duplicated genome} is a genome in which a subset of the markers are duplicated. 
\end{definition}

For example,  
$(1~~2~~{-3}~~\msnd{2}) ~ (\circ~~4~~{-5}~~\snd{1}~~\snd{5}~~\circ)$ is a duplicated genome  where markers $1$, $2$, and  $5$ are duplicated. A \emph{non-duplicated genome} is a genome in which no marker is duplicated. 
A \emph{totally duplicated genome} is a duplicated genome in which all markers are duplicated. For example,  
$(1~~2~~\msnd{2}) ~ (\circ~~{-3}~~\snd{1}~~\snd{3}~~\circ)$ is a totally duplicated genome.

A \emph{double-adjacency} in a genome $G$ is an adjacency $\aff{a}{b}$
such that $\ass{a}{b}$ or $\mss{b}{a}$ is an adjacency of $G$ as
well. Note that a genome always has an even number of
double-adjacencies.  For example, the four double-adjacencies in the
following genome are indicated by $\diamond$:
$$G =
(\circ~~\fst{1}~~\snd{1}~~\fst{3}~~\fst{2}~\diamond~\fst{4}~\diamond~\fst{5}~~\fst{6}
~~\snd{6}~~\fst{7}~~\snd{3}~~\fst{8}~~\snd{2}~\diamond~\snd{4}~\diamond~\snd{5}~~\fst{
9}~~\snd{8}~~\snd{7}~~\snd{9}~~\circ )$$

\begin{definition}
A \emph{perfectly duplicated genome} is a totally duplicated genome such that
all adjacencies are double-adjacencies, none of them in the form $(x~\msnd{x})$.
\end{definition}

For example, the genome $(1~~{2}~~{3}~~4~~\snd{1}~~\snd{2}~~\snd{3}~~\snd{4})$
is a perfectly duplicated genome, while $(\circ~~1~~2~~\msnd{2}~~-1~~\circ)$ is not.

(Note: this definition is equivalent to the one from \cite{Mixtacki08}).

\section{Duplicated content}

As aforementioned, all previously cited work concern a sub-class of genomes: genomes for which each marker appears only once.

This is in fact far from biological reality where genes can appear in multiple copies, and naturally several problems taking this into account have been designed and solved, while simple markers problem could be seen as a first step. As we are about to see, the next step is far from being trivial, from a computational point of view, even under the strong assumption that each marker can only appear twice in a genome.

\subsection{Exemplar distance and matching models}

\subsubsection*{Intro}

As computer scientists, a natural way to tackle duplications is to find a reduction to a simple marker problem. If one were able to distinguish copies of a marker, then they could be labeled as 2 distinct ones and the problem remains the same as its simple markers variant. 

Seeking to keep the maximum number of genes in both genomes while giving them a one-to-one correspondance is known as the \emph{maximum matching model} \cite{BFC04}.

Another approach is to keep just one copy of each marker and see which choices allow the minimal distance: this is the \emph{exemplar model} \cite{S99}.

\subsubsection*{Example}
\begin{center}
\begin{tikzpicture}[xscale=0.3,yscale=0.3]
    %\draw[style=help lines] (0,0) grid (38,-10);
    
    \draw (7,3) node {\textbf{Exemplar}};

    \draw (27,3) node {\textbf{Maximum matching}};
    
    \draw 
        (17,4) node {} --
        (17,-6) node {};  
        
    \foreach \i/\v in {0/\circ,2/1,4/\textcolor{lightgray}{\snd{2}},6/\textcolor{lightgray}{-3},8/2,10/\snd{3},12/\textcolor{lightgray}{\snd{1}},14/\circ,20/\circ,22/1,24/\snd{2},26/-3,28/2,30/\textcolor{lightgray}{\snd{3}},32/\snd{1},34/\circ} {
      \draw (\i,0) node {$\v$};
    };

    \foreach \i/\v in {1/\circ,3/1,5/2,7/3,9/\textcolor{lightgray}{\snd{2}},11/\textcolor{lightgray}{\snd{1}},13/\circ,21/\circ,23/1,25/2,27/3,29/\snd{2},31/\snd{1},33/\circ} {
      \draw (\i,-5) node {$\v$};
    };
    \fontsize{4pt}{4pt}\selectfont
    %\tikzstyle{every node}=[circle,draw,fill]

    \draw 
    (2,-1) node {$\bullet$} --
    (3,-4) node {$\bullet$};  

    \draw 
    (8,-1) node {$\bullet$} --
    (5,-4) node {$\bullet$}; 

    \draw 
    (10,-1) node {$\bullet$} --
    (7,-4) node {$\bullet$}; 

    \draw 
    (22,-1) node {$\bullet$} --
    (23,-4) node {$\bullet$};  

    \draw 
    (24,-1) node {$\bullet$} --
    (25,-4) node {$\bullet$}; 

    \draw 
    (26,-1) node {$\bullet$} --
    (27,-4) node {$\bullet$}; 

    \draw 
    (28,-1) node {$\bullet$} --
    (29,-4) node {$\bullet$}; 

    \draw 
    (32,-1) node {$\bullet$} --
    (31,-4) node {$\bullet$}; 
        
\end{tikzpicture}
\end{center}

\subsubsection*{History and references}

The \emph{exemplar model} was first introduced under the breakpoint and the signed reversal distances \cite{S99}, along with branch-and-bound algorithms, while the exact complexity was not settled.

A NP-completeness proof for both distances was given in \cite{B00}.

The \emph{maximum matching} model was introduced in \cite{BFC04} as a mean to study rearrangements where reversals, insertions and deletions are allowed, but it was proven NP-hard.
Computing breakpoint distance or reversal distance under maximum matching is also NP-hard \cite{C05}.

In \cite{BCFRV07}, other comparison measures have also been proven to be NP-hard for both models.

A third matching model, the \emph{intermediate matching} was introduced as a bridge between both models \cite{A07}, where \emph{at least one} copy of each marker is kept.

Naturally, since both the exemplar and the maximum models are particular cases of the intermediate model, NP-hardness remains.

\subsubsection*{Outro}

Solved problems for simple markers become NP-hard as soon as duplications are introduced, since otherwise it would imply a polynomial-time matching.

Further restrictions of the exemplar distance have been studied, such as the \emph{zero exemplar distance problem}: Is the distance equal to zero? i.e. can two genomes be reduced to the same one via an exemplar matching?

Even this strong restriction is NP-hard for two monochromosomal genomes with as little as at most two occurences of each marker in both of them \cite{BFSV09}, or even disregarding gene sequences and only considering chromosomes as unordered sets of markers \cite{J11}.
It becomes polynomial only for a stronger restriction: each gene has to appear exactly once in one of the two genomes, and at least once in the other \cite{J11}.

\subsection{Genome Halving}

\subsubsection*{Intro}

In an attempt to avoid NP-hardness, the Genome Halving problem does not constrain the final configuration.

Given a duplicated genome $G$, we assume that all duplications are the result of a single \emph{whole genome duplication} (WGD) event.
The goal is to find the non-duplicated ancestor genome, ie. the state in which the genome was just before the WGD occured.

Naturally, it is done under the hypothesis of parsimony, meaning we want to minimize the distance between the genome and its ancestor.

\subsubsection*{Example}
\label{soahalving}
\begin{center}

$G = (\circ~1~~2~~\msnd{3}~\circ)~(\circ~\snd{1}~~3~\includegraphics[scale=0.08]{blua.png}~4~~\snd{5}~\includegraphics[scale=0.08]{blau.png}~\msnd{2}~\circ)~(~\snd{4}~5~)$

$\Downarrow$

$ (\circ~1~\includegraphics[scale=0.08]{blau.png}~2~~\msnd{3}~\includegraphics[scale=0.08]{blau.png}~\circ)~(~4~~\snd{5}~)~(\circ~\snd{1}~~3~~\msnd{2}~\circ)~(~\snd{4}~5~)$

$\Downarrow$

$G' = (\circ~1~~\snd{3}~~-2~\circ)~(~4~~\snd{5}~)~(\circ~\snd{1}~~3~~\msnd{2}~\circ)~(~\snd{4}~5~)$

\end{center}

Note that the biological scenario really is \emph{backwards}: chronologically, there was a non-duplicated genome $G'' = (~1~~\snd{3}~~-2~)~(~4~~\snd{5}~)$ which underwent a $WGD$ and thus became the perfectly duplicated genome $G'$, then other rearrangements happened so that we finally observe $G$.

\subsubsection*{History and references}

Genome Halving has first been studied under reversals
\cite{Mabrouk98} and HP distance \cite{Mabrouk03,AP07}, all with polynomial algorithms as answer. 

The DCJ variant was studied in \cite{Warren08}, and brilliantly in \cite{Mixtacki08} yielding a linear-time algorithm.

Surprising results arise with the generalized breakpoint distance model: while it remains polynomial for multichromosomal genomes with circular chromosomes allowed \cite{Tannier08}, it is NP-hard for monochromosomal and multilinear genomes \cite{K11}.

\subsubsection*{Outro}

The genome halving problem was a breakthrough, being the first polynomial duplicated rearrangement problem.

However, its polynomiality can be accounted for by the fact there is an exponential amount of optimal solutions, coupled with an exponential amount of possible scenarios to each of them.

Loosely speaking, the genome halving problem allows to sort half of the adjacencies so that they mimick the adjacencies formed by their paralogs. It is conceptually very close to a simple marker rearrangement problem.

Needless to say, so many different genomes being optimal also raises realism issues from a biological point of view.

\subsection{Other classical problems}

\subsubsection*{Intro}

The genome halving problem led to two other classic problems:

The \emph{genome aliquoting} is a generalization of genome halving, with more than two copies per marker.

The \emph{guided halving} is a problem designed to answer the realism issues of genome halving. By providing a \emph{reference} genome in addition to a duplicated genome, we now look for a perfectly duplicated genome that minimizes the sum of its distances towards the duplicated genome and the reference genome.

\subsubsection*{Genome aliquoting}

The genome aliquoting problem was first introduced in \cite{WS09} along with a heuristic algorithm as a first attempt to provide a solution. Further studies were made by the same authors using the generalized breakpoint distance as a 2-approximation for the DCJ distance \cite{WS11}.

The problem remains open as of today.

\subsubsection*{Guided halving}

Using an external reference genome to constrain solutions to the genome halving was first introduced in \cite{ZZS06} along with a heuristic algorithm, then proven NP-hard, even under breakpoint distance \cite{ZZS08}.

Under generalized breakpoint with circular chromosomes allowed, however, it becomes polynomial, first solved in $O(n^3)$ \cite{TZS09}, further enhanced to $O(n \log n)$ \cite{K11}.

On the whole this problem seems to display the same complexity as the median problem, which is not so surprising given conceptual similarities between the two problems\footnote{NP-hardness was proven by reduction from the median problem}.

\subsubsection*{Outro}

While the genome halving was a promising start, it was drifting away from biological reality, and both attempts at solving those issues are NP-hard.

Was the genome halving a breakthrough, or was it a deadend? Was its polynomiality the fleeting ray of light reinforcing the notion of darkness?

\section{Closing words}

In conclusion, rearrangement problems are defined on 3 main parameters:
\begin{itemize}
\item type of input genome(s)
\item allowed operations
\item desired configuration
\end{itemize}

A recent analysis of rearrangement with duplications has been published by experts on the field \cite{ES12}, which is a much recommended reading for anyone interested in the field.

\chapter[Rearrangements with duplicated markers]{Rearrangements with duplicated markers}

In this section I study rearrangement problems in the vein of genome halving.

The difference between genome halving and the problems I studied is that, rather than basing the problem on ``whole genome duplication" events, I studied segmental duplications instead, under various assumptions. I also studied another model where duplications occur on-the-fly as rearrangement operations happen.

For segmental duplications, I proved that a single tandem could be reconstructed in polynomial time. I provided $O(n^2)$ algorithms for the scenario, and an $O(n)$ computation for the distance, for both the DCJ and Block Interchange operation models.

I also proved NP-hardness of several problems based on multiple tandem reconstruction.

The breakpoint duplication model is also proved NP-hard, and I designed a FPT algorithm in the number of cycles present in a graph made for it.

These analyses led to several publications, in \cite{Thomas11} \cite{Thomas13pre} \cite{Thomas12} and \cite{Thomas13}.
 \section[Preliminary game II: genome halving]{Preliminary game II: genome halving}
\label{halvingrevisit}
It would benefit the reader and generally help the self-containment of this thesis to briefly revisit the classical genome halving by DCJ problem as it was handled in \cite{Mixtacki08}, because I use this problem as a starting point for the single tandem halving problem developed in section \ref{sec:wholetandem}, and because it provides good illustrations to what I am about to develop in section \ref{sec:dl}.

Similarly to what was done earlier in chapter \ref{chap:rearrangement}, I will focus on general workings rather than technical details. Since the notations have already been introduced previously, we shall dive even faster into the subject, considering the reader to already be familiar with genome rearrangements at this point.

Although what follows is similar to the analysis framework of chapter \ref{chap:rearrangement}, it is also slightly more complex as the genome halving problem presents some additional subtleties.

\subsection{Rules and example}

I recall the genome halving by DCJ problem, using terms defined in section \ref{sec:notations2}.

The input genome is a \emph{totally duplicated genome} (each marker appears exactly twice), and the goal is to find a closest (with respect to the DCJ distance) perfectly duplicated genome (each adjacency must be a double-adjacency, none of them them in the form $(x~\msnd{x})$).

Here is an example halving scenario.

\begin{center}

$(\circ~1~\includegraphics[scale=0.08]{blua.png}~\snd{4}~~\snd{1}~-2~\snd{2}~\includegraphics[scale=0.08]{blua.png}~3~~-4~~\snd{3}~\circ)$

$\Downarrow$

$ (\circ~1~\msnd{2}~2~\msnd{1}~\includegraphics[scale=0.08]{blau.png}~\msnd{4}~3~-4~\snd{3}~\includegraphics[scale=0.08]{blua.png}~\circ)$

$\Downarrow$

$(\circ~\includegraphics[scale=0.08]{blau.png}~\circ)~(\circ~1~~\msnd{2}~\includegraphics[scale=0.08]{blau.png}~2~\msnd{1}~\circ)~(~\msnd{4}~~3~-4~~\snd{3})$

$\Downarrow$

$(\circ~1~~\msnd{2}~\circ)~(\circ~2~\msnd{1}~\circ)~(~\msnd{4}~~3~-4~~\snd{3})$

\end{center}

Note that markers don't necessarily have to end in consecutive order, as it might not always be the shortest way to go.

Also, about the last performed operation in this example, it is a \emph{fission}. With our usual DCJ definition, we have to consider that the second breakpoint is in fact in an empty linear chromosome (which is a theoretical tool rather than an actual part of the genome). You can ignore it for now as it is much easier to explain in a graph, as we are about to see.

\subsection{General case}

\subsubsection{Drawing the \emph{natural graph}}

Because we are aiming at a closest genome satisfying a \emph{pattern} (namely, \emph{any} perfectly duplicated genome), the exact goal genome is not known, and thus we cannot build a graph relying on ``what we \emph{should} obtain" as it was done in chapter \ref{chap:rearrangement}.

However, we will still go in a generally similar line of reasoning as we are about to design a graph such that any perfectly duplicated genome is always corresponding to a particular configuration in terms of connected components.

By definition, a perfectly duplicated genome is a genome such that every adjacency is a double-adjacency. It follows that if we use edges to join paralogous extremities, a double adjacency can easily be identified, as shown in figure \ref{fig:halvpdg}. The graph defined this way is called the \emph{natural graph}.

\begin{figure}[h!]
\center
\begin{tikzpicture}[xscale=0.3,yscale=0.3]
    %\draw[style=help lines] (0,0) grid (38,-10);
    
    \foreach \i/\v in {0/(,2/\circ,4/1,6/2,8/3,10/\circ,12/),14/(,16/\circ,18/\snd{1},20/\snd{2},22/\snd{3},24/\circ,26/),28/(,30/4,32/5,34/\snd{4},36/\snd{5},38/)} {
      \draw (\i,0) node {$\v$};
    };

    \fontsize{4pt}{4pt}\selectfont
    %\tikzstyle{every node}=[circle,draw,fill]

    \draw 
    (3,-1) node {$\bullet$} --
    (17,-1) node {$\bullet$};

    \draw 
    (9,-6) node {$\bullet$} --
    (23,-6) node {$\bullet$};

    \draw 
    (5,-2) node {$\bullet$} --
    (19,-2) node {$\bullet$};

    \draw[densely dotted]
    (19,-2) node {$\bullet$} --
    (19,-3) node {$\bullet$};

    \draw
    (19,-3) node {$\bullet$} --
    (5,-3) node {$\bullet$};
    \draw[densely dotted] 
    (5,-3) node {$\bullet$} -- 
    (5,-2) node {$\bullet$};

    \draw 
    (7,-4) node {$\bullet$} --
    (21,-4) node {$\bullet$};

    \draw[densely dotted]
    (21,-4) node {$\bullet$} --
    (21,-5) node {$\bullet$};
    \draw
    (21,-5) node {$\bullet$} --
    (7,-5) node {$\bullet$};
    \draw[densely dotted] 
    (7,-5) node {$\bullet$} -- 
    (7,-4) node {$\bullet$};

    \draw 
    (29,-1) node {$\bullet$} --
    (33,-1) node {$\bullet$};

    \draw[densely dotted]
    (33,-1) node {$\bullet$} --
    (33,-2) node {$\bullet$};
    \draw
    (33,-2) node {$\bullet$} --
    (29,-2) node {$\bullet$};
    \draw[densely dotted] 
    (29,-2) node {$\bullet$} -- 
    (29,-1) node {$\bullet$};

    \draw 
    (31,-3) node {$\bullet$} --
    (35,-3) node {$\bullet$};

    \draw[densely dotted]
    (35,-3) node {$\bullet$} --
    (35,-4) node {$\bullet$};
    \draw
    (35,-4) node {$\bullet$} --
    (31,-4) node {$\bullet$};
    \draw[densely dotted] 
    (31,-4) node {$\bullet$} -- 
    (31,-3) node {$\bullet$};

\end{tikzpicture}

\caption{Natural graph of a perfectly duplicated genome. It consists of 1-paths and 2-cycles only.}
\label{fig:halvpdg}
\end{figure}
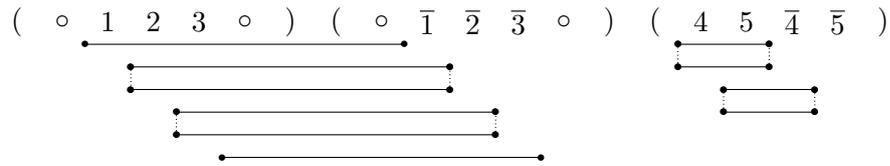

In the natural graph, a telomeric double-adjacency is necessarily a 1-path, while double-adjacencies concerning two distinct markers are 2-cycles.

Thus, any natural graph comprised of only 1-paths and 2-cycles is the natural graph of a perfectly duplicated genome.

\subsubsection{Using the graph I: DCJ on a graph}

Here are the successive natural graphs for the above scenario example. It will prove useful for studying how a DCJ can alter the graph.

\begin{center}
\center
\begin{tikzpicture}[xscale=0.3,yscale=0.3]
    %\draw[style=help lines] (0,0) grid (38,-10);
    
    \foreach \i/\v in {0/(,2/\circ,4/1,6/\snd{4},8/\snd{1},10/-2,12/\snd{2},14/3,16/-4,18/\snd{3},20/\circ,22/)} {
      \draw (\i,0) node {$\v$};
    };

    \fontsize{4pt}{4pt}\selectfont
    %\tikzstyle{every node}=[circle,draw,fill]

    \draw 
    (11,-1) node {$\bullet$};

\draw
(3,-5) node {$\bullet$} --
(7,-5) node {$\bullet$};
\draw[densely dotted]
(7,-5) node {$\bullet$} --
(7,-6) node {$\bullet$};
\draw
(7,-6) node {$\bullet$} --
(15,-6) node {$\bullet$};
\draw[densely dotted]
(15,-6) node {$\bullet$} --
(15,-7) node {$\bullet$};
\draw
(15,-7) node {$\bullet$} --
(19,-7) node {$\bullet$};

\draw
(5,-1) node {$\bullet$} --
(9,-1) node {$\bullet$};
\draw[densely dotted]
(9,-1) node {$\bullet$} --
(9,-2) node {$\bullet$};
\draw
(9,-2) node {$\bullet$} --
(13,-2) node {$\bullet$};
\draw[densely dotted]
(13,-2) node {$\bullet$} --
(13,-3) node {$\bullet$};
\draw
(13,-3) node {$\bullet$} --
(17,-3) node {$\bullet$};
\draw[densely dotted]
(17,-3) node {$\bullet$} --
(17,-4) node {$\bullet$};
\draw
(17,-4) node {$\bullet$} --
(5,-4) node {$\bullet$};
\draw[densely dotted]
(5,-4) node {$\bullet$} --
(5,-1) node {$\bullet$};

\end{tikzpicture}

Contents: one 3-path, one 1-cycle and one 4-cycle.

$\Downarrow$

\begin{tikzpicture}[xscale=0.3,yscale=0.3]
    %\draw[style=help lines] (0,0) grid (38,-10);
    
    \foreach \i/\v in {0/(,2/\circ,4/1,6/\msnd{2},8/2,10/\msnd{1},12/\msnd{4},14/3,16/-4,18/\snd{3},20/\circ,22/)} {
      \draw (\i,0) node {$\v$};
    };

    \fontsize{4pt}{4pt}\selectfont
    %\tikzstyle{every node}=[circle,draw,fill]

    \draw 
    (7,-1) node {$\bullet$};

\draw
(3,-5) node {$\bullet$} --
(11,-5) node {$\bullet$};
\draw[densely dotted]
(11,-5) node {$\bullet$} --
(11,-6) node {$\bullet$};
\draw
(11,-6) node {$\bullet$} --
(15,-6) node {$\bullet$};
\draw[densely dotted]
(15,-6) node {$\bullet$} --
(15,-7) node {$\bullet$};
\draw
(15,-7) node {$\bullet$} --
(19,-7) node {$\bullet$};

\draw
(5,-2) node {$\bullet$} --
(9,-2) node {$\bullet$};
\draw[densely dotted]
(9,-2) node {$\bullet$} --
(9,-3) node {$\bullet$};
\draw
(9,-3) node {$\bullet$} --
(5,-3) node {$\bullet$};
\draw[densely dotted]
(5,-3) node {$\bullet$} --
(5,-2) node {$\bullet$};

\draw
(13,-3) node {$\bullet$} --
(17,-3) node {$\bullet$};
\draw[densely dotted]
(17,-3) node {$\bullet$} --
(17,-4) node {$\bullet$};
\draw
(17,-4) node {$\bullet$} --
(13,-4) node {$\bullet$};
\draw[densely dotted]
(13,-4) node {$\bullet$} --
(13,-3) node {$\bullet$};

\end{tikzpicture}

A DCJ extracted a 2-cycle from the 4-cycle (remainder: \textbf{2-cycle})

Contents: one 3-path, one 1-cycle and two 2-cycles.

$\Downarrow$

\begin{tikzpicture}[xscale=0.3,yscale=0.3]
    %\draw[style=help lines] (0,0) grid (38,-10);
    
    \foreach \i/\v in {0/(,2/\circ,4/1,6/\msnd{2},8/2,10/\msnd{1},12/\circ,14/),16/(,18/\msnd{4},20/3,22/-4,24/\snd{3},26/)} {
      \draw (\i,0) node {$\v$};
    };

    \fontsize{4pt}{4pt}\selectfont
    %\tikzstyle{every node}=[circle,draw,fill]

    \draw 
    (7,-1) node {$\bullet$};

\draw
(3,-4) node {$\bullet$} --
(11,-4) node {$\bullet$};

\draw
(5,-2) node {$\bullet$} --
(9,-2) node {$\bullet$};
\draw[densely dotted]
(9,-2) node {$\bullet$} --
(9,-3) node {$\bullet$};
\draw
(9,-3) node {$\bullet$} --
(5,-3) node {$\bullet$};
\draw[densely dotted]
(5,-3) node {$\bullet$} --
(5,-2) node {$\bullet$};

\draw
(17,-2) node {$\bullet$} --
(21,-2) node {$\bullet$};
\draw[densely dotted]
(21,-2) node {$\bullet$} --
(21,-3) node {$\bullet$};
\draw
(21,-3) node {$\bullet$} --
(17,-3) node {$\bullet$};
\draw[densely dotted]
(17,-3) node {$\bullet$} --
(17,-2) node {$\bullet$};

\draw
(19,-4) node {$\bullet$} --
(23,-4) node {$\bullet$};
\draw[densely dotted]
(23,-4) node {$\bullet$} --
(23,-5) node {$\bullet$};
\draw
(23,-5) node {$\bullet$} --
(19,-5) node {$\bullet$};
\draw[densely dotted]
(19,-5) node {$\bullet$} --
(19,-4) node {$\bullet$};

\end{tikzpicture}

A DCJ extracted a 2-cycle from the 3-path (remainder: \textbf{1-path})

Contents: one 1-path, one 1-cycle, three 2-cycles.

$\Downarrow$

\begin{tikzpicture}[xscale=0.3,yscale=0.3]
    %\draw[style=help lines] (0,0) grid (38,-10);
    
     \foreach \i/\v in {0/(,2/\circ,4/1,6/\msnd{2},8/\circ,10/),12/(,14/\circ,16/2,18/\msnd{1},20/\circ,22/),24/(,26/\msnd{4},28/3,30/-4,32/\snd{3},34/)} {
          \draw (\i,0) node {$\v$};
        };
    
    \fontsize{4pt}{4pt}\selectfont
    %\tikzstyle{every node}=[circle,draw,fill]

    \draw 
    (3,-4) node {$\bullet$} --
    (19,-4) node {$\bullet$};

\draw
(7,-1) node {$\bullet$} --
(15,-1) node {$\bullet$};

\draw
(5,-2) node {$\bullet$} --
(17,-2) node {$\bullet$};
\draw[densely dotted]
(17,-2) node {$\bullet$} --
(17,-3) node {$\bullet$};
\draw
(17,-3) node {$\bullet$} --
(5,-3) node {$\bullet$};
\draw[densely dotted]
(5,-3) node {$\bullet$} --
(5,-2) node {$\bullet$};

\draw
(25,-1) node {$\bullet$} --
(29,-1) node {$\bullet$};
\draw[densely dotted]
(29,-1) node {$\bullet$} --
(29,-2) node {$\bullet$};
\draw
(29,-2) node {$\bullet$} --
(25,-2) node {$\bullet$};
\draw[densely dotted]
(25,-2) node {$\bullet$} --
(25,-1) node {$\bullet$};

\draw
(27,-3) node {$\bullet$} --
(31,-3) node {$\bullet$};
\draw[densely dotted]
(31,-3) node {$\bullet$} --
(31,-4) node {$\bullet$};
\draw
(31,-4) node {$\bullet$} --
(27,-4) node {$\bullet$};
\draw[densely dotted]
(27,-4) node {$\bullet$} --
(27,-3) node {$\bullet$};

\end{tikzpicture}

A DCJ transformed the 1-cycle into a 1-path (remainder: none)

Contents: two 1-paths, three 2-cycles.

\end{center}

\clearpage

As the nodes from the graph are corresponding to adjacencies of the genome, a DCJ operation, from the graph point-of-view, will cut connected components at 2 positions and glue the exposed extremities in any possible way.

If you have trouble visualizing it, imagine the edges are pieces of string maintained together by blu-tack, the nodes. A DCJ essentially rips two blu-tack pieces in half before joining the four resulting ripped pieces in another way, changing the way the pieces of string are connected together.

Back to our graph, it means that a DCJ can extract a cycle from any bigger component, or merge two components together, depending on whether the two breakpoints are shared by the same component or not.

Note that you also have the right to consider the two positions are degenerated (two breakpoints on the same node), as illustrated in the last step of our example. This allows a DCJ to transform a path into a cycle or vice versa, or to split a path into two smaller paths.

\subsubsection{Using the graph II: halving distance}

The analysis will be done in two steps, in the same vein as chapter \ref{chap:rearrangement}. First we will find what is invariant as it will serve as a bound for the formula, then we will inspect more closely how this bound relates to the exact distance by reasoning on the kind of connected components we are trying to reconstruct.

\subsubsection*{Invariant}

The number of edges in the graph is invariant and equal to the total number of markers. We will write it as $2n$, with $n$ the number of \emph{distinct} markers. 

We aim at reconstructing double-adjacencies, and a double-adjacency is a 2-cycle in the graph. Therefore, when bigger components are present in the graph, extracting 2-cycles from them are good operations. Since extracting a 2-cycle takes care of two edges from the graph, we can expect at most $\frac{2n}{2}=n$ operations, which gives us a first upperbound on the distance, $d \leq n$.

~~

Just like it was already the case in chapter \ref{chap:rearrangement}, computing the exact distance will be done by reasoning on the \textbf{remainders} left by our operations.

\subsubsection*{2-cycles}

A DCJ can always extract a 2-cycle from any bigger component, and bigger \emph{even} cycles will eventually give an additional 2-cycle as remainder (see step 1 of our example). That is like one operation for free, or in other terms, a bonus of one operation. 2-cycles already present in the graph count as bonuses as well, since they are also operations we don't need to perform.

$\Rightarrow$ If the natural graph contains $EC$ even cycles, then there are $EC$ 2-cycles which either are already present, or will eventually be formed as remainders of other 2-cycle extractions. Thus a better upperbound on the distance is $d \leq n - EC$

\subsubsection*{1-paths}

I said that double-adjacencies were 2-cycles in the graph. While this is true for adjacencies concerning two distinct markers, in the case of telomeric adjacencies, they are 1-paths. This means that 1-paths already present and 1-paths obtained as remainders also contribute to the distance formula.

A single DCJ could create two 1-paths at once (by splitting a 2-path). It follows that 1-paths already present only count as half a bonus each\footnote{note that this is also consistent with the generalized breakpoint distance where telomeric adjacencies count for 0.5}. Bigger odd-paths would also each eventually give a free 1-path as remainder.

$\Rightarrow$ If there are $OP$ odd paths in the graph, then there are $OP$ 1-paths that are either already present or will be formed as remainders of 2-cycle extractions. Since a 1-path is a bonus of half an operation, that makes $\frac{OP}{2}$ operations for free. It follows an even better upperbound would be $d \leq n - EC - \frac{OP}{2}$.

\subsubsection*{Distance}

There is one last subtlety we need to address in order to express the distance.

In the case where there is an odd number of odd paths, only an even number of them will effectively give bonuses, because the remaining one eventually induces a remaining 1-cycle somewhere else (indeed, the total number of edges is even), which forces us to perform an operation to transform it into a 1-path (see step 3 of our example). It is to note that such operation affects only one edge from the graph, instead of the usual two. This is basically equivalent to a penalty of half an operation, which cancels the bonus we got through the 1-path remainder. Therefore the exact contribution of odd paths is rather $\lfloor \frac{OP}{2} \rfloor$.

Note that odd cycles and even paths cannot give us bonuses since perfectly duplicated genomes graphs cannot contain any. Thus, there exists no other way to save operations, and we indeed have the exact distance.

$\Rightarrow$ In conclusion, the distance formula is equal to $d = n - EC - \lfloor \frac{OP}{2} \rfloor$.

\section{Meta-problems: general results}

\subsection{Dual layered vision of rearrangement problems}
\label{sec:dl}

In this section I describe a vision of rearrangement problems which helped me obtaining results and that I have not seen described elsewhere. This is essentially a generalization of the analysis framework I used in my revisits of \cite{Christie96} and \cite{Mixtacki08}.

I will first recall how the breakpoint distance can be used to establish bounds, as it is a prerequisite.

\subsubsection{Bounds}

Following the same logic used in chapter \ref{chap:rearrangement}, if an operation affects two breakpoints, then a single operation can reconstruct a \emph{maximum} of two adjacencies. This observation alone provides a lowerbound on the distance.

Conversely, by looking at the \emph{minimum} number of adjacencies that can always be reconstructed by an operation, an upperbound can be established. This upperbound is also useful to define approximation ratios.

For example, because a DCJ cuts the genome at two positions and is able to glue the exposed extremities in any desired way, it can always reconstruct one adjacency, and sometimes two, which implies that the breakpoint distance (number of adjacencies to be reconstructed) serves not only as an upperbound to the distance, but also as a 2-approximation\footnote{I recall the breakpoint distance was first introduced as a 2-approximation for the reversal distance in \cite{KS93}} since in the best case two adjacencies can be reconstructed at each step yielding a distance that is half the breakpoint distance.

The same reasoning also holds for more general operation models as long as they are defined on a fixed number of breakpoints, like the multi-break model from \cite{A08} (I recall that a k-break operation cuts the genome at $k$ positions and glue all exposed extremities in any possible way, which means a DCJ can be seen as a 2-break operation).

Indeed, the breakpoint distance serves as a k-approximation for k-break operations (in the best case, each k-break operation succesfully reconstructs k adjacencies, leading to a distance $k$ times smaller than the breakpoint distance).

\subsubsection{From bounds to distance}

We have seen that operations can be defined on a number of breakpoints, however it happens that different operations are defined on a same number of breakpoints yet act differently on them. To take this into account, one might say an operation is described on two levels:

\begin{itemize}
\item the lower level, related to the number of breakpoint it affects.
\item the higher level, related to the way these breakpoints are affected.
\end{itemize}

For example, a reversal is defined on \emph{two} breakpoints on a same linear or circular chromosome. The way these breakpoints are affected is by reversing the segment contained within, leaving the chromosomal structure intact.
A DCJ, however, while also defined on two breakpoints, doesn't constrain the breakpoint selection (they don't need to be on a same chromosome), and furthermore any possible recombination of exposed extremities is allowed, which can lead to a change in chromosomal structure (a linear chromosome can be transformed into a linear and a circular chromosome, for example).

Based on this description, bounds are related to this  lower level, while data structures are used to better study the higher level, allowing us to uncover what separates the distance from its bounds.

In this context, a distance formula can be seen as 

\begin{center}
$d =$ (number of expected operations) $-$ (sum of bonuses)
\end{center}

The \emph{number of expected operations} is a bound. It is the lower level contribution to the distance. It answers the question \emph{``how many operations are needed in the worst case, given how many adjacencies there are to be reconstructed and how many of them can always be reconstructed by an operation?"}.

Note that this number is an invariant. It only depends on the problem, not on the considered genome, and thus it remains constant during a rearrangement scenario.

For example, in chapter \ref{chap:rearrangement}, BI were the only allowed operations. For a genome of length $n$, there were $n+1$ adjacencies to be reconstructed in the worst case (telomeric adjacencies were not given a different weight in this model since the problem was defined on unilinear genomes only). A single BI can always reconstruct at least $2$ adjacencies. It follows that the number of expected operations in this case is $\frac{n+1}{2}$.

With my other example, on duplicated markers, in section \ref{halvingrevisit}, for a genome consisting of $n$ markers each appearing twice, there are at most $n$ adjacencies to be reconstructed (even though the genome length amounts to a total of $2n$ adjacencies, half of them can be left untouched, using the paralogs to copy them). A DCJ, which is the only allowed operation here, can always reconstruct one adjacency. It follows the number of expected operations in this case is $n$.

The sum of bonuses is what separate the bound from the actual distance. It depends on the genome, and it is the number of times we will be able to reconstruct more adjacencies in a fortuitous way during a scenario. It is usually more apparent on the graph (through \emph{remainders} of operations) as it depends on the exact transformation applied, and this is why it is the higher level contribution to the distance.

I will keep the same two examples to illustrate this part of the formula:

In chapter \ref{chap:rearrangement}, each sorting operation extracts two 1-cycles from bigger cycles. It follows that 1-cycles already present in the graph could be seen as the result of previously performed operations, and since an operation would create two of them, observing one is like observing the result of half an operation. In other words, we might say each 1-cycle already present is a bonus of half an operation. Bigger cycles will also each eventually give a 1-cycle as remainder, which is half a bonus too. 

It follows that each cycle is half a bonus, so the \emph{sum of bonuses} is $\frac{C}{2}$, with $C$ the number of cycles in the graph.

In conclusion, as the number of expected operations was $\frac{n+1}{2}$, the distance formula becomes $\frac{n+1-C}{2}$.

In section \ref{halvingrevisit}, the goal graph is consisting of 2-cycles and 1-paths. A DCJ would extract a 2-cycle, therefore 2-cycles already present are a bonus of one operation each. Bigger \emph{even} cycles eventually give a 2-cycle as remainder. It follows the number of even cycles, $EC$, contributes to the sum of bonuses.
 With 1-paths it's a bit more subtle as previously explained. Since a single DCJ could create two 1-paths at once (by splitting a 2-path), they can only count as half a bonus each. Through extraction of 2-cycles, bigger \emph{odd} paths will also each give a 1-path as remainder, eventually. It is also important to note that in the case where there is an odd number of odd paths, one of the paths would not contribute (refer to section \ref{halvingrevisit} for details). It follows that, with $OP$ being the number of odd paths, $\lfloor \frac{OP}{2} \rfloor$ also contributes to the sum of bonuses for genome halving.

In conclusion, as the number of expected operations was $n$, the distance formula becomes $n - (EC + \lfloor \frac{OP}{2} \rfloor )$.

In addition to helping to establish distance formulas, such vision also helps grasping a better understanding of scenarios construction, which might ultimately prove useful when studying solution spaces, for example. Indeed, any sorting operation must necessarily increase the sum of bonuses.

\subsubsection{On complexity}

The lower level contribution is directly related to the breakpoint distance, which is usually polynomial even for problems that become NP-hard for more elaborate operation models.

It follows that the NP-hard part of such problems lies in computing the sum of bonuses.

Separating the NP-hard part from the polynomial part of a rearrangement problem allows for an easier reduction, and might even lead to FPT algorithms if the sum of bonuses is independent from the genome length.

For example, later in the present document, namely section \ref{sec:dedoubling}, I study and prove NP-hardness of a rearrangement problem which would be a very good illustration of this principle. Obviously I cannot expand too much on it for now since the problem hasn't been defined yet, but I'll just say that the distance formula has the usual $n$ as number of expected operations, while the number of bonuses, $C_i$ is the direct result of a well-known NP-hard problem, allowing for a straightforward reduction. I will get back to it when the problem is defined: in section \ref{sec:dedoublebonus} I will show how to use the ideas from the present section to quickly establish the distance formula.

In conclusion, thinking in terms of ``what is the number of expected operations" and ``what are the bonuses" might prove useful not only to establish rearrangement distances, but also NP-hardness proofs. For this reason it has become one of my first approach when tackling a new rearrangement problem.

I will add that this dual layered view might also be seen as a mere first step. One might put additional layers in the description of the bonuses themselves as an attempt to better understand how a problem works, or to better express a NP-hard aspect. For example, staying on the genome halving example, I could state that 2-cycles and 1-paths already present in the graph are \emph{bonuses of order zero}, while bigger even cycles and odd paths are \emph{bonuses of higher order}, since they will create additional zero order bonuses at a later step in the scenario.

Interestingly enough, when using such further layering with rearrangement problems, we notice that \emph{(number of expected operations) - (sum of zero order bonuses)} is generally exactly the breakpoint distance.

\subsection{Scenario, distance and complexity class}
\label{sec:distscenar}

Note: As it is rather easy to prove, I don't think this result is new, and I admit I did not actively look for such result in bibliography (I would not know where to look as it is very general). However, since I have seen that question being brought up several times without an answer, I am including my quick proof just in case. I would also like to thank Eric Tannier for pointing out this is a result of self-reducibility.

Given an optimal scenario, computing the minimal distance is trivial. It follows that a polynomial scenario implies a polynomial distance. I prove the converse is true as well, and thus that most rearrangement problems are \emph{self-reducible} under reasonable assumptions on the operation model.

To make the proof easier to follow I will assume we are in the DCJ operation model, while generalization to other models will be discussed afterwards.

\begin{lemma}
\label{lem:possibleop}
On a given genome, there is a polynomial number of possible DCJ (with respect to the genome length).
\end{lemma}

\begin{proof}
A DCJ is defined on a \emph{fixed} number of breakpoints, $b=2$, independent of the genome length. It follows there are $O(n^b) = O(n^2)$ possible ways of placing them on the genome. Naturally, the fact there are two possible DCJ that can be performed for each breakpoint positionning doesn't matter as we still have $O(n^2)$ possible DCJ on a genome of length $n$\qed
\end{proof}

\begin{lemma}
\label{lem:polydist}
Optimal DCJ scenarios have a polynomial length (with respect to the genome length).
\end{lemma}

\begin{proof}
The breakpoint distance is an upperbound for the DCJ distance and it is a number that is polynomial with respect to the genome length.\qed
\end{proof}

\begin{theorem}
If the distance can be computed in polynomial time, then an optimal scenario can be as well.
\end{theorem}

\begin{proof}
If an operation decreases the distance, then it is optimal. An algorithm reconstructing an optimal scenario can always be built this way: At each step, we look for an optimal operation by trying all possible breakpoints positions, and compute the distance for the resulting genome. By lemma \ref{lem:possibleop} each step is polynomial iff the distance computation is polynomial, and by lemma \ref{lem:polydist}, there is a polynomial number of iterations. Thus, this algorithm is polynomial on the whole iff the distance computation is polynomial.\qed
\end{proof}

Naturally, this result can be generalized to other operation models, as long as 1) there is a polynomial number of possible operations at each step, whose resulting genomes can each be computed polynomially and 2) the distance itself is a polynomial number with respect to the genome length.

One can verify this holds for most common operation models such as reversals, translocations, block interchange, transpositions, and also k-breaks.

As a final note, I'll add that the resulting scenario computing algorithm is far from efficient and that its interest is most likely limited to this self-reducibility result and its implications.

\section{Model I: Breakpoint duplication}

This work has been published in \cite{Thomas11}.

Minor revisions have been made since, giving more information about the orientation cost for reversals, a proof has been rewritten, and vocabulary has been fixed (the algorithm I designed for this problem was a fixed-parameter tractable (FPT) algorithm, thanks to Laurent Bulteau for pointing out this fact).

Since I could develop the dual layered view in section \ref{sec:dl}, I also added a new section explaining the DCJ distance formula and complexity for this problem using this tool.

I proved the genome dedoubling problem is NP-hard for DCJ and reversals, and gave a FPT algorithm in the number of cycles to solve it under DCJ, and another one under reversals for a subclass of genomes.

\subsection{Biological motivation}

Gene duplication is an important source of variations in eukaryotes.
Recently, several studies have highlighted biological evidence for 
abundant segmental duplications that occur around breakpoints of 
rearrangement events in mammalians \cite{B04, D11}, and in Drosophila species group \cite{R07} \cite{M05} \cite{R05} \cite{M09}.

You might refer to these papers or to \cite{Thomas11} for more details.

\subsection{Model}

\subsubsection{Considered genomes}

The genomes considered in this section are duplicated and totally duplicated genomes, as defined in section \ref{sec:notations2}.

I also introduce a particular kind of duplicated genome, namely \emph{dedoubled genomes}.

\begin{definition}
A \emph{dedoubled genome} is a duplicated genome $G$ such that for any
duplicated marker $x$ in $G$,  either $(x~~\snd{x})$,  or  $(\snd{x}~~x)$ is an
adjacency of $G$.
\end{definition}

For example, $G = (\circ~-1~\msnd{1}~2~\circ)~(~4~\snd{4}~\snd{3}~3~)$ is a dedoubled genome with 3 duplicated markers.

\subsubsection{Considered operations}

I use two operation models in this section: DCJ, and reversals.

I recall that reversals are particular kinds of DCJ.

I also define new operations as follows:

A \emph{1-breakpoint-duplication DCJ} (1-BD-DCJ) operation on a genome $G$ is 
a rearrangement operation that alters two different adjacencies $(a~~b)$  
and $(c~~d)$ of $G$, by:
\begin{itemize}
\item first adding marker $\snd{a}$  at the appropriate position to produce 
segment $(a~~\snd{a}~~b)$,
\item then applying a DCJ operation that cuts adjacencies $(a~~\snd{a})$ and 
$(c~~d)$ to produce either  $(a~~d)$ and $(c~~\snd{a})$, or  $(a~~{-c})$ and 
$(\msnd{a}~~d)$. 
\end{itemize}

A \emph{2-breakpoint-duplication DCJ} (2-BD-DCJ) operation on a genome $G$ 
is a rearrangement operation that alters two different adjacencies 
$(a~~b)$  and $(c~~d)$ of $G$, by:
\begin{itemize}
\item first adding markers  $\snd{a}$ and  $\snd{c}$ at the appropriate 
positions to produce segments  $(a~~\snd{a}~~b)$ and $(c~~\snd{c}~~d)$,
\item then applying a DCJ operation that cuts adjacencies $(a~~\snd{a})$ 
and $(c~~\snd{c})$ to produce either $(a~~\snd{c})$ and $(c~~\snd{a})$, or  
$(a~~{-c})$ and $(\msnd{a}~~\snd{c})$.
\end{itemize}

In this context, a regular DCJ could also be regarded as a 0-BD-DCJ. I will include this possibility in the following general definition of a BD-DCJ, since regular DCJ can also be used in BD-DCJ scenarios.

\begin{definition}
A \emph{breakpoint-duplication DCJ} (BD-DCJ) operation on a genome $G$ is 
either a DCJ, a 1-BD-DCJ operation, or a 2-BD-DCJ operation.
\end{definition}

In the sequel, if some markers are duplicated by a BD-DCJ operation, they 
are indicated in bold font in the initial genome. 
For example, the following rearrangement is a 2-BD-DCJ operation that acts 
on adjacencies $({-2}~~{-1})$ and $(4~~{-3})$, and duplicates markers $2$ and 
$4$. The intermediate step resulting in the duplication of markers $2$ and $4$
is shown above the arrow.

$$ (1~\breakpoint~\mathbf{2}) ~ (\circ~~3~~\breakpoint ~\mathbf{-4}~~\circ) \overset{(1~\cdot~\snd{2}~\breakpoint~2) ~ (\circ~~3~\cdot~\msnd{4}~\breakpoint -4~~\circ)}{\rightarrow} (\circ~~3~\cdot~\msnd{4}~\cdot~2~~1~\cdot~\snd{2}~\cdot~{-4}~~\circ)$$

To summarize, a BD-DCJ operation consists of a \emph{first step} in which
zero, one or two markers are duplicated, followed by a \emph{second step} where
a  DCJ operation is applied. Similarly, we now define a 
\emph{breakpoint-duplication reversal} (BD-reversal) operation.

\begin{definition}
A \emph{breakpoint-duplication reversal} (BD-reversal) operation on a genome 
$G$ is a BD-DCJ operation such that the DCJ operation applied in the second 
step of the BD-DCJ operation is a reversal.
\end{definition}

For example, the following rearrangement is a BD-reversal that is a 
1-BD-DCJ operation that acts on adjacencies $(2~~{-1})$ and $({-3}~~4)$, 
and duplicates marker $2$.
$$ (\circ~~1~\breakpoint~\mathbf{-2}~~{-3}~\breakpoint~4~~\circ)\overset{(\circ~~1~\cdot~{-\snd{2}}~\breakpoint~{-2}~{-3}~\breakpoint~4~~\circ)}{\rightarrow} (\circ~~1~\cdot~{-\snd{2}}~\cdot~3~~2~\cdot~4~~\circ) $$

A \emph{BD-DCJ scenario} (resp. \emph{BD-reversal scenario}) between a 
non-duplicated genome $A$ and a duplicated genome $B$ is a sequence composed 
of BD-DCJ (resp. BD-reversal) operations allowing to transform $A$ into $B$.

\begin{definition}
Given a non-duplicated genome $A$ and a duplicated genome $B$, 
the \emph{BD-DCJ distance} (resp. \emph{BD-reversal distance}) 
between $A$ and $B$
is the minimal length of a  BD-DCJ (resp. BD-reversal)
scenario between $A$ and $B$.
\end{definition}

We now give an obvious, but useful property allowing to reduce a 
BD-DCJ scenario to a DCJ scenario.

\begin{proposition}
\label{BD-DCJtoDCJ}
Given  a non-duplicated genome $A$ and a duplicated genome $B$, 
for any a BD-DCJ (resp. BD-reversal) scenario between $A$ and $B$, 
there exists a DCJ (resp. reversal) scenario of same length between a 
dedoubled genome $D$ and $B$ such that the reduction of $D$ is $A$ ($D^R = A$).
\end{proposition}
\begin{proof} 

Let $S$ be a BD-DCJ (resp. BD-reversal) scenario between $A$ and $B$.
$D$ is the genome obtained from $A$, by adding, for any marker $x$ 
duplicated by a BD-DCJ operation in $S$, the 
marker $\snd{x}$ in a way to produce either adjacency $(\snd{x}~x)$, 
or $(x~\snd{x})$ as done in $S$.

It is easy to see that $D^R = A$. The DCJ (resp. reversal) scenario
between $D^R$ and $B$ having the same length as $S$, is the sequence of
DCJ  (resp. reversal) contained in $S$ or in BD-DCJ (resp. BD-reversal) 
operations of $S$, with the same order as in $S$.\qed
\end{proof}

\noindent For example, in the following, a BD-reversal scenario between
$A=  (\circ ~ 1~~2~~3~~4~~5~ {\circ})$
and $B = (\circ ~ 1~\msnd{4}~~2~~\snd{3}~\msnd{5} ~\msnd{2}~ \msnd{1}~~4~{-3}~~5~ {\circ})$ is transformed into a reversal scenario between
$D = (\circ ~ 1~~\snd{1}~~\snd{2}~~2~~\snd{3}~~3~~\snd{4}~~4~~\snd{5}~~5~ {\circ})$ and $B$.

$$
\begin{array}{lll}
\text{BD-reversal  scenario} &~~& \text{Reversal ~ scenario}
\\
A = (\circ ~ \mathbf{1}~\breakpoint~2~~3~\breakpoint~\mathbf{4}~~5~ {\circ}) & & D = (\circ ~ 1~\breakpoint~\snd{1}~~\snd{2}~~2~~\snd{3}~~3~~\snd{4}~\breakpoint~4~~5~~\snd{5}~ {\circ})

\\(\circ ~1~\msnd{4}~\breakpoint~{-3}~\mathbf{-2}~\breakpoint~\msnd{1}~~4~~5~ {\circ}) & & (\circ ~ 1~\msnd{4}~\breakpoint~{-3}~\msnd{3}~{-2}~\breakpoint~\msnd{2}~\msnd{1}~~4~~5~~\snd{5}~ {\circ})

\\
(\circ ~1~\msnd{4}~~2~~3~\breakpoint~\msnd{2}~\msnd{1}~~4~\breakpoint~5~ {\circ}) & & (\circ ~ 1~\msnd{4}~~2~~\snd{3}~~3~\breakpoint~\msnd{2}~\msnd{1}~~4~\breakpoint~~5~~\snd{5}~ {\circ})
\\
(\circ ~1~\msnd{4}~~2~\breakpoint~~\mathbf{3}~{-4}~~\snd{1}~~\snd{2}~~\mathbf{5}~ \breakpoint~{\circ}) & & (\circ ~ 1~\msnd{4}~~2~~\snd{3}~\breakpoint~3~{-4}~~\snd{1}~~\snd{2}~5~\breakpoint~\snd{5}~ {\circ})

\\(\circ ~1~\msnd{4}~~2~~\snd{3}~{-5}~\msnd{2}~\msnd{1}~~4~{-3}~~\snd{5}~ {\circ}) & & (\circ ~ 1~\msnd{4}~~2~~\snd{3}~{-5} ~\msnd{2}~ \msnd{1}~~4~{-3}~~\snd{5} ~{\circ})
\end{array}
$$

\subsection{Genome Dedoubling}
\label{sec:dedoubling}

I now state the genome dedoubling problem.

\begin{definition}
Given a duplicated genome $G$, the \emph{DCJ (resp. reversal) genome 
dedoubling problem} consists in finding a non-duplicated genome $H$ such 
that the BD-DCJ (resp. BD-reversal) distance between $H$ and $G$ is minimal.
\end{definition}

Given a duplicated genome $G$, we denote by $d_{dcj}(G)$ (resp. $d_{rev}(G)$),
the minimum BD-DCJ (resp. 
BD-reversal) distance between any non-duplicated genome and $G$.
From Proposition \ref{BD-DCJtoDCJ}, the following proposition is straigthforward.

\begin{proposition}
\label{BD-dedoubletoDCJ-dedouble}
Given a duplicated genome $G$, the DCJ (resp. reversal) genome dedoubling problem on $G$ is equivalent to finding  a dedoubled genome $D$ such that the DCJ  (resp. reversal) distance between  $D$ and $G$ is minimal.
\end{proposition}

The next proposition describes a further reduction of the genome dedoubling problem on a duplicated genome $G$.

\begin{proposition}
\label{partial-duptototal-dup}
Given a duplicated genome $G$, the \emph{DCJ (resp. reversal) genome dedoubling problem} on $G$ is equivalent  to the \emph{DCJ (resp. reversal) genome dedoubling problem} on the totally duplicated genome $G^T$ obtained from $G$ by replacing every maximal subsequence of non-duplicated markers beginning with a marker $x$ by the pair $x~\snd{x}$.
\end{proposition}

\begin{proof}
By definition, building adjacencies of the type $\afs{x}{x}$ or $\asf{x}{x}$is the focus of breakpoint-duplication rearrangement scenarios, and as it will be shown in the next sections, destroying such already formed adjacencies is never needed. Therefore, there exists an optimal scenario that preserves the consecutivity of unduplicated markers grouped into a subsequence.\qed
\end{proof}

For example, solving the DCJ (resp. reversal) genome dedoubling problem on
 $G = (\circ ~ 1~~4~\mathbf{-7}~~\snd{1}~\mathbf{{-5} ~~10}~{-8}~ \msnd{4}~~\mathbf{2~~6~{-9} ~~3}~~\msnd{8}~ {\circ})$ 
is equivalent to solving it on
 $G^T = (\circ ~ 1~~4~\mathbf{{-7}~\msnd{7}}~~\snd{1}~\mathbf{{-5}~\msnd{5}}~{-8}~ \msnd{4}~~\mathbf{2~~\snd{2}}~~\msnd{8}~ {\circ})$.\\
The transformations applied on $G$ to obtain $G^T$ are indicated in bold font.

In the sequel, $G$ will always denote a totally duplicated genome, and 
we focus in Sections \ref{sec:DCJ-dedoubling} and \ref{sec:Reversal-dedoubling}  on the problem of finding  a dedoubled genome $D$ such that the DCJ  (resp. reversal) distance between  $D$ and $G$ is minimal.

\subsection{DCJ}
\label{sec:DCJ-dedoubling}

In this section, $G$ denotes a totally duplicated genome. In order to give
a formula for the DCJ dedoubling distance of $G$, $d_{dcj}(G)$, we use a 
graph called the  \emph{dedoubled adjacency graph} of $G$.

\subsubsection{Dedoubled adjacency graph}

\begin{definition}
The \emph{dedoubled adjacency graph} of $G$, denoted by $\mathcal{A}(G)$, 
is the graph whose vertices are the adjacencies of $G$, and for any marker $x$ 
there is one edge between  the vertices $\aff{x}{u}$ and  $\afs{v}{x}$, 
and  one edge between  the vertices $\aff{y}{x}$ and  $\asf{x}{z}$.
\end{definition}

An example of dedoubled adjacency  graph is depicted in 
Fig. \ref{fig:adjacency}. In the following, we will simply refer to 
dedoubled adjacency graphs as
adjacency graphs.

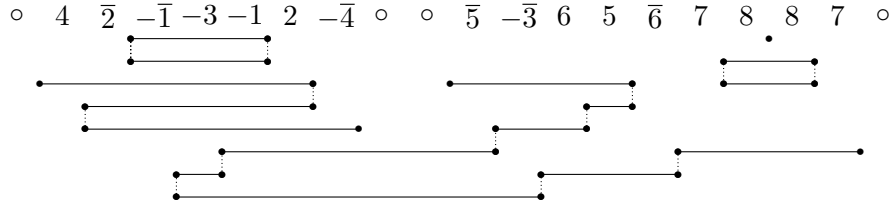
\begin{figure}[htbp]
    \centering
    
\begin{tikzpicture}[xscale=0.3,yscale=0.3]
    %\draw[style=help lines] (0,0) grid (38,-10);
    
    \foreach \i/\v in {0/\circ,2/4,4/\snd{2},6/\msnd{1},8/{-3},10/{-1},12/2,14/\msnd{4},16/\circ,18/\circ,20/\snd{5},22/\msnd{3},24/6,26/5,28/\snd{6},30/7,32/8,34/8,36/7,38/\circ} {
      \draw (\i,0) node {$\v$};
    };

    \fontsize{4pt}{4pt}\selectfont
    %\tikzstyle{every node}=[circle,draw,fill]

    \draw (33,-1) node {$\bullet$};

    \draw 
    (31,-2) node {$\bullet$} --
    (35,-2) node {$\bullet$};

    \draw[densely dotted]
    (35,-2) node {$\bullet$} --
    (35,-3) node {$\bullet$};
    \draw
    (35,-3) node {$\bullet$} --
    (31,-3) node {$\bullet$};
    \draw[densely dotted] 
    (31,-3) node {$\bullet$} -- 
    (31,-2) node {$\bullet$};

    \draw 
    (5,-1) node {$\bullet$} --
    (11,-1) node {$\bullet$};
    \draw[densely dotted] 
    (11,-1) node {$\bullet$} --
    (11,-2) node {$\bullet$};
    \draw
    (11,-2) node {$\bullet$} --
    (5,-2) node {$\bullet$};
\draw[densely dotted]  
   (5,-2) node {$\bullet$} --
   (5,-1) node {$\bullet$} -- cycle;

    \draw   
    (1,-3)  node {$\bullet$} --
    (13,-3) node {$\bullet$};
\draw[densely dotted]  
    (13,-3) node {$\bullet$} --
    (13,-4) node {$\bullet$};
\draw
    (13,-4) node {$\bullet$} --
    (3,-4)  node {$\bullet$};
\draw[densely dotted]  
    (3,-4)  node {$\bullet$} --
    (3,-5)  node {$\bullet$};
\draw
    (3,-5)  node {$\bullet$} --
    (15,-5) node {$\bullet$};

    \draw 
    (19,-3) node {$\bullet$} --
    (27,-3) node {$\bullet$};
    \draw[densely dotted]   
    (27,-3) node {$\bullet$} --
    (27,-4) node {$\bullet$};
    \draw
    (27,-4) node {$\bullet$} --
    (25,-4) node {$\bullet$};
    \draw[densely dotted]   
    (25,-4) node {$\bullet$} --
    (25,-5) node {$\bullet$};
    \draw
    (25,-5) node {$\bullet$} --
    (21,-5) node {$\bullet$};
    \draw[densely dotted]      
    (21,-5) node {$\bullet$} --
    (21,-6) node {$\bullet$};
    \draw
    (21,-6) node {$\bullet$} --
    (9, -6) node {$\bullet$};
    \draw[densely dotted]  
    (9, -6) node {$\bullet$} --
    (9, -7) node {$\bullet$};
    \draw
    (9, -7) node {$\bullet$} --
    (7, -7) node {$\bullet$};
    \draw[densely dotted]  
    (7, -7) node {$\bullet$} --
    (7, -8) node {$\bullet$};
    \draw
    (7, -8) node {$\bullet$} --
    (23,-8) node {$\bullet$};
    \draw[densely dotted]  
    (23,-8) node {$\bullet$} --
    (23,-7) node {$\bullet$};
    \draw
    (23,-7) node {$\bullet$} --
    (29,-7) node {$\bullet$};
    \draw[densely dotted]  
    (29,-7) node {$\bullet$} --
    (29,-6) node {$\bullet$};
    \draw
    (29,-6) node {$\bullet$} --
    (37,-6) node {$\bullet$};

\end{tikzpicture}

\caption{The adjacency graph of $G=(\circ ~ 4~~\snd{2}~\msnd{1}~{-3}~{-1}~~2~\msnd{4}~ {\circ})~~ (\circ ~\snd{5}~\msnd{3}~~6~~5~~\snd{6}~~7~~8~~8~~7 {\circ})$}
\label{fig:adjacency}
\end{figure}

Note that all vertices in  $\mathcal{A}(G)$ have degree one or two.
Thus, the connected components of  $\mathcal{A}(G)$ are only paths 
and cycles. These paths and cycles are called \emph{elements} of 
$\mathcal{A}(G)$.

Given a couple of paralogous markers $(x,\snd{x})$, an element
of the graph $\mathcal{A}(G)$  is said to \emph{contain} the couple 
$(x,\snd{x})$ if it contains the edge linking vertices $\aff{x}{u}$ and  
$\afs{v}{x}$, or the edge linking vertices  $\aff{y}{x}$ and  $\asf{x}{z}$.

By definition,  a couple $(x,\snd{x})$ can possibly be contained in only one
element $A$ of $\mathcal{A}(G)$ if said element $A$ contains both edges  
$(\aff{x}{u},\afs{v}{x})$ and  $(\aff{y}{x},\asf{x}{z})$. In this case,
$A$ is said to \emph{contain twice} the couple $(x,\snd{x})$, and 
$A$ is called a \emph{duplicated} element of $\mathcal{A}(G)$. If an element 
$A$ contains no couple  $(x,\snd{x})$ twice, then it is called a 
\emph{non-duplicated} element of $\mathcal{A}(G)$.
If the two edges $(\aff{x}{u},\afs{v}{x})$ and  $(\aff{y}{x},\asf{x}{z})$
belong to two different elements $A$ and $B$ of $\mathcal{A}(G)$, then  
$A$ and $B$ will both contain $(x,\snd{x})$. In this case, we say that
$A$ and $B$ \emph{intersect}. If two elements $A$ and $B$ do not intersect, then we say that $A$ and $B$ are \emph{independent}.
For example in Fig. \ref{fig:adjacency} the two paths of the adjacency 
graph are duplicated, while the three cycles are non-duplicated. The leftmost
path and the leftmost cycle intersect because they both contain the couple 
$(2,\snd{2})$, while the two paths are independent.
 Given an element $A$ of $\mathcal{A}(G)$, the \emph{set induced by} $A$ 
is the set of couples $(x,\snd{x})$ contained in $A$.

\subsubsection{General sorting}

In this section, I prove the following theorem:

\begin{theorem}
\label{thm:general-dcj}
Let $n$ be the number of couples of paralogous markers in $G$. Let
$C_{i}$ be the maximum size of a subset of non-duplicated pairwise 
independent cycles in $\mathcal{A}(G)$.
The DCJ dedoubling distance of $G$ is $d_{dcj}(G) = n-C_{i}$. It is NP-hard to compute.
\end{theorem}

For example, in Fig. \ref{fig:adjacency}, the maximum size of a 
subset of non-duplicated pairwise independent cycles is $2$ as there
are three cycles, and the two rightmost cycles intersect. The distance 
would then be $d_{dcj}(G) = 8-2=6$.

To prove Theorem \ref{thm:general-dcj}, I use the following properties:

\begin{property}
\label{prop:independent}
Let $n$ be the number of couples of paralogous markers in $G$.
\begin{enumerate}
\item The maximum size $C_i$ of a set of pairwise independent 
cycles in the graph $\mathcal{A}(G)$ is $n$, since each of the $n$
couples $(x,\snd{x})$ is 
contained in at most one cycle of a set of pairwise independent 
cycles. 
\item By definition of a dedoubled genome, if $G$ is a dedoubled genome, 
then the graph $\mathcal{A}(G)$
has $n$ non-duplicated pairwise independent cycles, 
each one containing a single couple of paralogous markers, plus possibly 
other cycles. Thus, in this case, $C_i = n$.
\item A DCJ operation can only alter the maximum size $C_i$ of a 
set of pairwise independent cycles by $-1$, $0$ or $+1$, because
a DCJ operation can only either extract
a new cycle that contributes to increase $C_i$ by $1$, or destroy a single 
cycle in any  set of pairwise independent cycles, thus decreasing 
$C_i$ by $1$, or leaves $C_i$ unchanged.
\end{enumerate}
\end{property}

Algorithm \ref{algo:general-dcj} is an algorithm that provides a $n-C_i$ 
length DCJ scenario transforming $G$ into a  dedoubled genome. It is FPT in the number of cycles in the graph.

\begin{algorithm}
\caption{Transforming a totally duplicated genome $G$ into a dedoubled genome by DCJ}
\label{algo:general-dcj}
\begin{algorithmic}[1]

\STATE Construct $\mathcal{A}(G)$.
\STATE Choose a maximum size set $S_i$ of non-duplicated pairwise independant cycles.
\FOR{Any couple $(x,\snd{x})$ of paralogous markers}
\IF{$(x,\snd{x})$ is contained in a cycle $c$ of $S_i$ containing more than one couple}
\STATE Perform the DCJ that creates adjacency $(x~\snd{x})$ or $(\snd{x}~x)$
by splitting $c$\\ into two cycles, one of the cycles containing only the couple
$(x,\snd{x})$. 
\STATE Replace $c$ in $S_i$ by the two new cycles.
\ELSE
\STATE Perform any DCJ that creates adjacency $(x~\snd{x})$ or $(\snd{x}~x)$, unless such adjacency is already present.
\ENDIF
\ENDFOR
\end{algorithmic}
\end{algorithm}

We now have all the pre-requisites to give the proof of Theorem  
\ref{thm:general-dcj}.

\begin{proof}{of Theorem  \ref{thm:general-dcj}.}
From Property \ref{prop:independent}, a DCJ operation cannot increase 
$C_i$ by more than $1$. Algorithm \ref{algo:general-dcj} provides a 
DCJ scenario transforming $G$ into a dedoubled genome, by increasing  
$C_i$ by $1$ at each DCJ operation until it reaches its upper bound $n$.  
Algorithm  \ref{algo:general-dcj} then provides a minimum length scenario
which is of length $n-C_i$ (any shorter scenario would necessarily increase $C_i$ by more than 1 at some point which is a contradiction).\qed
\end{proof}

\begin{lemma}
\label{lem:complecity-dcj}
Choosing a maximum size set of pairwise independant cycles is a NP-hard, APX-complete problem, approximable with an approximation ratio of $2$.
\end{lemma}
\begin{proof}
We show the equivalence of the problem with a
 \emph{2-frequency Maximum Set Packing}, known to be APX-complete
\cite{BF-95} and 2-approximable \cite{H-83}. 
A \emph{2-frequency collection of sets} 
is a collection of finite sets such that each element of any set belongs 
to at most two sets of the collection. Given a 2-frequency collection $C_n$
of sets, the 2-frequency Maximum Set Packing problem on $C_n$
asks to find the maximum
number of pairwise disjoint sets in $C_n$.

Computing the maximum size $C_{i}$ of a subset of non-duplicated pairwise independent cycles in $\mathcal{A}(G)$ can obviously be reduced to the \emph{2-frequency 
Maximum Set Packing problem} on a 2-frequency collection $C_n$ of sets:

\begin{itemize}
\item Treat each non-duplicated cycle in $\mathcal{A}(G)$ as the set of its edges.
\item Transform them into sets from $C_n$ by replacing each edge $(x~\snd{x})$ or $(\snd{x}~x)$ by the element $x$.
\end{itemize}

Conversely, a 2-frequency collection $C_n$
of sets, containing elements in the form $(k_1,...,k_n)$, can be converted into a totally duplicated genome $G$ such that 
the non-duplicated cycles of $\mathcal{A}(G)$ induce the sets of $C_n$:

\begin{itemize}
\item Create markers $i$ and $\snd{i}$ for all distinct element from all sets.
\item For each set $(k_1,...,k_n)$ create all adjacencies $(k_i~\snd{k_{i+1}})$, as well as $(k_n~\snd{k_1})$ (if one marker extremity is already taken, use both paralogs instead, creating $(\snd{k_i}~k_{i+1})$ or $(\snd{k_n}~k_1)$). 
\item Put telomeres on all remaining free extremities.
\end{itemize}
\qed
\end{proof}

\begin{corollary}
\label{lem:hardness-dcj}
The Genome Dedoubling problem by DCJ is NP-complete. Algorithm 
\ref{algo:general-dcj} solves the problem 
in linear time complexity, except for the computation of the set 
of cycles $S_i$ that is 2-approximable.
\end{corollary}

\begin{corollary}
As the computation of the set 
of cycles $S_i$ can be done in exponential time with respect to the total number of cycles, algorithm 
\ref{algo:general-dcj} is FPT in the number of cycles in $\mathcal{A}(G)$. Indeed, the number of cycles in the graph is independent from the genome length (since it is always possible to merge cycles with DCJ, it is possible to build arbitrarily long genomes with any fixed number of cycles).
\end{corollary}

\subsubsection{Using the dual layered view}
\label{sec:dedoublebonus}

In this section, I will use the ideas I developed in section \ref{sec:dl} to quickly establish the distance formula and thus show the same basic reasonings still hold for more complex problems.

First, the \emph{number of expected operations} is $n$. This is straightforward since we are trying to reconstruct $n$ doublets and a DCJ can always reconstruct one.

Now, in order to establish what is the \emph{sum of bonuses}, we have to keep the graph in mind.

We are trying to reconstruct 1-cycles for each of the markers, and although each marker appears twice in the graph, we only need one of them. This implies the following properties.

\begin{enumerate}
\item{The bonuses are all in the cycles, since only the cycles will eventually give a 1-cycle as remainder.}

\item{However, if a cycle contains a same marker twice, it won't count as a bonus. This is because having adjacency $(x \snd{x})$ as free remainder when adjacency $(\snd{x} x)$ is already present is not a bonus, since having both is not necessary.}

\item{More generally, and for the same reason, if we have a cycle that contains $(x \snd{x})$, then another cycle containing $(\snd{x} x)$ won't count as a bonus either.}
\end{enumerate}

It follows that there are as many bonuses as there are ``non-duplicated (property 2) independent (property 3) cycles (property 1)".

The minimum distance is reached through the maximum number of bonuses. Therefore we are looking for a maximum number of non-duplicated independent cycles. This is exactly the 2-frequency set-packing problem, which is NP-hard.

\subsubsection{Sorting between linear unichromosomal genomes}

In this section, as a first step to study genome dedoubling by reversal, we search for a minimum length DCJ scenario that transforms
$G$ into a dedoubled genome consisting of a single linear chromosome.

In this section and the sequel, $G$ denotes a totally 
duplicated genome consisting of a single linear chromosome. In this case,
the graph $\mathcal{A}(G)$ contains exactly \emph{one path, and possibly several
cycles}.

\begin{definition}
The path in $\mathcal{A}(G)$ is said to be \emph{valid} if it contains every 
couple $(x,\snd{x})$ of paralogous markers in $G$.
\end{definition}

A DCJ operation that merges a cycle $c$ of $\mathcal{A}(G)$ in the path $p$
is a DCJ operation that acts on an adjacency of $c$ and an adjacency of $p$,
thus gathering $c$ and $p$ into a longer path.

Note that if $G$ is a dedoubled genome, then the path in $\mathcal{A}(G)$ is 
necessarily valid. We call such a genome a \emph{dedoubled linear genome}. So, 
if 
the path in $\mathcal{A}(G)$ is not valid, then any DCJ scenario transforming 
$G$ into a dedoubled linear genome will merge, in the path, cycles
containing the couples  $(x,\snd{x})$ that are not contained in the path.

In the following, we always denote by $m$ the minimum number of cycles
required to make the path of $\mathcal{A}(G)$ valid. We also denote by $C_{i}$ 
the maximum size of a subset of non-duplicated pairwise independent cycles.
First, we have the following property:
 \begin{property}
\label{prop:cycles}
Let $C$ be the number of cycles in $\mathcal{A}(G)$. We have $C_i=C-m$.
\end{property}
\begin{proof}
Let $\mathcal{S}$ be the set of all cycles from $G$ ($|\mathcal{S}| = C$).

If $\mathcal{S}_m$ is a set of
cycles that can be merged in the path to make it valid, then  necessarily
$\mathcal{S}-\mathcal{S}_m$ is a set of non-duplicated pairwise independent cycles: 
\begin{itemize}
\item Any cycle $c$ contained in $\mathcal{S}-\mathcal{S}_m$ is non-duplicated, 
otherwise the path obtained after merging would not be valid as it would 
not contain any couple  of paralogous markers $(x,\snd{x})$ contained twice 
in $c$. 
\item Any two cycles $c_1$ and $c_2$ of $\mathcal{S}-\mathcal{S}_m$ are 
independent, otherwise the path obtained after merging would not be valid as it 
would not contain any couple of paralogous markers $(x,\snd{x})$ contained 
in both $c_1$ and $c_2$.

\end{itemize}
\qed
\end{proof}

From Property \ref{prop:cycles}, we then have the following lemma.
\begin{lemma}
\label{lem:linear-dcj}
Let $n$ be the number of couples of paralogous markers in $G$. 
Let $C$ be the number of cycles in $\mathcal{A}(G)$.
The  minimum length $d$ 

of a DCJ scenario transforming
$G$ into a dedoubled genome consisting of a single linear 
chromosome equals $d = n-C+2m$.
\end{lemma}
\begin{proof}
First, from Property \ref{prop:cycles}, we have $n-C+2m = n-C_i+m$.
Similarly to $C_i$, a DCJ operation can only alter $m$ by $+1$, $-1$
or $0$.
Next, a DCJ operation that increases $C_i$ by $1$ also increases $C$ by $1$,
and then leaves $m$ unchanged. Conversely, a DCJ operation that decreases $m$ 
by $1$ leaves $C_i$ unchanged.

 Algorithm \ref{algo:general-dcj} in which we add the line 
(2': Merge in the path all the cycles that are not in $S_i$)
between
lines $2$ and $3$
provides a DCJ scenario that first decreases $m$ until it reaches its lower 
bound $0$ (in $m$ DCJ operations), then increases $C_i$ until it reaches its 
upper bound $n$ (in $n-C_i$ DCJ operations). \qed
\end{proof}

\begin{corollary}
\label{lem:hardness-dcj-linear}
The problem of finding a DCJ scenario transforming $G$ into a dedoubled genome 
consisting of a single linear chromosome is NP-hard.
Algorithm \ref{algo:general-dcj}, in which we add the line 
(2': Merge in the path all the cycles that are not in $S_i$)
between lines $2$ and $3$, solves the problem in linear time complexity, 
except for the computation of the set of cycles $S_i$ that is 2-approximable.
\end{corollary}

\subsection{Reversal}
\label{sec:Reversal-dedoubling}

The graph used in this section behaves like the \emph{arc overlap graph}
used in \cite{B-05} for the Hannenhalli-Pevzner theory of sorting by 
reversal \cite{HP-95}. The genome $G$ consists of a single linear
chromosome.

\subsubsection{Dedoubled overlap graph}

For any couple $(x,\snd{x})$ of paralogous markers in $G$, the segment of $(x,\snd{x})$ is the smallest segment of $G$
containing both markers $x$ and $\snd{x}$. For example, in genome 
$G=(\circ ~ 1~~3~~\snd{1}~\msnd{2}~{-4}~ \msnd{3}~~2~\msnd{4}~ {\circ})$, the 
segment of  $(1,\snd{1})$ is $(1~~3~~\snd{1})$, and the segment of 
$(2,\snd{2})$ is $(\msnd{2}~{-4}~ \msnd{3}~~2)$.

\begin{definition}
The \emph{dedoubled overlap graph} of $G$, denoted by $\mathcal{O}(G)$, 
is the graph whose vertices are all the couples $(x,\snd{x})$ of paralogous 
markers of $G$, and there is an edge between two vertices  $(u,\snd{u})$ and  
$(v,\snd{v})$ if the segments of $u$ and $v$ overlap. 
\end{definition}

An example of dedoubled overlap graph is depicted in Fig. \ref{fig:overlap}.a.
In the following, I will simply refer to dedoubled overlap graphs as
overlap graphs.

The vertex $(x,\snd{x})$ of the graph $\mathcal{O}(G)$ is \emph{oriented} 
if $x$ and $\snd{x}$ have different signs in $G$, 
otherwise it is \emph{unoriented}. If the vertex $(x,\snd{x})$ is 
oriented then there exists a reversal operation denoted by 
$\text{Rev}(x~~\snd{x})$ 
that produces the adjacency $(x~~\snd{x})$ and a reversal operation 
denoted by $\text{Rev}(\snd{x}~~x)$ that produces the adjacency 
$(\snd{x}~~x)$. 
For example, in genome  $G=(\circ ~ 1~~3~~\snd{1}~\msnd{2}~{-4}~ \msnd{3}~~2~\msnd{4}~ {\circ})$, both copies of 3 have opposite sign, therefore $(3,\snd{3})$ is an oriented vertex of $\mathcal{O}(G)$.\\
\\
$\text{Rev}(3~~\snd{3}) = (\circ ~ 1~~3~\underline{~\snd{1}~\msnd{2}~{-4}~ \msnd{3}~}~2~\msnd{4}~ {\circ})  \rightarrow (\circ ~ 1~~3~~\snd{3}~~4~~\snd{2}~\msnd{1}~~2~\msnd{4}~ {\circ})$.\\
\\
$\text{Rev}(\snd{3}~~3) = (\circ ~ 1~\underline{~3~~\snd{1}~\msnd{2}~{-4}}~ \msnd{3}~~2~\msnd{4}~ {\circ})  \rightarrow (\circ ~ 1~~4~~\snd{2}~\msnd{1}~{-3}~~\msnd{3}~~2~\msnd{4}~ {\circ})$.\\
\\

\begin{figure}[h!]
    \centering
 a.   \begin{tikzpicture}[xscale=2,yscale=1.5]
        \foreach \i/\x/\y/\c in {1/0/1/white,2/0/0/black!10,3/1/1/black!10,4/1/0/white,5/2/1/white,6/2/0/white} {
          \draw (\x,\y) node[circle,draw,fill=\c] (N-\i) {$\left(\fst{\i},\snd{\i}\right)$};
        };
        \draw (N-1) -- (N-3);
        \draw (N-2) -- (N-3);
        \draw (N-3) -- (N-4);
        \draw (N-2) -- (N-4);
        \draw (N-5) -- (N-6);
    \end{tikzpicture}
b.
\begin{tikzpicture}[xscale=2,yscale=1.5]
        \foreach \i/\x/\y/\c in {1/0/1/black!10,2/0/0/white,3/1/1/white,4/1/0/black!10,5/2/1/white,6/2/0/white} {
          \draw (\x,\y) node[circle,draw,fill=\c] (N-\i) {$\left(\fst{\i},\snd{\i}\right)$};
        };
        \draw (N-1) -- (N-2);
        \draw (N-1) -- (N-4);
        \draw (N-5) -- (N-6);
    \end{tikzpicture}

\caption{a. The overlap graph of $G=(\circ ~ 1~~3~~\snd{1}~\msnd{2}~{-4}~ \msnd{3}~~2~\msnd{4}~ {\circ})~~ (\circ ~\snd{5}~6~~5~~\snd{6}~ {\circ})$. Oriented vertices are colored in grey. The graph
$\mathcal{O}(G)$ has two connected components, one oriented and one unoriented. 
b. the overlap graph obtained after applying the reversal $\text{Rev}(3~~\snd{3})$ to produce adjacency $(3~~\snd{3})$.}
\label{fig:overlap}
\end{figure}
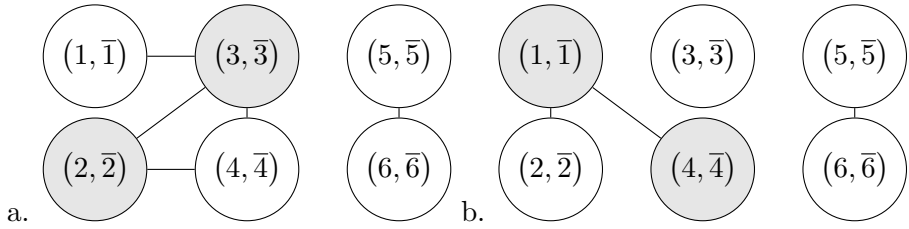

The overlap graph of $G$ behaves like arc overlap graphs used in \cite{B-05}
for the Hannenhali-Pevzner theory of sorting by reversal \cite{HP-95}. Indeed, given an 
oriented vertex 
$(x,\snd{x})$ of the graph $\mathcal{O}(G)$, performing the reversal 
$\text{Rev}(x~~\snd{x})$ or $\text{Rev}(\snd{x}~~x)$ complements the subgraph induced by $(x,\snd{x})$ and all its neighbouring vertices, and changes the orientation of all vertices in this subgraph  (see Fig. \ref{fig:overlap}.b).

A connected component of the graph $\mathcal{O}(G)$ 
is \emph{oriented} if it contains at least one oriented vertex, 
otherwise it is \emph{unoriented}.
A genome is \emph{oriented} if all connected components of 
the graph $\mathcal{O}(G)$ are oriented, otherwise it is  \emph{unoriented}.

Given  an oriented vertex $(x,\snd{x})$ of the graph $\mathcal{O}(G)$,
the score of $(x,\snd{x})$ is the number of oriented vertices in
the genome obtained after applying $\text{Rev}(x~~\snd{x})$ on $G$. 
Note that the same number of oriented vertices is obtained after applying 
 $\text{Rev}(\snd{x}~~x)$ on $G$.

If you are familiar with other papers on sorting by reversal, you might notice this score definition differs from the usual one. In fact they both produce the same ranking of vertices.

\begin{property}
\label{prop:max-score}
Let $(x,\snd{x})$ be an oriented vertex of  $\mathcal{O}(G)$ of maximum score.
Performing $\text{Rev}(x~~\snd{x})$ or  $\text{Rev}(\snd{x}~~x)$ does not
create  new unoriented connected components in the overlap graph of the
genome obtained.
\end{property}
\begin{proof}
Because the vertices ranking is the same, the proof goes the same way as the proof of Theorem 10 in \cite{B-05}: if 
$\text{Rev}(x~~\snd{x})$ or  $\text{Rev}(\snd{x}~~x)$ creates a 
new unoriented connected component $C$ in the overlap graph. Then, 
$C$ necessarily 
contains a vertex $(w,\snd{w})$ that was adjacent to $(x,\snd{x})$ and
oriented before applying $\text{Rev}(x~~\snd{x})$ or  $\text{Rev}(\snd{x}~~x)$, 
and such that the score of $(w,\snd{w})$ was greater than the score of 
$(x,\snd{x})$, which is a contradiction.\qed
\end{proof}

In the sequel, we focus on sorting oriented genomes using reversal 
dedoubling scenarios. 
A totally duplicated genome $G$ consisting of a single linear chromosome 
is called a \emph{valid-path genome} if the single path in $\mathcal{A}(G)$ 
is valid. Otherwise, it is called a \emph{non-valid-path genome}.

\subsubsection{Sorting an oriented valid-path genome}

In this section, we consider an oriented valid-path genome $G$.

With $n$ the number of couples of paralogous 
markers in $G$, and $C$ the number of cycles in  $\mathcal{A}(G)$, I prove the following theorem:

\begin{theorem}
\label{th:vp-rev}
The reversal dedoubling distance of $G$ is  $d_{rev}(G) = n-C$.
\end{theorem}
\begin{proof}
In Algorithm \ref{algo:vp-rev}, since $G$ is a valid path genome, then
$S_i$ is the set of all cycles. Thus,  Algorithm \ref{algo:vp-rev}
 provides a reversal scenario of length $n-C$, 
which is the smallest length that can be reached since 
$d_{dcj}(G) = n-C_i=n-C+m=n-C$ as $m=0$ here, and a $d_{dcj}(G) \leq d_{rev}(G)$ since a reversal scenario is also a DCJ scenario.\qed
\end{proof}

\begin{algorithm}
\caption{Transforming an oriented genome $G$ into a dedoubled genome by reversals}
\label{algo:vp-rev}
\begin{algorithmic}[1]
\STATE Construct $\mathcal{A}(G)$.
\STATE Construct $\mathcal{O}(G)$.
\STATE Choose a maximum size set $S_i$ of non-duplicated pairwise independant cycles.
\STATE Merge in the path all the cycles that are not in $S_i$.
\WHILE{$G$ is not a dedoubled genome}
\STATE Pick a maximum score vertex $(x,\snd{x})$ in $\mathcal{O}(G)$.
\IF{$(x,\snd{x})$ is contained in a cycle $c$ of $S_i$}
\STATE Choose between $\text{Rev}(x~\snd{x})$ and $\text{Rev}(\snd{x}~x)$ the
reversal that splits $c$ into two cycles, and perform it. 
\STATE Replace $c$ in $S_i$ by the two new cycles.
\STATE Apply the modification induced by the reversal on $\mathcal{O}(G)$.
\ELSE
\STATE Perform any reversal that creates adjacency $(x~\snd{x})$ or $(\snd{x}~x)$.
\ENDIF
\ENDWHILE
\end{algorithmic}
\end{algorithm}
\vspace{-5mm}%test

\subsubsection{Sorting an oriented non-valid-path genome}
In this section, $G$ denotes an oriented non-valid path genome. At least $m$ cycles of $\mathcal{A}(G)$ have to be merged in the path to make it valid.

An edge $(\aff{x}{u},\afs{v}{x})$ or  $(\aff{y}{x},\asf{x}{z})$ of the adjacency graph $\mathcal{A}(G)$ is called oriented if markers $x$ and $\snd{x}$ have different signs.
Note that extracting a cycle from any element of the graph $\mathcal{A}(G)$ 
requires this element to contain oriented edges. It is easy to see that given two 
adjacencies picked in a given element, a reversal 
acting on these adjacencies will extract a cycle if and only if the path 
linking these adjacencies contains an odd number of oriented edges.
Thus, we have the following lemma:
\begin{lemma}
\label{lem:mergeorient}
Merging a cycle in the path never creates unoriented connected components
in $\mathcal{O}(G)$.
\end{lemma}
\begin{proof}
A \emph{component} of $G$ is the smallest segment of $G$ containing all 
markers of a connected component in $\mathcal{O}(G)$. Breakpoints of merging reversals can always be inside two distinct components
of the genome. Performing the operation gathers these two components of 
the genome, possibly with some other as well, into a single component, obviously containing all couples of paralogous markers that used to be in the merged components.
Therefore, it has to be an oriented component. Assuming otherwise, it would 
be possible to extract a cycle from an unoriented component, which is nonsense, by performing the inverse operation in the resulting genome.\qed
\end{proof}

\begin{theorem}
Let $G$ be a non-valid-path oriented genome. Let $C$ be the number of cycles in the graph  and $m$ be the minimum number of cycles to merge in the path to make it valid. The reversal dedoubling distance of $G$ is  $d_{rev}(G) = n-C+2m$.
\end{theorem}
\begin{proof}

From lemma \ref{lem:linear-dcj}, we have that $d_{rev}(G) \geq n - C + 2m$ as a reversal scenario is always a DCJ scenario.
Algorithm \ref{algo:vp-rev} provides a scenario of length 
$n-C_i + m = n - C + 2m$.
\qed
\end{proof}

\begin{corollary}
\label{lem:hardness-rev}
The Genome Dedoubling problem by reversal on oriented genomes is NP-hard.
Algorithm \ref{algo:vp-rev} solves the problem in quadratic time complexity, 
except for the computation of $S_i$ that is 2-approximable.
\end{corollary}

\subsubsection{A few words on unoriented genomes}
\label{unoriented}

I reviewed this work with Laurent Bulteau, visiting my lab for a week in May 2012. 

We spent time on the last unsolved step: sorting unoriented genomes by reversal.

In \cite{BMS-04}, the distance for sorting between non-duplicated genomes by reversal is expressed as $d_{rev} = d_{dcj} + t$, this last term being a parameter called \emph{orientation cost}.

Unfortunately, the techniques used in this article cannot be directly applied to the genome dedoubling problem, leaving open the question of computing this cost.

The work was promising thanks to Laurent Bulteau contribution and it seemed we found a dynamic programming algorithm for computing the orientation cost in polynomial time.
However we did not have time to finish this work.

I leave the result here as a conjecture. 

\begin{conjecture}
The orientation cost for genome dedoubling by reversal can be computed in polynomial time using dynamic programming. This leads to another FPT algorithm in the number of cycles in $\mathcal{A}(G)$ for genome dedoubling by reversal.
\end{conjecture}

\subsection{Closing words and credits}

Upon presenting this work at \emph{RECOMB-CG 2011, Galway, Ireland}, I learnt from Anne Bergeron that the DCJ genome dedoubling was a problem Aida Ouangraoua originally worked on with her, and that she did not know Aida O. kept working on it with other people after she went back to France.

Aida O. explained the dedoubling problem, showed me the reduction from BD scenario to dedoubling scenario, and described the dedoubled adjacency graph but could not prove the DCJ distance.

I took care of proving the DCJ distance, providing a sorting algorithm, as well as studying the problem under reversals (using \cite{BMS-04} as a starting point, I designed a study for the dedoubling problem).

While the original publication contained a biological application using genomes taken from \cite{R07} by Aida O., I chose not to include it here, since I have always considered it was a display of circular reasoning\footnote{I recall that the breakpoint duplication model was made based on Ranz' observations. Using his own data to validate it is irrelevant.}.

Jean-Stéphane Varré drew some of the figures for the article.

A while later I briefly revisited this work with Laurent Bulteau which led to the aforementioned additions.

I took care of proving every single result from these papers.

\section{Model II: Whole tandem duplication}
\label{sec:wholetandem}

I studied this model separately for block interchange and for DCJ.

The block interchange study has been published in \cite{Thomas13pre}, then further extended in a journal version \cite{Thomas13}.

The DCJ study has been published in \cite{Thomas12}.

I proved the single tandem halving problem is polynomial for both BI and DCJ, providing $O(n^2)$ algorithms for the scenario and a $O(n)$ computation for the distance.

\subsection{Biological motivation}

This problem was designed as a starting point on rearrangement problems using segmental duplications as the cause to duplicated content. Since the genome halving is the main polynomial result on duplicated genome rearrangements, I used it as a base and made a tandem duplication variant.

While a tandem duplication of the whole content of a unilinear genome might not make much sense biologically, we were hoping that the results would allow a better insight on more realistic tandem duplication models, as shown in further studies from section \ref{sec:partialdup}.

\subsection{Model}

\subsubsection{Considered genomes}
\label{sec:dedouble}

The genomes considered in this section are totally duplicated genomes, and perfectly duplicated genomes, as defined in section \ref{sec:notations2}.

I also introduce single tandem duplicated genomes (or 1-tandem duplicated genomes), and a process called reduction whose purpose is just to help shortening the definition:

\emph{Reduction} is the process of rewriting a consecutive sequence of double-adjacencies a
single marker. In the following example, 
genome $G$ could be reduced by rewriting $\fst{2}~\diamond~\fst{4}~\diamond~\fst{5}$
and their paralogs

as $\fst{10}$ and $\snd{10}$:

$$G =
(\circ~~\fst{1}~~\snd{1}~~\fst{3}~~\fst{2}~\diamond~\fst{4}~\diamond~\fst{5}~~\fst{6}
~~\snd{6}~~\fst{7}~~\snd{3}~~\fst{8}~~\snd{2}~\diamond~\snd{4}~\diamond~\snd{5}~~\fst{
9}~~\snd{8}~~\snd{7}~~\snd{9}~~\circ )$$
$$G^r = (\circ
~~\fst{1}~~\snd{1}~~\fst{3}~~\fst{10}~~\fst{6}~~\snd{6}~~\fst{7}~~\snd{3}~~\fst{
8}~~\snd{10}~~\fst{9}~~\snd{8}~~\snd{7}~~\snd{9}~~\circ )$$

\begin{definition}
\label{def:tandem}

    A \emph{single tandem duplicated genome}  (or \emph{1-tandem duplicated genome}) is a totally
    duplicated genome which can be reduced to a genome of the form
    $(\circ~~\fst{x}~~\snd{x}~~\circ )$.
\end{definition}

In other words, a tandem duplicated genome is composed of a single linear
chromosome where all adjacencies, except the two telomeric adjacencies 

and the central adjacency, are
double-adjacencies.
For example, the genome
$(\circ~\fst{1}\diamond\fst{2}\diamond\fst{3}\diamond\fst{4}~\snd{1}\diamond\snd{2
}\diamond\snd{3}\diamond\snd{4}~\circ)$ is a tandem-duplicated genome that can
be reduced to 

$(\circ~\fst{5}~\snd{5}~\circ )$
by rewriting $\fst{1}\diamond\fst{2}\diamond\fst{3}\diamond\fst{4}$ and
$\snd{1}\diamond\snd{2}\diamond\snd{3}\diamond\snd{4}$ as  $\fst{5}$ and $\snd{5}$.

I recall the perfectly duplicated genome definition.

\begin{definition}
\label{def:perfectly}
A \emph{perfectly duplicated genome} is a totally duplicated genome such that
all adjacencies are double-adjacencies, none of them in the form $(x~\msnd{x})$.
\end{definition}

For example, the genome $(1~~{2}~~{3}~~4~~\snd{1}~~\snd{2}~~\snd{3}~~\snd{4})$
is a perfectly duplicated genome, while $(\circ~~1~~2~~\msnd{2}~~-1~~\circ)$ is not.

It is to note that this definition is equivalent to the one from \cite{Mixtacki08}, which slightly differs from the one originally introduced in \cite{Warren08}.

From definitions \ref{def:tandem} and \ref{def:perfectly}, we get the following property:

\begin{property}
    \label{prop:dtdp}
    In the case of
    unichromosomal genomes, a perfectly duplicated genome is a single
    tandem duplicated genome which has been circularized (the
    perfectly duplicated genome can be reduced to $(x~\snd{x})$, it
    just lacks telomeres).
\end{property}

\subsubsection{Considered operations}

I use two separate operation models in this section: Block Interchange, and DCJ.

Here is a useful property linking BI operations to DCJ operations.

\begin{property}
\label{1BIto2DCJ}
A single BI operation on a linear chromosome is equivalent to two DCJ operations: an excision followed by a reintegration.
\end{property}

\begin{proof}
Let $(\circ~~1~~U~~2~~V~~3~~\circ)$ be a genome, $U$ and $V$ the two intervals 
that are to be swapped by a block interchange operation, $1$ $2$ and $3$ the 
intervals constituting the rest of the genome (note that each of them may be 
empty). 

The first DCJ operation is the excision that produces the adjacency
$(1~~V)$ by extracting and circularizing the interval $[U~;~2]$: 
$$(\circ~~1~\breakpoint~U~~2~\breakpoint~V~~3~~\circ) \rightarrow
(\circ~~1~~V~~3~~\circ)(U~~2~)$$

The second DCJ operation is the integration that produces the adjacency
$(U~~3)$ by reintegrating the circular chromosome $(U~~2)$ in the
appropriate way:
$$(\circ~~1~~V~\breakpoint~3~~\circ)(U~~2~\breakpoint) \rightarrow (\circ~~1~~V~~2~~U~~3~~\circ).$$
\end{proof}

\subsection{Single tandem halving}

Because the single tandem halving problem is close to the genome halving problem, I will recall both problems definitions, with adapted notations to avoid ambiguity between both distances.

\begin{definition}
    Given a unilinear totally duplicated genome $G$, the \emph{single
      tandem halving problem} (or 1-tandem halving problem) consists
    in finding an \emph{optimal} 1-tandem duplicated genome $H$, such that the distance between $G$
    and $H$ is minimal. This minimal distance is called the
    \emph{1-tandem halving distance}, and is denoted $d^t(G)$.
\end{definition}

\begin{definition}[\cite{Mixtacki08}]
    Given a totally duplicated genome $G$, the \emph{DCJ genome
      halving problem} consists in finding an \emph{optimal} perfectly
    duplicated genome $H$, such that the DCJ distance between
    $G$ and $H$ is minimal. This minimal distance is called the
    \emph{genome halving distance} and is denoted $d^p(G)$.
\end{definition}

Thus, when the goal is a 1-tandem duplicated genome, the distance is noted $d^t$, when the goal is a perfectly duplicated genome, it is $d^p$.\\

Moreover, I will also recall the operation model as subscript.

For example, using these notations, $d^p_{DCJ}(G)$ is the DCJ genome halving distance as defined in \cite{Mixtacki08}, while $d^t_{BI}(G)$ is the BI 1-tandem halving distance I am about to prove in the next section.

\subsection{Block Interchange}

\subsubsection{Lowerbound}
\label{sec:lb}

\noindent
In this section I give a lowerbound on the BI 1-tandem halving distance of a totally duplicated genome. I use a data structure representing the genome called the \emph{natural graph} introduced in \cite{Mixtacki08}. 

\begin{definition}\emph{\cite{Mixtacki08}}
The natural graph of a totally duplicated genome $G$, denoted by $\NG(G)$, is the graph whose vertices are the \emph{adjacencies} of $G$, and for any marker $u$ there is one edge between  $\aff{u}{v}$ and  $\asf{u}{w}$, and  one edge between  $\aff{x}{u}$ and  $\afs{y}{u}$.

\end{definition}

Note that the number of edges in the natural graph of a genome $G$ containing $n$ distinct markers, each one present in two copies, is always $2n$. Moreover, since every vertex has degree one or two, then the natural graph consists only of cycles and paths. For example, the natural graph of genome $G = (\circ~~ \fst{1}~~\snd{2}~~\snd{1}~~\snd{4}~~\fst{3}~~\fst{4}~~\snd{3}~~\fst{2}~~\circ)$ is depicted in figure \ref{fig:NGdef}.

\begin{figure}[htbp]
    \centering

\includegraphics[width=4cm]{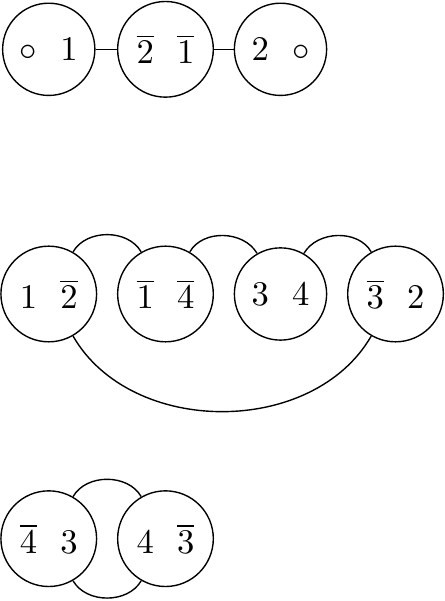}

\caption{The natural graph of genome $G = (\circ~~ \fst{1}~~\snd{2}~~\snd{1}~~\underline{\snd{4}~~\fst{3}}~~\underline{\fst{4}~~\snd{3}}~~\fst{2}~~\circ)$ ; it is composed of one path and two cycles.}
\label{fig:NGdef}
\end{figure}

\begin{definition}
    Given an integer $k$, a \emph{$k-$cycle} (resp. \emph{$k-$path}) in the 
    natural graph of a totally duplicated genome is a cycle (resp. path) 
    that contains $k$ edges. If $k$ is even, the cycle (resp. path) is called 
   \emph{even}, and \emph{odd} otherwise.
\end{definition}

\def\EC{\ensuremath{\mbox{EC}}}
\def\OP{\ensuremath{\mbox{OP}}}

Based on the natural graph, a formula for the DCJ halving distance was given in \cite{Mixtacki08}. Given a totally duplicated genome $G$ such that the number of even cycles and the number of odd paths in $\NG(G)$
are respectively denoted by $\EC$  and $\OP$, the  DCJ halving distance of $G$ is:
    $$d^p_{DCJ}(G) = n - \EC - \left\lfloor \frac{\OP}{2} \right\rfloor $$

In the case of the BI 1-tandem halving distance, some peculiar properties of the natural graph need to be stated, allowing one to simplify the formula of the DCJ halving distance, and leading to a lowerbound on the BI 1-tandem halving distance.

In the following properties, we assume that $G$ is a genome composed of a 
single linear chromosome containing $n$ distinct markers, each one present 
in two copies in $G$. 

\begin{property}
The natural graph $\NG(G)$ contains only even cycles and paths: 
\label{prop:EC}
\label{prop:1EP}
\begin{enumerate}
\item All cycles in the natural graph $\NG(G)$ are even.
\item The natural graph $\NG(G)$ contains only one path, and this path is even.
\end{enumerate}
\end{property}

\begin{proof}
    First, if \aff{a}{x} is a vertex of the
    graph that belongs to a cycle $C$, then there exists an edge between
    \aff{a}{x} and a vertex \asf{a}{y}. These two
    adjacencies are the only two containing a copy of the marker $a$ at the
    first position. So, if we consider the set of all the first markers in
    all adjacencies contained in the cycle $C$, then each marker in this set is
    present exactly twice. Therefore, the cycle $C$ is an even cycle.
  
    Secondly, the graph contains exactly two vertices (adjacencies) containing the
    marker $\circ$ which are both necessarily ends of a path in $\NG(G)$. Thus
   there can be only one path in the graph. Since the number of edges in the 
   graph is even and all cycles are even, then the single path is also even.\qed
\end{proof}

I now give a lowerbound on the minimum length of a DCJ scenario transforming 
$G$ into a 1-tandem duplicated genome.

\begin{lemma}
    Let $d^t_{DCJ}(G)$  be the minimum DCJ distance between $G$ and any 
   1-tandem duplicated genome. If $\NG(G)$ contains $C$ cycles then a 
   lowerbound on $d^t_{DCJ}(G)$ is given by:
    $$d^t_{DCJ}(G) \geq n - C - 1$$
\label{DCJdistLB}
\end{lemma}

\begin{proof}
    First, since all cycles  of $\NG(G)$ are even and $\NG(G)$ contains no odd 
    path, 
   then, from the DCJ halving distance formula, the DCJ halving distance of 
   $G$   is $d^p_{DCJ}(G) = n - C$.

    Now, since any  1-tandem duplicated genome can be transformed into 
    a perfectly duplicated genome with one DCJ, then $d^t_{DCJ} + 1
    \geq d^p_{DCJ}$. Therefore, we have $ d^t_{DCJ} \geq d^p_{DCJ} - 1 \geq
    n - C - 1$.\qed  
\end{proof}

We are now ready to state a lowerbound on the BI 1-tandem halving distance of a totally duplicated genome $G$.

\begin{theorem}
    If $\NG(G)$ contains $C$ cycles, then a lowerbound on the BI 1-tandem halving distance
    is given by: 
$$d^t_{BI}(G) \geq \left \lfloor \frac{n - C}{2} \right \rfloor$$
\end{theorem}

\begin{proof}
    %is the minimal number of DCJ operations required to turn $G$ into a 1-tandem duplicated genome, assuming that $\NG(G)$ contains $C$ cycles. 

We denote by $\ell(S)$ the length of a rearrangement scenario $S$.
   Let $S_{BI}$ be a BI scenario transforming $G$ into a
   1-tandem duplicated genome.
    From property \ref{1BIto2DCJ}, we have that $S_{BI}$ is equivalent to
    a DCJ scenario $S_{DCJ}$ such that  $\ell(S_{DCJ})=2*\ell(S_{BI})$. 
    Now, suppose that $\ell(S_{BI}) < \lfloor \frac{n - C}{2}
    \rfloor$, then $\ell(S_{BI}) \leq \lfloor \frac{n - C}{2}
    \rfloor - 1 \leq \lceil \frac{n - C - 1}{2} \rceil - 1$.
    
    This implies $\ell(S_{DCJ}) \leq 2\lceil \frac{n - C - 1}{2}
    \rceil - 2 \leq n - C - 2 < n - C - 1$. Thus, from Lemma 
   \ref{DCJdistLB} we have $\ell(S_{DCJ}) < d^t_{DCJ}$ 
    which contradicts the fact that $d^t_{DCJ}$ is the minimal number
    of DCJ
    operations required to transform $G$ into a 1-tandem duplicated genome.

    In conclusion, we always have $d^t_{BI}(G) \geq \lfloor \frac{n - C}{2} \rfloor$.  \qed
\end{proof}

\subsubsection{Distance}
\label{sec:dist}

\noindent
In this section, I show that the established lowerbound is in fact the actual distance.

In other words, the BI 1-tandem halving distance of a totally duplicated genome $G$ with $n$ distinct markers such that  $\NG(G)$ contains $C$ cycles is exactly:

$$d^t_{BI}(G) = \left \lfloor \frac{n - C}{2} \right \rfloor$$

Which means that enforcing the constraint that successive couples 
of consecutive DCJ operations have to be equivalent to BI operations does not 
change the distance even though it 
obviously restricts the DCJ that can be performed at each step of the scenario. 

In
the following, $G$ denotes a totally duplicated genome $G$
constisting in a single linear chromosome with $n$ distinct markers
after the reduction process, and such that $\NG(G)$ contains $C$ cycles.

We begin by recalling some useful definitions and properties of the DCJ
operations that allow one to decrease the DCJ halving distance by $1$ in the resulting genome.

\begin{definition}
A DCJ operation on $G$ producing genome $G'$ is \emph{sorting} if it decreases the DCJ halving distance by $1$: $d^p_{DCJ}(G') = d^p_{DCJ}(G)-1 = n-C-1$.
\end{definition}

Since the number of distinct markers in $G'$ is $n$ and $d^p_{DCJ}(G') = n-C-1$, then $\NG(G')$ contains $C+1$ cycles. In other words, a DCJ operation is sorting if it increases the number of cycles in  $\NG(G)$ by $1$.

Given $(\fst{u}~~\fst{v})$ an adjacency of $G$ that is not a double-adjacency,
we denote by $\DCJ(\fst{u}~~\fst{v})$ the DCJ operation that cuts adjacencies $(\snd{u}~~\fst{x})$ and  $(\fst{y}~~\snd{v})$ to form adjacencies $(\snd{u}~~\snd{v})$ and  $(\fst{y}~~\fst{x})$, making $(\fst{u}~~\fst{v})$ a double-adjacency. 

\begin{property}
Let  $(\fst{u}~~\fst{v})$ be an adjacency of $G$ that is not a double-adjacency,
 $\DCJ(\fst{u}~~\fst{v})$ is a sorting DCJ operation.
\end{property}
\begin{proof}
$\DCJ(\fst{u}~~\fst{v})$ increases the number of cycles in  $\NG(G)$ by $1$, by creating a new cycle composed of adjacencies  $(\fst{u}~~\fst{v})$ and  $(\snd{u}~~\snd{v})$.\qed
\end{proof}

\begin{figure}[!h]
\centering

\includegraphics[width=7cm]{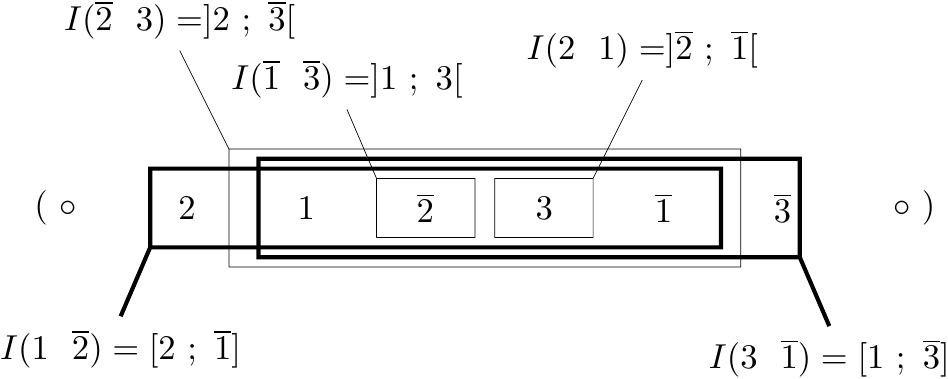}

\caption{ $\mathcal{I}(G) ~ = ~ \left\{ \oiss{2}{1} ~,
      \mathbf{\cifs{2}{1}} ~,~ \oifs{2}{3} ~,~ \mathbf{\cifs{1}{3}}
      ~,~ \oiff{1}{3}  \right\}$, the set of intervals of $G =
  (\circ~~\fst{2}~~\fst{1}~~\snd{2}~~\fst{3}~~\snd{1}~~\snd{3}~~\circ
  )$ depicted as boxes. The two boxes with thick lines represent two 
  overlapping intervals of $\mathcal{I}(G)$ inducing a \BI\xspace which 
  exchanges $\fst{2}$ and $\snd{3}$.}
    \label{fig:intervaldef}
\end{figure}

\begin{definition}
Let $(\fst{u}~~\fst{v})$,  $(\snd{u}~~\fst{x})$, and  $(\fst{y}~~\snd{v})$ be adjacencies of $G$. The \emph{interval} of the adjacency  $(\fst{u}~~\fst{v})$, denoted by $I(\fst{u}~~\fst{v})$ is either:
\begin{itemize}
\item the interval $[x~;~y]$ if $(\snd{u}~~\fst{x}) < (\fst{y}~~\snd{v})$.  In this case, we denote it by $]\snd{u}~;~\snd{v}[$, or
\item the interval $[\snd{v}~;~\snd{u}]$ if $(\fst{y}~~\snd{v}) < (\snd{u}~~\fst{x})$.
\end{itemize}
\end{definition}

For example, the intervals of the adjacencies in genome $(\circ~~\fst{2}~~\fst{1}~~\snd{2}~~\fst{3}~~\snd{1}~~\snd{3}~~\circ )$ are depicted in figure \ref{fig:intervaldef}.
Note that, given an adjacency  $(\fst{u}~~\fst{v})$ of $G$,  if  $(\fst{u}~~\fst{v})$ is a double-adjacency then the interval $I(\fst{u}~~\fst{v})$ is empty, otherwise  $\DCJ(\fst{u}~~\fst{v})$ is the excision operation that extracts the interval $I(\fst{u}~~\fst{v})$ to make it circular, thus producing the adjacency $(\snd{u}~~\snd{v})$.

Two intervals  $I\aff{a}{b}$ and $I\aff{x}{y}$ are said to be \emph{overlapping} if
their intersection is non-empty, and none of the intervals is included in the 
other.
It is easy to see, following Property \ref{1BIto2DCJ}, that given two
adjacencies $\aff{a}{b}$ and $\aff{x}{y}$ of $G$ such that
$I\aff{a}{b}$ and $I\aff{x}{y}$ are non-empty intervals, the
successive application of $\DCJ\aff{a}{b}$ and $\DCJ\aff{x}{y}$ is
equivalent  to a BI operation if and only if $I\aff{a}{b}$ and
$I\aff{x}{y}$ are overlapping. Note that in this case neither
$\aff{a}{b}$, nor  $\aff{x}{y}$ can be double-adjacencies in $G$ since
their  intervals are non-empty. Figure \ref{fig:intervaldef} shows an
example of two overlapping intervals. 

The following property states precisely in which case the successive application of $\DCJ\aff{a}{b}$ and $\DCJ\aff{x}{y}$ decreases the DCJ halving distance by $2$, meaning that both DCJ operations are sorting.

\begin{property}
Given two adjacencies $\aff{a}{b}$ and $\aff{x}{y}$ of $G$, such that 
$I\aff{a}{b}$ and $I\aff{x}{y}$ are overlapping, the successive application of $\DCJ\aff{a}{b}$ and $\DCJ\aff{x}{y}$ decreases the DCJ halving distance by $2$ if and only if  $x \neq \snd{a}$ and $y \neq \snd{b}$.

\end{property}
\begin{proof}
If $x \neq \snd{a}$ and $y \neq \snd{b}$, then the successive application of $\DCJ\aff{a}{b}$ and $\DCJ\aff{x}{y}$ increases the number of cycles in $\NG(G)$ by $2$, by creating two new 2-cycles. Otherwise, $\DCJ\aff{a}{b}$ first creates a new cycle that is then destroyed by  $\DCJ\aff{x}{y}$.\qed
\end{proof}

I denote by $\mathcal{I}(G)$, the set of intervals of all the adjacencies of $G$ that do not contain marker $\circ$.

\begin{remark}
Note that, if $G$ contains $n$ distinct markers, then there are $2n-1$ adjacencies in $G$ that do not contain marker $\circ$,  defining $2n-1$ intervals in  $\mathcal{I}(G)$.
\label{maxEdges}
\end{remark}

\begin{definition}
Two  intervals $I\aff{a}{b}$ and $I\aff{x}{y}$ of $\mathcal{I}(G)$ are said to be \emph{compatible} if they are overlapping and  $x \neq \snd{a}$ and $y \neq \snd{b}$.
\end{definition}

In the following, I prove the BI 1-tandem halving distance formula by showing that if genome $G$ contains more than three distinct markers, $n~>~3$,  then there exist two compatible intervals in $\mathcal{I}(G)$, and if $n=2$ or $n=3$ then $d^t_{BI}(G)=1$ and $2 \leq d^p_{DCJ}(G) \leq 3$.
This means that there exists a \BI ~halving scenario $S$ such that all \BI ~operations in $S$, possibly excluding the last one, are equivalent to two successive sorting DCJ operations.

From now on, until the end of the section,

$\aff{a}{b}$ is an adjacency of $G$ that is not a double-adjacency, $A$ is a genome consisting in a linear chromosome $\mathfrak{L}$ and a circular chromosome $\mathfrak{C}$, obtained by applying the \emph{sorting DCJ}, $\DCJ\aff{a}{b}$, on $G$.

If there exists an interval $I\aff{x}{y}$ in $\mathcal{I}(G)$ compatible with  $I\aff{a}{b}$, then applying $\DCJ\aff{x}{y}$ on $A$ consists in the integration of the circular chromosome $\mathfrak{C}$ into the linear chromosome $\mathfrak{L}$ such that the adjacency $\ass{x}{y}$ is formed.

Such an \emph{integration} can only be performed by cutting an adjacency $\asf{x}{u}$ in $\mathfrak{C}$ and an adjacency $\afs{v}{y}$ in $\mathfrak{L}$ (or inversely) to produce adjacencies $\ass{x}{y}$ and $\aff{v}{u}$.  This means that there must be an adjacency \aff{x}{y} in either $\mathfrak{C}$ or $\mathfrak{L}$ such that $\snd{x}$ is in $\mathfrak{C}$ and $\snd{y}$ in $\mathfrak{L}$ or inversely.
Hence, we have the following property:

\begin{property}
$\mathfrak{C}$ \emph{cannot} be reintegrated into $\mathfrak{L}$ by
applying a sorting DCJ,  $\DCJ\aff{x}{y}$, on $A$ if and only if either:

\begin{itemize} 
\item[(1)] for any adjacency $\aff{x}{y}$ in $\mathfrak{C}$
  (resp. $\mathfrak{L}$), markers $\snd{x}$ and $\snd{y}$ are in
  $\mathfrak{L}$ (resp. $\mathfrak{C}$), %$(1)$, 
  or

\item[(2)] for any adjacency $\aff{x}{y}$ in $\mathfrak{C}$ (resp. $\mathfrak{L}$), markers $\snd{x}$ and $\snd{y}$ are also in $\mathfrak{C}$ (resp. $\mathfrak{L}$). %$(2)$

\end{itemize}

\label{formuleENR}
\end{property}
\begin{proof}
If there exists no  adjacency $\aff{x}{y}$ in $A$ such that $\snd{x}$ is in $\mathfrak{C}$ and $\snd{y}$ in $\mathfrak{L}$ or inversely, then $A$ necessarily satisfies either $(1)$, or $(2)$.\qed 
\end{proof}

\begin{definition}
An interval $I\aff{a}{b}$ in  $\mathcal{I}(G)$ is called
\emph{interval of type 1} (resp. \emph{interval of type 2}) if $\DCJ\aff{a}{b}$ produces a genome $A$ satisfying configuration $(1)$  (resp.  configuration $(2)$) described in Property \ref{formuleENR}.

\end{definition}

For example, in genome  $(\circ~~\fst{2}~~\fst{1}~~\snd{1}~~\fst{3}~~\snd{2}~~\snd{3}~~\circ)$,  $I\asf{1}{3}$ is of type 1 as  $\DCJ\asf{1}{3}$ produces genome 
$(\circ~~\fst{2}~~\fst{1}~~\snd{3}~~\circ) ~ (\snd{1}~~\fst{3}~~\snd{2})$ ;  $I\ass{2}{3}$ is of type 2 as  $\DCJ\ass{2}{3}$ produces genome  $(\circ~~\fst{2}~~\fst{3}~~\snd{2}~~\snd{3}~~\circ) ~ (\fst{1}~~\snd{1})$.

Now we give the maximum numbers of intervals of type 1 and type 2 that can be contained in genome $G$.

\begin{lemma}
The maximum number of intervals of type 1 in $\mathcal{I}(G)$ is 2.
\label{maxType1}
\end{lemma}

\begin{proof}
First, note that there cannot be two intervals $I$ and $J$ of $\mathcal{I}(G)$ 
such that $I \neq J$, and both  $I$ and $J$ are of type 1.
Now, if $I$ is an interval of type 1, there can be at most two different 
adjacencies $\aff{x}{y}$ and $\aff{u}{v}$ such that 
$I\aff{x}{y} = I\aff{u}{v} = I$. In this case $G$ necessarily has a chromosome of the form $(\ldots ~~\snd{x}~~\snd{v}~~\ldots ~~\snd{u}~~\snd{y}~~\ldots)$ or $(\ldots ~~\snd{u}~~\snd{y}~~\ldots ~~\snd{x}~~\snd{v}~~\ldots)$.
Therefore, there are at most two intervals of type 1 in $\mathcal{I}(G)$.\qed

\end{proof}

\begin{lemma}
The maximum number of intervals of type 2 in  $\mathcal{I}(G)$ is $n$.
\label{maxType2}
\end{lemma}

\begin{proof}
\label{proof:proofType2}
First, note that for two adjacencies $\aff{x}{y}$ and $\asf{x}{z}$ in $G$ that 
do not contain marker $\circ$, if $\aff{x}{y}$ is of type 2 then $\asf{x}{z}$ 
cannot be of type 2.
Now, there is only one marker $u$ such that  $\asf{u}{\circ}$ is an adjacency 
of $G$. Let $\aff{u}{v}$ be the adjacency of $G$ having $u$ as first marker, 
then at most half of the intervals in $\mathcal{I}(G) - \{I\aff{u}{v}\}$ can 
be of type 2.
Therefore, there are at most $n$ intervals of type 2 in $\mathcal{I}(G)$.\qed

\end{proof}

\begin{theorem}
If $\NG(G)$ contains $C$ cycles, then the BI 1-tandem halving distance of $G$
    is given by: 
$$d^t_{BI}(G) =  \left \lfloor \frac{n - C}{2} \right \rfloor$$
\label{th:distance}
\end{theorem}

\begin{proof}

    Since there are $2n - 1$ intervals in $\mathcal{I}(G)$, and at
    most $n+2$ are of type $1$ or $2$, then if $G$ contains more than three distinct markers we have $n~>~3$, and since
    $2n-1~>~n+2$ then there exist two compatible intervals in
    $\mathcal{I}(G)$ inducing a BI operation that decreases the DCJ
    distance by $2$.

Next, I show that if $n = 2$ or $n = 3$, then $d^t_{BI}(G)=1$ and $2
\leq d^p_{DCJ}(G) \leq 3$.
 If $n = 2$, then the genome can be written, either as
$(\circ~\fst{a}~\fst{b}~\snd{b}~\snd{a}~ \circ)$, in which case a BI
can swap $\fst{a}$ and $\fst{b}$ to produce a 1-tandem duplicated
genome, or as $(\circ~\fst{a}~\snd{a}~\fst{b}~\snd{b}~\circ)$, in
which case a BI can swap $a$ and $\snd{a}~\fst{b}$ to produce a
1-tandem duplicated genome.

If $n = 3$, then the genome has two double-adjacencies to be
constructed, of the form $\ass{a}{b}$, $\ass{x}{y}$, with $\aff{a}{b}$
and $\aff{x}{y}$ being two adjacencies already present in the genome
such that $\fst{b} = \fst{x}$ or $\fst{b} = \snd{x}$ and $a$ and $y$
are distinct markers. One can rewrite $\aff{a}{b}$ and $\aff{x}{y}$ as
single markers since they will not be splitted, which makes a genome
with 4 markers such that at most 2 are misplaced. Then, a single BI
can produce a 1-tandem duplicated genome.

Now, it is easy to see to see that if  $n = 2$ or $n = 3$, then $d^p_{DCJ}(G) = n- C \leq 3$. Finally, if  $n = 2$ or $n = 3$, then $d^p_{DCJ}(G) \geq 2$, otherwise we would have  $d^p_{DCJ}(G) = 1$ which would imply, as $G$ consists in a single linear chromosome, $d^t_{BI}(G) = 0$.
In conclusion, if $n~>~3$ then there exist two compatible intervals in  $\mathcal{I}(G)$, otherwise if  $n = 2$ or $n = 3$, then $d^t_{BI}(G)=1$ and  $2 \leq d^p_{DCJ}(G) \leq 3$. Therefore $d^t_{BI} = \lfloor \frac{d^p_{DCJ}}{2}\rfloor = \lfloor \frac{n-C}{2}\rfloor$.\qed
\end{proof}

\subsubsection{Sorting algorithm}
\label{sec:scenario}

\noindent
In Section \ref{sec:dist}, I showed that if a genome $G$ contains more than 
three distinct markers after reduction then there exist two compatible 
intervals in $\mathcal{I}(G)$ inducing a BI to perform.
If $G$ contains two or three distinct markers then the BI to perform can be 
trivially computed.
Thus the main concern of this section is to describe an efficient algorithm 
for finding compatible intervals when $n~>~3$.

As in Section \ref{sec:dist}, in the following, $G$ denotes a genome consisting 
of $n$ distinct markers after reduction. 

It is easy to show that the set of intervals 
$\mathcal{I}(G)$ can be built in $O(n)$ time and space complexity.

We now show that finding 2 compatible intervals in  $\mathcal{I}(G)$ can be done in $O(n)$ time and space complexity.

\begin{property}

If $n~>~3$ %and $\mathcal{I}(G)$ contains compatible intervals
, then all the smallest intervals in $\mathcal{I}(G)$ that are not of type 2 admit compatible intervals.
\label{smallestOK}
\end{property}

\begin{proof}

    Let $J$ be a smallest interval that is not of type 2 in
    $\mathcal{I}(G)$. As $J$ is not of type 2, then $J$ has compatible
    intervals if $J$ is not of type 1.

Let us suppose that $J$ is of type 1, then for any adjacency $(a ~ b)$ such 
that markers $a$ and $b$ are not in $J$,  $\snd{a}$ and $\snd{b}$ are in $J$, 
and then $I(a ~ b)$ is strictly included in $J$ and $I(a ~ b)$ can't be of 
type 2. Such adjacency does exist as there are $n~>~3$ markers not included in $J$.

Therefore $J$ is not the smallest, which is a contradiction.\qed

\end{proof}

We are now ready to give the algorithm for sorting a duplicated genome $G$ into a 1-tandem duplicated genome with $\lfloor \frac{n - C}{2} \rfloor$ BI operations.

\begin{algorithm}                      % enter the algorithm environment
\caption{Reconstruction of a 1-tandem duplicated genome}          % give the algorithm a caption
\label{alg1}                           % and a label for \ref{} commands later in the document
\begin{algorithmic}[1]                    % enter the algorithmic environment

\WHILE{$G$ contains more than $3$ markers}
\STATE Construct $\mathcal{I}(G)$
\STATE Pick a smallest interval $I(a ~~ b)$ that is not of type 2 in $\mathcal{I}(G)$
\STATE Find an interval $I(x ~~ y)$ in $\mathcal{I}(G)$ compatible with $I(a ~~ b)$
\STATE Perform the \BI ~ equivalent to $\DCJ(a ~~ b)$ followed by $\DCJ(x ~~ y)$
\STATE Reduce $G$
\ENDWHILE
\IF{$G$ contains $2$ or $3$ markers}
\STATE Find the last BI operation and perform it
\ENDIF
\end{algorithmic}
\end{algorithm}

\begin{theorem}

Algorithm \ref{alg1} reconstructs a 1-tandem duplicated genome with a BI scenario of length $\lfloor \frac{n - C}{2} \rfloor$ in $O(n^2)$ time and space complexity, by computing pairs of sorting DCJ operations.
\end{theorem}

\begin{proof}

Building $\mathcal{I}(G)$ and finding two compatible intervals can be done in $O(n)$ time and space complexity. It follows that the while loop in the algorithm can be computed in $O(n^2)$ time and space complexity.

Finding and performing the last \BI~operation when  $2\leq n\leq 3$ can be done in constant time and space complexity.

Moreover, all \BI~operations, possibly excluding the last one, are computed as 
pairs of compatible intervals, ie. pairs of sorting DCJ operations, which ensures that the length of the scenario 
is $\lfloor \frac{n - C}{2} \rfloor$. \qed
\end{proof}

\begin{corollary}
Any BI scenario computed by Algorithm \ref{alg1} is also a most parsimonious DCJ scenario, twice as long since a BI is equivalent to 2 DCJ.
\end{corollary}

An example of scenario is shown in figure \ref{fig:scenarEx}.

\begin{figure}
\scriptsize{$d_{DCJ}=n-EC-\left\lfloor \frac{OP}{2} \right\rfloor = 4$}

~

\begin{center}

\begin{tikzpicture}[node distance=4mm]
\node (G1) {$\overset{~n=5;EC=1;d=4}{(\circ~\snd{1}~\snd{2}~3~\snd{3}~_{\blacktriangle}5~4_{\blacktriangle}~1~2~\snd{4}~\snd{5}~\circ)}$};

\node (G2) [below=of G1] {$\overset{~n=5;EC=2;d=3}{(\circ~\snd{1}~\snd{2}~_{\blacktriangle}3~\snd{3}~1~2~\snd{4}~\snd{5}~\circ)~(5~_{\blacktriangle}4)}$};

\node (G3) [below=of G2] {$\overset{~n=5;EC=3;d=2}{(\circ~_{\blacktriangle}\snd{1}~\snd{2}~4~5~3_{\blacktriangle}~\snd{3}~1~2~\snd{4}~\snd{5}~\circ)}$};

\node (G4) [below=of G3] {$\overset{~n=5;EC=4;d=1}{(\snd{1}~\snd{2}~4~5~_{\blacktriangle}3)~(\circ_{\blacktriangle}~\snd{3}~1~2~\snd{4}~\snd{5}~\circ)}$};
\node (G5) [below=of G4] {$\overset{~n=5;EC=5;d=0}{(\circ~3~\snd{1}~\snd{2}~4~5~\snd{3}~1~2~\snd{4}~\snd{5}~\circ)}$};

\draw[->] (G1) -- node[right] {excision} (G2);
\draw[->] (G2) -- node[right] {reintegration} (G3);
\draw[->] (G3) -- node[right] {excision} (G4);
\draw[->] (G4) -- node[right] {reintegration} (G5);

\node (n1) [right=of G1] {};
\node (n2) [right=of G2] {};
\node (n3) [right=of G3] {};
\node (n4) [right=of G4] {};
\node (n5) [right=of G5] {};

 \path (n3) -| node (b2) {} (n2);
    \path (n1) -| node (b1) {} (n2);
 \path (n5) -| node (b4) {} (n4);
    \path (n3) -| node (b3) {} (n4);

\draw[thick,decorate,decoration={brace, amplitude=4pt}] 
        (b1) -- (b2) node[midway, right=3pt]{
= BI

 $(\circ~\snd{1}~\snd{2}~\breakpoint~\mathbf{3~\snd{3}}\breakpoint~5~\breakpoint\mathbf{4} \breakpoint~1~2~\snd{4}~\snd{5}~\circ)$};

    \draw[thick,decorate,decoration={brace, amplitude=4pt}] 
        (b3) -- (b4) node[midway, right=3pt]{
= BI

$(\circ~\breakpoint\mathbf{\snd{1}~\snd{2}~4~5}\breakpoint~\breakpoint\mathbf{3}\breakpoint~\snd{3}~1~2~\snd{4}~\snd{5}~\circ)$};
\end{tikzpicture}
\end{center}
\begin{flushright}
\scriptsize{$ d_{BI} = \left \lfloor \frac{d_{DCJ}}{2} \right \rfloor = 2 $}
\end{flushright}
\caption{A BI scenario computed by algorithm \ref{alg1}.}
\label{fig:scenarEx}
\end{figure}

\subsubsection{Conclusion}
\noindent
By expressing BI operations through the DCJ model I could show that restricting the scope of allowed DCJ operations under the constraint of performing only BI doesn't affect the 1-tandem halving distance.

This also means BI scenarios computed by the algorithm are in fact optimal DCJ scenarios and that it solves the DCJ 1-tandem halving problem for a subclass of genomes, namely genomes for which all markers must have the same orientation (sign) as their paralogs.

We will now see how to solve this problem for DCJ in the general case.

\subsection{DCJ}

\def\EP{{{\textsc{EP}}}}
\def\EC{{{\textsc{EC}}}}
\def\OC{{{\textsc{OC}}}}
\def\OP{{{\textsc{OP}}}}

In this section, because I only consider DCJ operations and no BI this time, I will drop the subscript.

$d^p(G)$ denotes the DCJ genome halving distance, towards a perfectly duplicated genome.

$d^t(G)$ denotes the DCJ 1-tandem halving distance, towards a 1-tandem duplicated genome.

~~

I recall that $d^{p}(G)$ can be computed using a data structure called the
\emph{natural graph}, first introduced in \cite{Mabrouk03}.  $\NG(G)$
is the graph whose vertices are the adjacencies of $G$, and 2 vertices
are connected by an edge iff they share a paralogous \emph{extremity} (see figure
\ref{fig:NG}).

\begin{figure}
\begin{tikzpicture}[scale=1]
    \small
    \node at (9,-1) {$G = (\circ~
      \fst{1}~\mfst{4}~~\fst{5}~~\snd{2}~\mfst{3}~~\fst{2}~~\snd{1}~~\snd{4}~~\snd{5}~~\snd{3}~\circ)$};
    \node at (9,-2) {$|OP| = 1 \qquad |EP| = 0$};
    \node at (9,-3) {$|OC| = 1 \qquad |EC| = 1$};
    
    \scriptsize
    \node[draw,circle] at (0,-1) (l0) {$\circ~~1$};
    \node[draw,circle] at (1.5,-1) (l1) {$2~~\snd{1}$};
    \node[draw,circle] at (3,-1) (l2) {$\snd{2}\mfst{3}$};
    \node[draw,circle] at (4.5,-1) (l3) {$\snd{3}~~\circ$};
    \draw (l0) -- node[above] {$\cdot 1$} (l1) -- node[above]
    {$2\cdot$} (l2) -- node[above] {$\cdot 3$} (l3);

    \node[draw,circle] at (0,-2) (ca0) {$\fst{5}~~\snd{2}$};
    \node[draw,circle] at (2,-2) (ca1) {$\mfst{3}~\fst{2}$};
    \node[draw,circle] at (1,-3.5) (ca2) {$\snd{5}~~\snd{3}$};
    \draw (ca0) -- node[above] {$\cdot 2$} (ca1) -- node[right]
    {$3\cdot$} (ca2) -- node[left] {$5\cdot$} (ca0);

    \node[draw,circle] at (4,-2) (cb0) {$\fst{1}\mfst{4}$};
    \node[draw,circle] at (6,-2) (cb1) {$\snd{1}~~\snd{4}$};
    \node[draw,circle] at (6,-3.5) (cb2) {$\mfst{4}~\fst{5}$};
    \node[draw,circle] at (4,-3.5) (cb3) {$\snd{4}~~\snd{5}$};
    \draw (cb0) -- node[above] {$1\cdot$} (cb1) -- node [right]
    {$\cdot 4$} (cb2) -- node[below] {$\cdot 5$} (cb3) -- node [left]
    {$4\cdot$} (cb0);

\end{tikzpicture}
\caption{The natural graph of $G$ and the number of odd
  and even paths and cycles.}
\label{fig:NG} 
\end{figure}
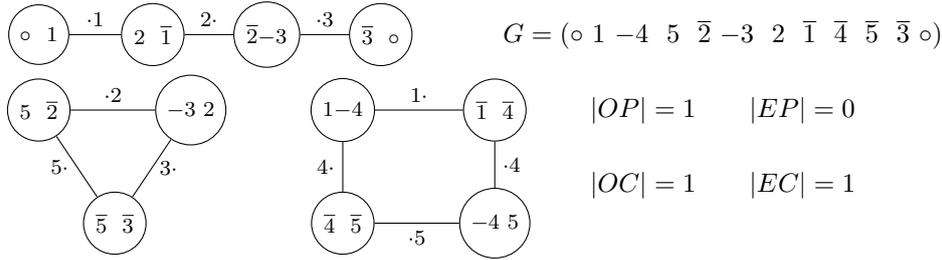
As an adjacency concerns a maximum of 2 markers extremities, this graph has
a maximum degree of 2. Thus, it is composed of paths and cycles only.

Moreover, it consists of nothing but 2-cycles and 1-paths if and only
if $G$ is a perfectly duplicated genome (a $k$-cycle or $k$-path is a
cycle or path containing $k$ edges).

Using this graph, Mixtacki gave the following distance formula:
\begin{theorem}[\cite{Mixtacki08}]
Let $G$ be a totally duplicated genome whose natural graph contains $|\EC|$ even
cycles and $|\OP|$ odd paths.
Then  $d^p(G) = n - |\EC| - \left\lfloor \frac{\left|\OP\right|}{2} \right\rfloor$.
\label{th:dp}
\end{theorem}
Unlike the genome halving problem, the aim of the 1-tandem halving
problem is to find a 1-tandem duplicated genome.

This induces one double-adjacency not to be reconstructed, which is
inelegant to deal with.

We will conveniently get rid of this concern.

From property \ref{prop:dtdp}, a 1-tandem genome that has been
circularized is a perfectly duplicated genome and conversely.

This allows us to establish a property that will reduce the 1-tandem
halving problem to a constraint on genome halving.

\begin{lemma}
\label{lem:unicirc=tandem}
Let $G$ be a unilinear genome. Let $G_c$ be the unicircular genome obtained by circularizing $G$.
Then for any scenario that transforms $G$ into a 1-tandem duplicated genome, there exists an equivalent scenario (of same length) transforming $G_c$ into a unicircular perfectly duplicated genome, and vice versa.
\end{lemma}

\begin{proof}
As $G$ and $G_c$ present the same breakpoints, the scenario conversion is straightforward. It suffices to apply the same DCJ on the same breakpoints.\qed
\end{proof}

Thus, in the rest of this section, the focus will be on reconstructing
an optimal perfectly duplicated genome such that it is unichromosomal.
This is essentially a shape constraint on the genome halving
solutions.

I will follow an approach a bit similar\footnote{I had to
  expand on the ideas and add my own to make it a more complete system, due to the nature of the single tandem halving
  problem.} to what has been done in \cite{Kovac}, as
they enforced another shape constraint on optimal perfectly duplicated
genome configurations.

It consists in taking any optimal configuration then applying a number of
successive transformations (which I will call 
\emph{shapeshifting}) on it, such that they
preserve the distance, and that the optimal configuration converges
towards the desired shape.

In the following sections $G$ will denote a totally duplicated
genome, %consisting of $n$ markers
and $G_c$ its circularized version. $H$ will be an optimal perfectly
duplicated genome for $G_c$.

Following theorem \ref{th:dp}, one can observe that circularization can alter the halving distance, depending on whether the path of $\NG(G)$ is even or odd.%, as circularizing transforms the path into a cycle. 

\begin{property}
\label{prop:dpgc}
If $G$ is a genome such that $\NG(G)$ contains an even path, $d^p(G_c) = d^p(G) - 1$. Else, $d^p(G_c) = d^p(G)$. 
\end{property}

From Mixtacki's formula (Theorem \ref{th:dp}), we know that optimal halving scenarios on circular genomes are scenarios which increase the number of even cycles at each step.
There are two ways of increasing it. Either by splitting a cycle (\textit{i.e.} extracting an even cycle from any cycle), or by merging two odd cycles.

As it can be quite complex at first sight, the shapeshifting system will first be described on a restricted class of genomes, namely those whose natural graph contains only even cycles. This way, we ensure that optimal halving scenarios consist only in cycle extractions.
This special system will then be easily generalized to all genomes by considering cycle-merging operations.

\subsubsection{Special shapeshifting}

Here we consider that $\NG(G_c)$ has only even cycles.
It follows that $\NG(G)$ has an even path and $d^p(G_c) = d^p(G) - 1$.

\paragraph{Anatomy of a multicircular perfectly duplicated genome.}

$H$ is an optimal perfectly duplicated genome for $G_c$.

Since $G_c$ is unicircular, $\NG(G_c)$ contains nothing but cycles.
Therefore, $H$ consists of circular chromosomes only.

For $H$ to be a perfectly duplicated genome, circular chromosomes can
be of two kinds: doubled chromosomes, which can be reduced to
$(x~\snd{x})$, and single chromosomes, which can be reduced to $(x)$
and have a \emph{paralog chromosome} in $H$, which can be reduced to
$(\snd{x})$. Thus the number of single chromosomes is even.

\paragraph{Shapeshifting.}

Any optimal perfectly duplicated genome $H$ induces a class
$\mathcal{C}_H$ of optimal halving scenarios (the class of all optimal
DCJ scenarios transforming $G_c$ into $H$). By observing the structure
of $G_c$ and $H$, we will look for small changes to apply to
$\mathcal{C}_H$, along two criteria: $H$ must converge toward the
desired shape, and it must preserve its optimality.  Such small
changes are called \emph{shapeshifters}.

In our case, we want to end up with the least number of chromosomes in
$H$ (ideally only one), therefore we will look for ways to merge
chromosomes while preserving optimality. This leads us to the
following definition:

\begin{definition}
A shapeshifter is an adjacency \aff{x}{y} such that $x$ and $y$ belong to \emph{different chromosomes} of $H$ (convergence toward the desired shape), and such that \aff{x}{y} (and therefore \ass{x}{y} as well) can be reconstructed by an optimal halving scenario (preservation of optimality).
\end{definition}

For example, if $H$ contains markers $x$ and $y$ in different
chromosomes, $C_x$ and $C_y$, and if \aff{x}{y} can be reconstructed
by an optimal halving scenario, then such scenario induces a new shape
for $H$ such that $C_x$ and $C_y$ cannot be distinct chromosomes
anymore.

As for now we consider genomes whose natural graph has even cycles only, shapeshifters are adjacencies reconstructible by extracting even cycles.

\begin{property}
Adjacencies \aff{x}{y} reconstructible by extracting even cycles are those such that there exists, in $\NG(G_c)$, an induced subgraph which is an \emph{even} path, whose endpoints have outgoing edges $x\cdot$ and $\cdot y$.
\end{property}

Indeed, a DCJ reconstructing \aff{x}{y} will cut at the endpoints of such path and transform it into an even cycle.
However, it is not necessary to consider all even paths, so w.l.o.g we shall focus only on 2-paths (ie. adjancencies \aff{x}{y} that are \emph{present in $G_c$}), which correspond to 2-cycles extractions.

For example, $\aff{1}{4}$ in fig. \ref{fig:NG} is a shapeshifter, as the 2-path induced by vertices $(\fst{1}~\mfst{4})$, $\ass{1}{4}$, and $\aff{-4}{5}$ meets the requirements.

We may proceed and show how to simply apply a shapeshifter on $\mathcal{C}_H$:

Let \aff{x}{p} be the adjacency containing the extremity $x\cdot$ in
$H$, and \aff{q}{y} the one containing the extremity $\cdot y$,

it suffices to perform on $H$ one DCJ cutting adjacencies \aff{x}{p}
and \aff{q}{y} to reconstruct \aff{x}{y} (and \aff{p}{q}), and the
equivalent DCJ on the paralogs, cutting adjacencies \ass{x}{p} and
\ass{q}{y} to reconstruct \ass{x}{y} (and \ass{p}{q}).

One can easily verify that the resulting genome is still optimal (first DCJ brings $H$ closer to $G_c$, second one reconstructs a perfectly duplicated genome). %(sketch of proof: \aff{x}{p} and \aff{q}{y} were part of an optimal genome so they are optimally attainable, and so is \aff{x}{y}... therefore \aff{q}{p} is too..). 

Now we may proceed and study the shapeshifting induced by these DCJ.

Let \aff{x}{y} be a shapeshifter in $G_c$.

$x$ and $y$ belong to different chromosomes in $H$, so there are only
3 possible cases depending on the types of chromosomes ($C_S$ for
single chromosomes, and $C_D$ for doubled ones) which contain these
markers: 1) {$x \in C_S, y \in C_D$}, 2) {$x,y \in C_D$}, 3) {$x,y \in
  C_S$}. The last one could lead to different shapes. Figure
\ref{fig:shapeshift1} illustrates how the genome shape can be altered,
for each case.

\begin{figure}[htbp]
    \centering
    \begin{tabular}{ccc}
\begin{tikzpicture}[scale=0.5]
    \draw (0,0) arc (0:360:1);
    \draw (-1,0) node {$C_x$};
    \draw (0,0) ++(.2,.25) node {\tiny\fst{p}};
    \draw (0,0) ++(.2,-.25) node {\tiny\fst{x}};
    \draw (0,0) ++ (-.1,0) -- ++(.2,0);

    \draw (3,0) arc (0:360:1);
    \draw (2,0) node {$C_y$};
    \draw (1,0) ++(-.25,.25)  node {\tiny\fst{q}};
    \draw (1,0) ++(-.25,-.25) node {\tiny\fst{y}};
    \draw (1,0) ++ (-.1,0) -- ++(.2,0);
    \draw (3,0) ++(.2,.25)   node {\tiny\snd{y}};
    \draw (3,0) ++(.2,-.25)  node {\tiny\snd{q}};
    \draw (3,0) ++ (-.1,0) -- ++(.2,0);

    \draw (6,0) arc (0:360:1);
    \draw (5,0) node {$\overline{C_x}$};
    \draw (4,0) ++(-.25,.25) node {\tiny\snd{x}};
    \draw (4,0) ++(-.25,-.25) node {\tiny\snd{p}};
    \draw (4,0) ++ (-.1,0) -- ++(.2,0);

    \draw[dotted] (0,0) -- ++(1,0);
    \draw[dotted] (3,0) -- ++(1,0);

    \draw[->] (2,-1.3) -- +(0,-0.5);
    \draw(-1,1.7) node {\tiny \underline{case 1})};

    \draw (0,-3) ++(0,0) arc (10:350:1);    
    \draw (0,-3) ++(3,0) arc (10:170:1);
    \draw (0,-3) ++(3,-0.3) arc (-10:-170:1);
    \draw (0,-3) ++(4,-0.3) arc (-170:170:1);
    \draw (0,-3) ++(0,0) -- +(1,0);
    \draw (0,-3) ++(0,0) ++(.2,.25) node {\tiny\fst{p}};
    \draw (0,-3) ++(0,0) ++(.2,-.5) node {\tiny\fst{x}};
    \draw (0,-3) ++(3,0) -- +(1,0);
    \draw (0,-3) ++(1,0) ++(-.25,.25)  node {\tiny\fst{q}};
    \draw (0,-3) ++(1,0) ++(-.25,-.55) node {\tiny\fst{y}};
    \draw (0,-3) ++(3,0) ++(.2,.25)   node {\tiny\snd{y}};
    \draw (0,-3) ++(3,0) ++(.2,-.55)  node {\tiny\snd{q}};
    \draw (0,-3) ++(0,-0.3) -- +(1,0);
    \draw (0,-3) ++(4,0) ++(-.25,.25) node {\tiny\snd{x}};
    \draw (0,-3) ++(4,0) ++(-.25,-.55) node {\tiny\snd{p}};
    \draw (0,-3) ++(3,-0.3) -- +(1,0);
\end{tikzpicture}        
&\begin{tikzpicture}[scale=0.5]
    \draw (0,2) arc (90:270:1);
    \draw[dashed] (0,0) arc (-90:90:1);
    \node at (0,1) {$C_x$};
    \draw (0,2.2) ++(-.2,0) node {\tiny \snd{p}};
    \draw (0,2.2) ++(.2,0)  node {\tiny \snd{x}};
    \draw (0,-.2) ++(-.2,0) node {\tiny \fst{x}};
    \draw (0,-.2) ++(.2,0)  node {\tiny \fst{p}};
    \draw (0,1.9) -- ++(0,.2);
    \draw (0,-.1) -- ++(0,.2);

    \draw (3,0) arc (-90:90:1);
    \draw[dashed] (3,2) arc (90:270:1);
    \node at (3,1) {$C_y$};
    \draw (3,2.2) ++(-.2,0) node {\tiny \snd{y}};
    \draw (3,2.2) ++(.2,0)  node {\tiny \snd{q}};
    \draw (3,-.2) ++(-.2,0) node {\tiny \fst{q}};
    \draw (3,-.2) ++(.2,0)  node {\tiny \fst{y}};
    \draw (3,1.9) -- ++(0,.2);
    \draw (3,-.1) -- ++(0,.2);

    \draw[dotted] (0,0) -- (3,0);
    \draw[dotted] (0,2) -- (3,2);

    \draw[->] (1.5,-0.3) -- +(0,-0.5);
    \draw(-1.2,2.5) node {\tiny \underline{case 2})};

    \draw (0,-3.5) ++(0,2.1) arc (90:270:1.1);
    \draw[dashed] (0,-3.5) ++(0,0) arc (-90:90:1);
    \draw (0,-3.5) ++(0,2.2) ++(-.2,0.2) node {\tiny \snd{p}};
    \draw (0,-3.5) ++(0,2.2) ++(.2,-.4)  node {\tiny \snd{x}};
    \draw (0,-3.5) ++(0,-.2) ++(-.2,-.2) node {\tiny \fst{x}};
    \draw (0,-3.5) ++(0,-.2) ++(.2,0.4)  node {\tiny \fst{p}};
    \draw (0,-3.5) ++(3,-0.1) arc (-90:90:1.1);
    \draw[dashed] (0,-3.5) ++(3,2) arc (90:270:1);
    \draw (0,-3.5) ++(3,2.2) ++(.2,0.2) node {\tiny \snd{y}};
    \draw (0,-3.5) ++(3,2.2) ++(-.2,-.4)  node {\tiny \snd{q}};
    \draw (0,-3.5) ++(3,-.2) ++(.2,-.2) node {\tiny \fst{q}};
    \draw (0,-3.5) ++(3,-.2) ++(-.2,0.4)  node {\tiny \fst{y}};
    \draw (0,-3.5) ++(0,-0.1) -- +(3,0);
    \draw (0,-3.5) ++(0,2.1) -- +(3,0);
    \draw[dashed] (0,-3.5) ++(0,0) -- +(3,0);
    \draw[dashed] (0,-3.5) ++(0,2) -- +(3,0);
\end{tikzpicture}
&\begin{tikzpicture}[scale=0.5]
    \draw (0,0) arc (0:360:1);
    \draw (-2,0) ++(-.25,.25) node {\tiny\snd{p}};
    \draw (-2,0) ++(-.25,-.25) node {\tiny\snd{x}};
    \draw (0,0) ++ (-2.1,0) -- ++(.2,0);
    \draw (0,0) ++(.2,.25) node {\tiny\fst{q}};
    \draw (0,0) ++(.2,-.25) node {\tiny\fst{y}};
    \draw (0,0) ++ (-.1,0) -- ++(.2,0);

    \draw (3,0) arc (0:360:1);
    \draw (1,0) ++(-.25,.25)  node {\tiny\fst{x}};
    \draw (1,0) ++(-.25,-.25) node {\tiny\fst{p}};
    \draw (1,0) ++ (-.1,0) -- ++(.2,0);
    \draw (3,0) ++(.2,.25)  node {\tiny\snd{y}};
    \draw (3,0) ++(.2,-.25) node {\tiny\snd{q}};
    \draw (3,0) ++ (-.1,0) -- ++(.2,0);

    \draw[dotted] (0,0) -- ++(1,0);
    \draw[dotted] (-2,0) .. controls ++(-1,1.2) .. ++(3,1.2) .. controls
    ++(3,0) .. (3,0);

    \draw[->] (0.5,-0.5) -- +(0,-0.5);
    \draw(-1,2) node {\tiny \underline{case 3b and c})};

    \draw (0,-3) ++(0,0) ++(0,0.2) arc (0:180:1);
    \draw (0,-3) ++(3,0) ++(0,0.2) arc (0:180:1);
    \draw (0,-3) ++(0,0.2) -- +(1,0);
    \draw (0,-3) ++(-2,0.2) .. controls ++(.2,1.5) .. ++(2.5,1.5) .. controls    ++(2.3,0) .. ++(2.5,-1.5); %arc (0:-180:2.5);
    \draw (0,-3) ++(-2,0) ++(-.25,.25) node {\tiny\snd{p}};
    \draw (0,-3) ++(-2,0) ++(-.25,-.25) node {\tiny\snd{x}};
    \draw (0,-3) ++(0,0) ++(-.25,.25) node {\tiny\fst{q}};
    \draw (0,-3) ++(0,0) ++(-.25,-.25) node {\tiny\fst{y}};
    \draw (0,-3) ++(0,0) ++(0,-.2) arc (0:-180:1);
    \draw (0,-3) ++(3,0) ++(0,-.2) arc (0:-180:1);
    \draw (0,-3) ++(0,-0.2) -- +(1,0);
    \draw (0,-3) ++(-2,-0.2) .. controls ++(.2,-1.5) .. ++(2.5,-1.5) .. controls    ++(2.3,0) .. ++(2.5,1.5); %arc (0:-180:2.5);
    \draw (0,-3) ++(1,0) ++(.2,.25)  node {\tiny\fst{x}};
    \draw (0,-3) ++(1,0) ++(.2,-.25) node {\tiny\fst{p}};
    \draw (0,-3) ++(3,0) ++(.2,.25)  node {\tiny\snd{y}};
    \draw (0,-3) ++(3,0) ++(.2,-.25) node {\tiny\snd{q}};
\end{tikzpicture}
\end{tabular}
\begin{tikzpicture}[scale=0.5]
    \draw (0,0) arc (0:360:1);
    \draw (-1,0) node {$C_x$};
    \draw (0,0) ++(.2,.25) node {\tiny\fst{p}};
    \draw (0,0) ++(.2,-.25) node {\tiny\fst{x}};
    \draw (0,0) ++ (-.1,0) -- ++(.2,0);

    \draw (3,0) arc (0:360:1);
    \draw (2,0) node {$C_y$};
    \draw (1,0) ++(-.25,.25)  node {\tiny\fst{q}};
    \draw (1,0) ++(-.25,-.25) node {\tiny\fst{y}};
    \draw (1,0) ++ (-.1,0) -- ++(.2,0);

    \draw[dotted] (0,0) -- ++(1,0);

    \draw (6,0) arc (0:360:1);
    \draw (5,0) node {$\overline{C_x}$};
    \draw (6,0) ++(.2,.25) node {\tiny\snd{p}};
    \draw (6,0) ++(.2,-.25) node {\tiny\snd{x}};
    \draw (6,0) ++ (-.1,0) -- ++(.2,0);

    \draw (9,0) arc (0:360:1);
    \draw (8,0) node {$\overline{C_y}$};
    \draw (7,0) ++(-.25,.25)  node {\tiny\snd{q}};
    \draw (7,0) ++(-.25,-.25) node {\tiny\snd{y}};
    \draw (7,0) ++ (-.1,0) -- ++(.2,0);

    \draw[dotted] (6,0) -- ++(1,0);

\draw[->] (9.3,0) -- +(0.5,0);
\draw(-1,1.7) node {\tiny \underline{case 3a})};

    \draw (12,0.2) ++(0,0) arc (10:350:1);
    \draw (12,0.2) ++(1,0) arc (170:-170:1);
    \draw (12,0.2) ++(0,0) -- +(1,0);
    \draw (12,0.2) ++(0,0) ++(.2,.25) node {\tiny\fst{p}};
    \draw (12,0.2) ++(0,0) ++(.2,-.5) node {\tiny\fst{x}};
    \draw (12,0.2) ++(0,-0.3) -- +(1,0);
    \draw (12,0.2) ++(1,0) ++(-.25,.25)  node {\tiny\fst{q}};
    \draw (12,0.2) ++(1,0) ++(-.25,-.55) node {\tiny\fst{y}};
    \draw (12,0.2) ++(6,0) arc (10:350:1);
    \draw (12,0.2) ++(7,0) arc (170:-170:1);
    \draw (12,0.2) ++(6,0) -- +(1,0);
    \draw (12,0.2) ++(6,-0.3) -- +(1,0);
    \draw (12,0.2) ++(6,0) ++(.2,.25)   node {\tiny\snd{p}};
    \draw (12,0.2) ++(6,0) ++(.2,-.55)  node {\tiny\snd{x}};
    \draw (12,0.2) ++(7,0) ++(-.25,.25) node {\tiny\snd{q}};
    \draw (12,0.2) ++(7,0) ++(-.25,-.55) node {\tiny\snd{y}};
\end{tikzpicture}
    \caption{The different shapes that can be obtained by applying a shapeshifter.}
    \label{fig:shapeshift1}
\end{figure}
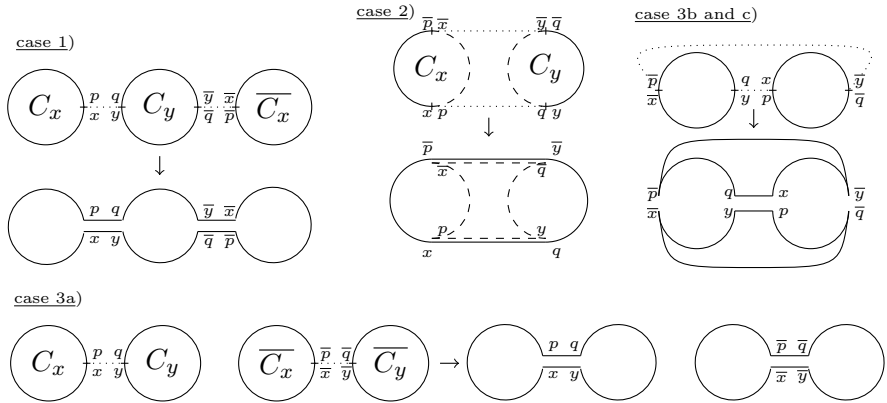

More formally, one can represent shapeshifting as a system of rewriting rules:
\begin{center}
    \begin{tabular}{ll@{$\quad$}ll@{$\quad$}ll}
        1) & {$2\times C_S + C_D \rightarrow C_D$}   & 3.a) & {$4\times C_S
          \rightarrow 2\times C_S$} & 3.c) & {$2\times C_S \rightarrow 2\times C_S$} \\
        2) & {$2\times C_D \rightarrow 2\times C_S$} & 3.b) & {$2\times C_S
          \rightarrow 2\times C_D$} \\
    \end{tabular}
\end{center}

This is convenient as one can deduce useful properties by looking at these rules, which we are about to do, in order to study limit states of the system.

\begin{property}
\label{prop:decreasing}
 Shapeshifting cannot increase the number of chromosomes.
\end{property}

Thus, any limit-cycle necessarily uses rules that do not change the number of chromosomes. Moreover, using rule 2 would eventually lead to using rule 3.b or 3.c as doubled chromosomes are changed into single chromosomes.

\begin{property}
\label{prop:limitcycle}
Any limit-cycle of the system necessarily uses rule 3.b or 3.c.
\end{property}

\begin{property}
\label{prop:invariant}
  Parity of $|C_D|$ is invariant by shapeshifting.
\end{property}

\begin{property}
\label{prop:steadystate}
 A unicircular genome (ie. one doubled chromosome) is the only steady state of the system.
\end{property}

\begin{lemma}
\label{lem:under3}
  By shapeshifting, the number of chromosomes in $H$ can always be decreased under 3.
\end{lemma}

\begin{proof}
Having 3 chromosomes or more guarantees existence of shapeshifters decreasing their number.
Consider the case where $H$ contains only 2 single chromosomes $C_S$ and $\snd{C_S}$. Label the markers from $G$ by the chromosome which holds them in $H$. Adding new chromosomes necessarily creates shapeshifters between at least one of the new chromosomes and $C_S$ or $\snd{C_S}$. Such shapeshifter decreases the number of chromosomes.\qed
\end{proof}

\begin{lemma}
\label{lem:oddCd}
There exists a unicircular optimal perfectly duplicated genome for $G_c$ if and only if $H$ has an \emph{odd} number of doubled chromosomes.
\end{lemma}
\begin{proof}
Straightforward from lemma \ref{lem:under3} and property \ref{prop:invariant}.\qed
\end{proof}

\begin{lemma}
\label{lem:evenCd}
If $H$ has an \emph{even} number of doubled chromosomes, the minimum number of DCJ operations required to reconstruct a unicircular perfectly duplicated genome is $d^p(G_c)+1$, and it can always be attained.
\end{lemma}
\begin{proof}
From lemma \ref{lem:oddCd}, it is impossible to attain a unicircular genome in $d^p(G_c)$ operations. However, from lemma \ref{lem:under3} and property \ref{prop:limitcycle}, it is then always possible to attain two single chromosomes.

Two single chromosomes can then be transformed into one doubled chromosome by one DCJ.\qed
\end{proof}

In conclusion, special shapeshifting allows to compute the tandem distance of any genome $G$ such that $\NG(G)$ contains only even cycles.
 \begin{theorem}
\label{th:restricted}
Let $G$ be a totally duplicated genome such that $\NG(G)$ contains only even cycles.
Let $G_c$ be its circularized version, and $H$ any optimal perfectly duplicated genome for $G_c$.

$d^t(G) = d^p(G)-1$ if and only if $H$ contains an odd number of doubled chromosomes. Else $d^t(G) = d^p(G)$.
\end{theorem}

\begin{proof} 
Since $\NG(G)$ contains only even cycles, it contains an even path. Therefore from property \ref{prop:dpgc}, $d^p(G_c) = d^p(G)-1$.

From lemma \ref{lem:unicirc=tandem} we have that $d^t(G) = d^p(G_c)$ if and only if there exists a unicircular optimal perfectly duplicated genome.
Theorem then follows from lemmas \ref{lem:oddCd} and \ref{lem:evenCd}.\qed 
\end{proof}

The next step is to generalize the shapeshifting system in order to take all possible genomes into account.

\subsubsection{General shapeshifting}

As usual, $G$ is a totally duplicated genome, $G_c$ its circularized version, and $H$ an optimal perfectly duplicated genome for $G_c$.
I will also keep the same notations related to shapeshifters as in the previous section: \aff{x}{y} is a shapeshifter such that $x$ (resp. $y$) is present in chromosome $C_x$ (resp. $C_y$) of $H$, through adjacency \aff{x}{p} (resp. \aff{q}{y} ).

The difference with special shapeshifting is that, \emph{in addition} to everything covered by special shapeshifting, optimal halving scenarios may now also contain cycle merges. Therefore I have to consider shapeshifters that are adjacencies which can be optimally reconstructed through merges.

\begin{property}
Adjacencies \aff{x}{y} reconstructible by merges are those such that extremities $x\cdot$ and $\cdot y$ are in \emph{two distinct odd cycles} of $\NG(G_c)$.
\end{property}

Corresponding shapeshifters can still allow the same shapeshifting rules depending on the types of $C_x$ and $C_y$. Additionally, it is now possible to have $p = \snd{y}$ and $q = \snd{x}$.

This implies that $C_y = \snd{C_x}$ and induces yet another
degenerated case.

The general shapeshifting set of rule becomes: 
\begin{center}
    \begin{tabular}{ll@{$\quad$}ll@{$\quad$}ll}
        1) & {$2\times C_S + C_D \rightarrow C_D$}   & 3.a) & {$4\times C_S
          \rightarrow 2\times C_S$} & 3.c) & {$2\times C_S \rightarrow 2\times C_S$} \\
        2) & {$2\times C_D \rightarrow 2\times C_S$} & 3.b) & {$2\times C_S
          \rightarrow 2\times C_D$} & \textbf{3.d)} & {$\mathbf{2\times C_S \rightarrow C_D}$} \\
    \end{tabular}
\end{center}
This new rule gives general shapeshifting a very interesting property.

\begin{property}
Rule 3.d changes parity of $C_D$.
\end{property}

\begin{lemma}
\label{lem:parity}
If $\NG(G_c)$ contains odd cycles, and if $H$ is made of two single chromosomes, then rule 3.d can be applied.
\end{lemma}
\begin{proof}
As $\NG(G_c)$ contains odd cycles, there are merges in any optimal scenario from $G_c$ to $H$. Thus, there exists an adjacency \aff{x}{p} in $C_x$ such that the adjacencies concerning extremities $x\cdot$ and $\cdot p$ are in two distinct odd cycles of $\NG(G_c)$. By definition, the adjacency concerning extremity $\cdot \snd{p}$ is in the same cycle as the one concerning $\cdot p$.
Therefore, \afs{x}{p} is a shapeshifter inducing rule 3.d.\qed
\end{proof}
\begin{corollary}
Presence of odd cycles in $\NG(G_c)$ ensures a unicircular optimal
perfectly duplicated genome that can always be reached, as rule 3.d can always adjust the parity of $C_D$ if needed. % through restricted shapeshifting, one can always reduce $H$ to a unicircular genome or two single chromosomes (lemma \ref{lem:under3}). In the latter case, using rule 3.d. will make it unicircular.
\end{corollary}

\begin{theorem}
\label{th:generalized}
Let $G$ be a totally duplicated genome such that $\NG(G)$ contains at least one odd cycle, and $G_c$ its circularized version.

Then $d^t(G) = d^p(G_c)$.
\end{theorem}

\begin{proof} 
    From lemma \ref{lem:unicirc=tandem} we have $d^t(G) =
    d^p(G_c)$ iff there exists a unicircular optimal perfectly
    duplicated genome.  Corollary from lemma \ref{lem:parity} ensures
    that there does.\looseness=-1 \qed
\end{proof}

\subsubsection{Distance}

I may finally state a definite formula for the 1-tandem halving distance, as well as results on computational complexity of this problem, by gathering results from the previous sections.

\begin{theorem}

$d^t(G) = n - |\EC| - |\EP| + f_G$

Where $f_G$ is a parameter that is equal to 1 iff $|C_D|$ is even and $|\OC| = 0$, and is equal to 0 otherwise. $|\EC|$, $|\EP|$ and $|\OC|$ are respectively the number of even cycles, even paths and odd cycles in $\NG(G)$.

\end{theorem}

\begin{proof} 
Straightforward from theorems \ref{th:restricted} and \ref{th:generalized}.\qed 
\end{proof}

\begin{theorem}
 $d^t(G)$ can be computed in linear time.
\end{theorem}
\begin{proof}
$\NG(G)$ can be computed in linear time, as well as an optimal perfectly duplicated genome.\qed
\end{proof}

\begin{theorem}
    Computing a scenario can be done in quadratic time.
\end{theorem}
\begin{proof}
An optimal perfectly duplicated genome can be computed in $O(n)$ time using Mixtacki's algorithm (\cite{Mixtacki08}).
From lemma \ref{lem:under3}, one can reduce $H$ to the minimum number of chromosomes using $O(n)$ shapeshifters. Each shapeshifter can be found in $O(n)$ time, so we have a $O(n^2)$ time shapeshifting algorithm.
An optimal DCJ scenario between $G$ and $H$ can then be computed in $O(n)$ time using Yancopoulos' algorithm (\cite{Yancopoulos05}). Thus the algorithm takes quadratic time on the whole.\qed
\end{proof}

\subsection{Closing words and credits}

I was asked by Jean-Stéphane Varré to work on tandem duplications for my master's degree in 2010. I developed the single tandem halving problem as a starting point.

~~

\textbf{On halving by block interchange}

~~

Aida Ouangraoua taught me Anne Bergeron's method for sorting by DCJ as in \cite{BMS06}.

Aida O. also heavily insisted I use her alternate\footnote{non-working} proof for the distance formula. She then rephrased one of my intermediate proofs (proof \ref{proof:proofType2}).

Jean-Stéphane V. drew some of the figures for the article.

I developed my BI 1-tandem halving study mainly with inspiration from \cite{Mixtacki08}. I will add I have been very admirative of Julia Mixtacki's work, in this paper and in the others, for it always presented results with very elegant proofs and reasonings. Her style has been and remains a major inspiration for me.

~

\textbf{On halving by DCJ}

~~

Jean-Stéphane V. drew some of the figures.

I developed shapeshifting with inspiration from \cite{Kovac}. I would like to thank the authors of that paper as it allowed me to gain a much better insight into the space of genome halving scenarios.

Aida O. did not contribute as she was in Canada during most of the time I have been working on this. She participated in some of the early discussions, before shapeshifting was developed.

Once again I took care of proving every single result from these papers.

\section{Model III: Partial tandem duplication}
\label{sec:partialdup}

This work has been published in \cite{Thomas12}. 

Along with the single tandem halving by DCJ, I studied other tandem models:

I studied a model where only a subset of the markers are duplicated. I could not settle complexity of this problem but provided a heuristic algorithm.

I also designed and studied various extended models were multiple tandem duplications occurred such that markers could be present in more than 2 copies.

I proved NP-hardness of all these variants.

\subsection{Model}

\subsubsection{Considered genomes}

I use duplicated genomes, perfectly duplicated genomes as defined in section \ref{sec:notations2}, and dedoubled genomes as defined in section \ref{sec:dedouble}.

I also introduce a generalization of tandem duplicated genomes, namely \emph{k-tandem duplicated genomes}.

\begin{definition}
A \emph{k-tandem duplicated genome} is a totally duplicated genome which can
be reduced to a unilinear dedoubled genome consisting of $k$ distinct markers.
\end{definition}

For example, the genome
$(\circ~\fst{1}\diamond\fst{2}\diamond\fst{3}~\snd{1}\diamond\snd{2}\diamond\snd{3
}~\fst{4}\diamond\fst{5}~\snd{4}\diamond\snd{5}~\circ )$ is a 2-tandem
duplicated genome that can be reduced to the dedoubled genome 

$(\circ~\fst{6}~\snd{6}~\fst{7}~\snd{7}~\circ )$.

Naturally, this new definition is consistent with the previous definition of a 1-tandem duplicated genome.

\subsubsection{Considered operations}

The considered operation model is the DCJ model.

\subsection{Disrupted Single Tandem Halving}
\label{sec:disrupted}

As we could solve the 1-tandem halving problem, a first direction for generalization will be considering genomes containing both duplicated and non-duplicated markers, as it is in better accordance with real biological data.

This can be seen as a 1-tandem halving problem in which adjacencies between duplicated markers can be broken by presence of non-duplicated ones. In other words, non-duplicated markers \emph{disrupt} the 1-tandem halving.

\begin{definition}
The \emph{disrupted 1-tandem halving problem} is a variant of the 1-tandem halving problem in which the genome contains both duplicated and non-duplicated markers. The duplicated markers have to be regrouped and arranged in tandem. The corresponding distance, the \emph{disrupted 1-tandem halving distance}, is denoted $d^{t'}(G)$.
\end{definition}

\subsection{DCJ}

Although a polynomial solution could not be found, in this section I describe a polynomial approximate algorithm and precise its bounds.

I suspect this problem to be NP-hard due to its relation to k-tandem halving.

\subsubsection{Preliminary analysis.}

Any optimal disrupted 1-tandem halving scenario performs two tasks: it gathers duplicated markers together (gathering phase), and it reorganizes them in a tandem (tandem phase).

\begin{definition}
A \emph{break} is an interval of non-duplicated markers surrounded by duplicated markers.
\end{definition}

From now on, $G$ is a duplicated genome containing $n$ duplicated markers separated by $b$ breaks.

\begin{definition}
A \emph{gathering operation} is a DCJ which reduces the number of breaks in $G$.
\end{definition}

Note that the presence of excisions in the gathering phase may produce a genome consisting of multiple chromosomes.

Excisions and their resulting chromosomes will be categorized depending on whether said chromosomes can be reintegrated at best in their source chromosome while increasing the number of even cycles (\emph{good} excision/chromosome), leaving it unchanged (\emph{neutral}) or decreasing it (\emph{bad}).

As this variation in $|\EC|$ changes the tandem distance, we get the following property.
\begin{property}
Once the gathering phase is over in $G$, the remaining distance is $d^t(G) + C^0 + 2C^-$, with $C^0$ the number of neutral chromosomes and $C^-$ the number of bad ones.
\end{property}

The key to build an optimal disrupted 1-tandem halving scenario is to find a gathering scenario that maximizes the number of even cycles while minimizing the number of neutral and bad excisions.

\subsubsection{Optimizing the gathering scenario.}

A DCJ can decrease the number of breaks by at most 1.%, by merging two paths in $\NG(G)$ or by circularizing one. It requires both breakpoints of said DCJ to be on path endpoints.

\begin{property}
The minimum number of gathering operations is $b$.
\end{property}

Gathering operations are DCJ whose breakpoints are on path endpoints from $\NG(G)$. Breakpoints in two distinct paths will merge them, while breakpoints on the endpoints of a same path will circularize it.

\begin{property}
An optimal gathering operation is one that either merges two odd paths, or circularizes an even path.
\end{property}

I now give the maximum number of even cycles a set of $b$ gathering operations can create.

\begin{lemma}
A shortest gathering scenario can create up to $\left \lfloor \frac{\left|\OP\right|}{2} \right \rfloor + |\EP| - 1$ even cycles.
\end{lemma}
\begin{proof}
sketch of proof: Any even path can be circularized by one DCJ, while any two odd paths can be turned into two even cycles with 2 DCJs. Since $b$ breaks induce $b+1$ paths in $\NG(G)$, the number of gathering operations we can use is $b = |\OP|+|\EP| - 1$.\qed
\end{proof}

\begin{corollary}
\label{cor:lb}
$d^{t'}(G) \geq n - |\EC| -1 + \left \lceil \frac{\left|\OP\right|}{2} \right \rceil$.
\end{corollary}

This is assuming a shortest gathering phase produced no bad nor neutral chromosome, and that we are in the best case for the remaining tandem distance ($d^t(G) = d^p(G) - 1$).

Neutral excisions induce a penalty which is the same as performing a
non-optimal gathering reversal, bad excisions are even
worse. Thus the greedy heuristic will proceed as follows: Look
for an optimal gathering operation which is a reversal or a good
excision. When there is none, perform a non-optimal gathering
reversal. \looseness=-1

Let $C_h(G)$ be the number of even cycles produced by the heuristic, then we obtain the following upperbound: $d^{t'}(G) \leq n-|\EC|+|\OP|+|\EP|-1 - C_h(G)$.

In the worst case, $C_h(G)$ can be equal to 0, however, the algorithm
seems to perform pretty well on random genomes, giving values close to the lowerbound.

\subsection{Beyond duplications: Multiple tandem halving}
\label{sec:mth}

Unlike 1-tandem halving, k-tandem halving can be defined in various
ways (is the content of each tandem fixed or only the number? Is the order constrained? etc...). I explored various cases, each time describing a more constrained model:

\begin{itemize}
\item Fixing the number of tandem to be reconstructed ($k$), the problem is NP-hard. 

\item Fixing the markers to be contained in each of the $k$ tandem, the problem remains NP-hard.

\item Fixing the order in which the tandem appear in the ancestral genome, still NP-hard. 

\item Lastly, a \emph{signed} version where the relative orientation of the
tandems is fixed is also NP-hard. 
\end{itemize}

It is unfortunate, as multiple tandems are more relevant from a biological point of view.

Detailed study and proofs follow.

\subsubsection{Genome Dedoubling}

As k-tandem duplicated genomes can be reduced to dedoubled
genomes, I will restate the \emph{genome dedoubling problem} (already studied in section \ref{sec:dedoubling}).

\begin{definition}
Given a rearranged duplicated genome $G$ composed of a single chromosome, the
\emph{genome dedoubling problem} consists in finding a dedoubled genome $H$ such
that the distance between $G$ and $H$ is minimal.
\end{definition}

I recall the general working of an optimal genome dedoubling algorithm, using $\DA(G)$ (refer to \ref{sec:dedoubling} for definition of $\DA(G)$ as well as detailed proofs):

\begin{enumerate}
\item{Pick a maximum number of pairwise disjoint cycles in $\DA(G)$.}

\item{Split them all into 1-cycles.}

\item{Extract 1-cycles concerning other markers in any way until you obtain at least $n$ disjoint 1-cycles.}

\item{\textit{(unilinear variant only)} merge all remaining cycles with the path of $\DA(G)$.}

\end{enumerate}

I remind the reader the genome dedoubling problem is NP-hard, since picking a maximum number of pairwise disjoint cycles in $\DA(G)$ is NP-hard. 

Naturally the unilinear variant is NP-hard as well.

I state a similar result for a small variation on this problem as it will prove useful later.

\begin{definition}
A \emph{loosely dedoubled genome} is a unilinear totally duplicated genome $G$ such that for each marker $x$, either \afs{x}{x}, \afs{-x}{x}, \aff{x}{\msnd{x}} or \aff{-x}{\msnd{x}} is an adjacency of $G$.
\end{definition}

Essentially it is a unilinear dedoubled genome in which the sign of each marker is disregarded. It means that for each marker $x$, $\DA(G)$ either has one 1-cycle for $x$ and one edge for $x$ in the path, or 2 consecutive edges for $x$ in the path.

\begin{definition}
The \emph{loose dedoubling problem} is a variant of the genome dedoubling problem where the aim is a loosely dedoubled genome. 
\end{definition}

\begin{theorem}
The loose genome dedoubling problem is NP-hard.
\end{theorem}
\begin{proof}
The loose variant allows one to avoid having to extract 1-cycles from the path when it presents consecutive edges for a same marker. However, in order to attain the minimum number of operation, it is still required to minimize the number of cycles to be merged with the path. In other words, one still has to pick a maximum number of pairwise disjoint cycles in $\DA(G)$.\qed
\end{proof}

We may now proceed and study k-tandem halving problems.

\subsubsection{Fixed tandem number}

Here we just aim at reconstructing k tandems, regardless of their respective marker contents.

\begin{definition}
    Let $G$ be a totally duplicated genome consisting of $n$ distinct markers, let $0 < k \leq n$ be an
    integer. The \emph{$k$-tandem halving} problem consists in finding
    a $k$-tandem duplicated genome $H$ such that the distance between
    $G$ and $H$ is minimal.
\end{definition}

\begin{theorem}
The $k$-tandem halving problem is NP-hard.
\end{theorem}

\begin{proof}
Genome Dedoubling problem is the particular case of k-tandem halving where $k$ = $n$. \qed
\end{proof}

\subsubsection{Fixed tandem content}

The goal is now to reconstruct k tandems whose respective marker contents are given.

\begin{definition}
    Let $G$ be a totally duplicated genome, consisting of $n$ distinct markers, let $P = \lbrace P_1, P_2, ... , P_k \rbrace $ be a partition of the set of distinct markers.
The \emph{$k$-fixed-tandem halving} problem consists in finding a $k$-tandem duplicated genome $H$ such that each
    tandem is made of the markers of a $P_i$ set, and such that the distance between $G$ and $H$ is minimal.
\end{definition}

\begin{theorem}
The $k$-fixed-tandem halving problem is NP-hard.
\end{theorem}

\begin{proof}
    Genome Dedoubling problem is the particular case of
    $k$-fixed-tandem problem where P is a set of singleton sets. \qed
\end{proof}

\subsubsection{Fixed tandem content and fixed tandem order}

I will now constrain, additionally to the tandems content, the order in which
the tandems are appearing in the final configuration.

\begin{definition}
    Let $G$ be a totally duplicated genome, consisting of $n$ distinct markers, let $P = \lbrace P_1, P_2, ... , P_k \rbrace $ be a partition of the set of distinct markers.
The \emph{$k$-ordered-tandem halving} problem consists in finding a $k$-tandem duplicated genome $H$ such that consecutive tandems are made of the markers of consecutives $P_i$ sets, and such that the distance between $G$ and $H$ is minimal.
\end{definition}

This is a very strong contraint, however the problem is
still NP-hard.
Let's first consider the genome dedoubling variant of this problem (ie. the case where P is a set of singleton sets).

\begin{theorem}
Ordered genome dedoubling problem is NP-hard.
\end{theorem}
\begin{proof}
Constraining the markers order in a dedoubled genome is a constraint on the path of $\DA(G)$. Thus, the choice of pairwise disjoint cycles remains.\qed
\end{proof}
\begin{corollary}
The $k$-ordered-tandem halving problem is NP-hard.
\end{corollary}

\subsubsection{Signed k-tandem halving} %{Fixed tandem orientation}

I will now enforce a constraint which makes genome dedoubling polynomial, and see if it can lead to a polynomial k-tandem halving problem.

\begin{definition}
The \emph{signed dedoubling problem} is a variant of the genome dedoubling problem where the sign of each doublet (ie. \afs{x}{x} or \aff{-x}{\msnd{x}}) is fixed. 
\end{definition}

\begin{lemma}
The signed dedoubling problem is polynomial.
\end{lemma}
\begin{proof}
There is no more possible choice of pairwise disjoint cycles. Indeed, the sign constraint enforces a particular edge (and thus a particular cycle) to be picked.\qed
\end{proof}

I will now conduct a deeper analysis of the signed k-tandem halving problem.

\paragraph*{Genome defragmentation}

 Similarly to the disrupted 1-tandem-halving problem, marker subsets have to be grouped during an optimal scenario. The main difference is that there are several groups to be reconstructed, disrupting each other. Thus, \emph{defragmentation} seems to be a more appropriate term.

\begin{definition}
A \emph{fragment} is an interval of markers from a same group, surrounded by markers from others groups \emph{or telomeres}.
\end{definition}

\begin{definition}
A \emph{defragmentation operation} is a DCJ which reduces the number of fragments in $G$.
\end{definition}

\begin{lemma}
\label{lem:defrag}
Computing the minimum number of defragmentation operations is NP-hard.
\end{lemma}
\begin{proof}
Any loose dedoubling problem instance can be seen as a defragmentation problem under the constraint that each group is split in no more than 2 fragments (one marker stands for a fragment in a genome).\qed% Cases with more than 2 fragments are obviously NP-hard as well (they correspond to variants on the dedoubling problem where markers have more than one paralog).\qed

\end{proof}

\begin{theorem}
Signed k-tandem halving problem is NP-hard.
\end{theorem}

\begin{proof}
This is proven by reduction, from the problem of computing the minimum number of defragmentation operations, to a subclass of signed k-tandem halving.
Consider the class of genomes for which there exists an optimal scenario consisting only of a defragmentation phase. Theorem then follows from lemma \ref{lem:defrag}.\qed
\end{proof}

\subsection{Closing words and credits}

~

\textbf{On disrupted tandem halving}

I recall I originally developed the single tandem halving problem as a starting point for disrupted tandem halving.

Jean-Stéphane V. drew some of the figures for the article.

\textbf{On multiple tandem halving}

I developed, studied and proved all of these variants alone.

Again, Aida O. did not participate as she was in Canada while I worked on this paper.

As usual I took care of proving every single result from this paper. 

\section{Conclusion}

To conclude this section, here is a table containing all of my results.

\begin{center} 
\resizebox{16cm}{!}{
    \begin{tabular}{|c||c|c|} 
       \hline
 Problem & Input genome & Goal configuration \\
        \hline
Genome Dedoubling & Totally duplicated & Dedoubled \\
        1-tandem Halving     &  Totally duplicated  &  1-tandem  \\
k-tandem Halving & Totally duplicated & k tandems \\ 
disrupted 1-tandem Halving & Duplicated & Tandem\\

\hline
    \end{tabular}
}\end{center}

\begin{center} 

\resizebox{16cm}{!}{
    \begin{tabular}{|c|c||c|} 
       \hline
 Problem & Operation model &  Complexity \\
        \hline
Genome Dedoubling & DCJ or Reversal & NP-hard (FPT in the number of cycles) \\
        1-tandem Halving     &  DCJ or Block Interchange   &  $O(n)$ (scenario in $O(n^2)$)\\
k-tandem Halving & DCJ & NP-hard \\ 
disrupted 1-tandem Halving & DCJ &  open \\

\hline
    \end{tabular}
}\end{center}

\part*{General conclusion}

I think my PhD thesis could be summarized by the following sentence:

~

\textit{``Biological reality is NP-hard, polynomial problems come with a catch: they don't make much sense".}

~

The catch being no duplicated content, or an exponential number of different genomes that are all optimal, or other aspects causing awkward moments in bioinformatics conferences when biologists ask questions.

From a computer science and mathematical point of view, however, it's much more interesting.

We've seen that data structures are a powerful tool to shift focus.

On this matter I'll add that any algorithm can be seen as a data structure itself, and this allows a certain freedom in splitting complexity : for example in some cases an  $O(n^4)$ time algorithm could be implemented as a loop of $O(n^2)$ steps each using a structure built in $O(n^2)$, or a loop of $O(n^3)$ steps using a linear structure instead... this kind of reasoning is what allowed me to find the ``magic" property (``the smallest overlapping interval is an optimal operation") for single-tandem halving by block interchange.

I'll conclude by saying that this work on the whole is meant to be seen as a starting point in the study of alternate explanative models for the presence of replicated markers. As said models are taken directly from biological studies, they are obviously not new. However, it is the first time they are studied in the context of classical rearrangement problems.

While most people would notice the diversity in mathematical proofs throughout my papers, I find it is even more stimulating to try to grasp the subtle underlying similarities they must share, since down the line they really are the expression of different analyses for similar problems.

Obvious further perspectives would be the study of other duplication hypotheses, or combination of previously studied ones, even though I feel the mastery of previously studied models and their behavior towards multiple operation models should take priority.

For example, the attempt at solving genome dedoubling by reversal with unoriented genomes (cf. section \ref{unoriented}) was never published and the reason is that the answer we seemed to find (computing an optimal reversal scenario through dynamic programming, in polynomial time) was not a satisfying one. My real goal with this work was to be able to directly compute the orientation cost without the need for building an optimal scenario, just as it was done for classical sorting by reversal in \cite{BMS-04}. I find this kind of result allows a much deeper understanding of the studied problem.

In the same line of thought, I'd consider that even the classical genome halving by reversal, as solved in \cite{Mabrouk98}, doesn't provide a satisfying answer and would deserve further studies.

More generally, I think that while operational research provides very interesting concepts to tackle hard problems, it should be reserved, as intended, to hard problems, be it NP-hard rearrangement problems or software application meant to process very large genomes. In the context of theoretical papers, finding a polynomial answer through such techniques should be seen as an encouragement to find a more elegant method that would provide a better understanding of the problem.

Finally, to give a few words about where I see this field going in a couple years, I would say that unless major changes occur, I do not foresee a bright future in rearrangements for the LIFL, given the type of people working in bioinformatics there.

Even though I wish them success in the name of scientific progress, I don't think it can be done by blatantly trashing ethics, claiming ownership of other people's work by nothing more than writing their names at the highest possible place on the paper, somewhat reminiscent of the way a dog would leave its mark on a fire hydrant.

\bgroup
\small

\egroup

\clearpage

\pagestyle{empty}

\pdfbookmark[0]{Abstract}{Abstract}

\subsection*{Résumé}

La compréhension de la dynamique des réarrangements génomiques est une problématique importante en phylogénie.
La phylogénie est l'étude de l'évolution des espèces. Un but majeur est d'établir les relations d'évolution au sein d'un groupe d'espèces, pour déterminer la topologie de l'arbre d'évolution formé par ce groupe et des ancêtres communs à certains sous-ensembles.

Pour ce faire, il est naturellement très utile de disposer d'un moyen d'évaluer les distances évolutionnaires relatives entre des espèces, ou encore d'être capable d'inférer à un groupe d'espèces le génome d'un ancêtre commun à celles-ci.

Ce travail de thèse, dans la lignée d'autres travaux, consiste à élaborer de tels moyens, ici dans des cas particuliers où les génomes possèdent des gènes en multiples copies, ce qui complique les choses.

Plusieurs hypothèse explicatives de la présence de duplications ont été considérées, des formules de distance ainsi que des algorithmes de calcul de scénarios ont été élaborés, accompagnés de preuves de complexité.

\medskip
\noindent {\bf Mots-clés:} bioinformatique, génomique comparative, réarrangements, marqueurs dupliqués, genome halving, duplication en tandem, breakpoints, inversion, DCJ, échange de blocs

\vfill
\hrule
\vfill

\subsection*{Abstract}

Understanding the dynamics of genome rearrangements is a major issue of phylogenetics.
Phylogenetics is the study of species evolution. A major goal of the field is to establish evolutionary relationships within groups of species, in order to infer the topology of an evolutionary tree formed by this group and common ancestors to some of these species.

In this context, having means to evaluate relative evolutionary distances between species, or to infer common ancestor genomes to a group of species would be of great help.

This work, in the vein of other studies from the past, aims at designing such means, here in the particular case where genomes present multiple occurrencies of genes, which makes things more complex.

Several hypotheses accounting for the presence of duplications were considered. Distances formulae as well as scenario computing algorithms were established, along with their complexity proofs.

\medskip
\noindent {\bf Keywords:} bioinformatics, comparative genomics, rearrangement, replicated markers, genome halving, tandem duplication, breakpoints, reversal, DCJ, block interchange

\end{document}